\documentclass{lmcs}
\pdfoutput=1
\usepackage[utf8]{inputenc}

\usepackage{lastpage}
\lmcsdoi{21}{2}{27}
\lmcsheading{}{\pageref{LastPage}}{}{}%
{Oct.~09,~2023}{Jun.~24,~2025}{}

\usepackage{soul,color}

\usepackage[T1]{fontenc}
\usepackage{amssymb}
\usepackage{enumitem}
\usepackage{xcolor}
\usepackage{hyperref}
\hypersetup{
    colorlinks = false,
}

\usepackage[normalem]{ulem}
\usepackage{amsmath}
\usepackage{stmaryrd}
\usepackage{tikz}
  \usetikzlibrary{arrows,automata,positioning,matrix,calc,petri,backgrounds}

\usepackage{subcaption}
\makeatletter         
\def\figurecaption#1#2{\noindent\hangindent 40pt
                       \hbox to 36pt {\small\sl #1 \hfil}
                       \ignorespaces {\small #2}}

\long\def\@makecaption#1#2{
  \vskip 10pt 
  \settowidth{\@tempdima}{#2}
  \ifdim\@tempdima>0pt
       \setbox\@tempboxa\hbox{#1: #2}
     \else
       \setbox\@tempboxa\hbox{#1 #2}
   \fi
   \ifdim \wd\@tempboxa >\hsize               
       \begin{list}{#1:}{
       \settowidth{\labelwidth}{#1:}
       \setlength{\leftmargin}{\labelwidth}
       \addtolength{\leftmargin}{\labelsep}
        }\item #2 \end{list}\par   
     \else                                    
       \hbox to\hsize{\hfil\box\@tempboxa\hfil}  
   \fi}
\makeatother

\usepackage{extarrows}
\usepackage[boxed]{algorithm2e}

\definecolor{PantoneReflexBlue}{HTML}{003399}
\definecolor{PantoneYellow}{HTML}{FFCC00}

\newcommand{\eustar}{\scalebox{0.1}{\ensuremath{\bigstar}}}

\newcommand{\makestars}{%
  \color{PantoneYellow}%
  \setlength{\unitlength}{1em}
  \divide\unitlength by18
  \begin{picture}(6,6)(-2,3.5)
    \put(6,0){\eustar}
    \put(5.196,3){\eustar}
    \put(3,5.196){\eustar}
    \put(0,6){\eustar}
    \put(-3,5.196){\eustar}
    \put(-5.196,3){\eustar}
    \put(-6,0){\eustar}
    \put(-5.196,-3){\eustar}
    \put(-3,-5.196){\eustar}
    \put(0,-6){\eustar}
    \put(3,-5.196){\eustar}
    \put(5.196,-3){\eustar}
  \end{picture}%
}

\newcommand{\euflag}[2]{%
  {%
    \fboxsep0pt
    \resizebox{#1}{#2}{%
      \colorbox{PantoneReflexBlue}{%
        \vbox to1em{%
          \hsize1.5em
          \parskip0pt
          \parindent0pt
          \centering
          \makestars
        }%
      }%
    }%
  }%
}

\newcommand{\buchiset}{G}

\newcommand{\tpl}[1]{\left\langle{#1}\right\rangle}
\newcommand{\tup}[1]{\tpl{#1}}

\newcommand{\src}{\textsf{src}}
\newcommand{\dst}{\textsf{dst}}

\newcommand{\broadcasts}{\textsf{bcts}}
\newcommand{\Broadcasts}{\Sigma_\broadcasts}
\newcommand{\activeprocs}{\textsf{active}}
\newcommand{\snd}{\textsf{snd}}
\newcommand{\rcv}{\textsf{rcv}}

\newcommand{\brd}{\mathfrak{b}}
\def\Actions{\Sigma_{\textsf{actn}}}
\def\Actionprts{\Sigma_{\textsf{rdz}}}
\def\Comm{\Sigma_{\textsf{com}}}

\newcommand\restr[2]{{
  \left.\kern-\nulldelimiterspace 
  #1 
  \vphantom{\big|} 
  \right|_{#2} 
  }}

\newcommand{\proctemp}{P}
\newcommand{\uwd}{\multimap}
\newcommand{\rwd}{\circledcirc}
\newcommand{\puwd}{P^\multimap}
\newcommand{\sysinst}[1]{{{P}^{#1}}} 
\newcommand{\sysinsttimed}[1]{{{T}^{#1}}}
\newcommand{\psys}{{{P}^\infty}} 
\newcommand{\psystimed}{{{T}^\infty}} 
\newcommand{\puwdsys}{(\puwd)^\infty}

\newcommand{\vas}{{\mathcal V}}

\newcommand{\runs}{\textit{runs}}
\newcommand{\comp}{\textit{comp}}

\newcommand{\proj}[1]{\textit{proj}_{#1}}
\newcommand{\exec}[1][\psys]{\textsc{exec}({#1})}

\newcommand{\execfin}[1][\psys]{\textsc{exec-fin}(#1)}
\newcommand{\execinf}[1][\psys]{\textsc{exec-inf}(#1)}

\newcommand{\execNFW}{{\mathcal A}}
\newcommand{\execBSW}{{\mathcal B}}

\newcommand{\execBSWinit}{\execBSW^{init}}
\newcommand{\execBSWbrd}{\execBSW^{grn}}
\newcommand{\execBSWnobrd}{\execBSW^{loc}}

\newcommand{\grn}{{grn}}
\newcommand{\loc}{{loc}}
\newcommand{\init}{{init}}

\newcommand{\M}{{\mathcal M}}
\newcommand{\A}{{\mathcal A}}

\newcommand{\T}{{\mathcal T}}

\newcommand{\good}{\textsf{light green}\xspace}
\newcommand{\sgood}{\textsf{dark green}\xspace}

\newcommand{\green}{\textsf{green}\xspace}
\newcommand{\locr}{locally-reusable\xspace}

\newcommand{\Pspec}{{\mathcal F}}

\newcommand{\PTIME}{\textsc{ptime}\xspace}
\newcommand{\EXPTIME}{\textsc{exptime}\xspace}
\newcommand{\PSPACE}{\textsc{pspace}\xspace}
\newcommand{\NPSPACE}{\textsc{npspace}\xspace}
\newcommand{\EXPSPACE}{\textsc{expspace}\xspace}

\newcommand\msg[1]{\ensuremath{{#1}}}

\newcommand\head[1]{\smallskip\noindent\textbf{#1.}}

\newcommand{\trans}[3]{#1 \xrightarrow{{#3}} #2}

\newcommand\LTL{{\textsf{LTL}}\xspace}
\newcommand\LTLf{{\textsf{LTLf}}\xspace}

\newcommand\Nat{\mathbb{N}_{> 0}}

\newcommand\Rat{\mathbb{Q}}

\newcommand\RatGEZ{\mathbb{Q}_{\geq 0}}
\newcommand\RatGZ{\mathbb{Q}_{> 0}}

\def\NatZero{\mathbb{N}} 

\newcommand{\clocks}{\ensuremath{C}}

\newcommand{\inc}{\mathsf{inc}}
\newcommand{\reset}{\mathsf{reset}}
\newcommand{\skp}{\mathsf{skip}}
\newcommand{\grd}{\textit{grd}}
\newcommand{\rst}{\textit{rst}}
\newcommand{\AP}{\textit{AP}}
\newcommand{\CP}{\textit{CP}}

\newcommand{\true}{\texttt{true}}
\newcommand{\false}{\texttt{false}}

\newcommand{\simrel}{M}

\newcommand{\loopindices}{\mathcal{I}}
\newcommand{\suc}[2][]{next^{#1}(#2)}
\newcommand{\pre}[1]{prev(#1)}

\DeclareMathOperator{\limp}{\to}
\DeclareMathOperator{\nextX}{\mathsf{X}}

\DeclareMathOperator{\until}{\mathbin{\mathsf{U}}}

\DeclareMathOperator{\always}{\mathsf{G}}

\DeclareMathOperator{\eventually}{\ensuremath{\mathsf{F}}\xspace}

\theoremstyle{plain}
\newtheorem{proposition}[prop]{Proposition}
\newtheorem{theorem}[thm]{Theorem}
\newtheorem{lemma}[lem]{Lemma}
\newtheorem{corollary}[cor]{Corollary}

\theoremstyle{definition}
\newtheorem{remark}[rem]{Remark} 
\newtheorem{example}[exa]{Example} 
\newtheorem{definition}[defi]{Definition}

\begin{document}

\title[PMC of Discrete-Timed Networks and Symmetric-Broadcast Systems]{Parameterized Model-checking of Discrete-Timed Networks and Symmetric-Broadcast Systems}

\author[B.~Aminof]{Benjamin Aminof}[a,d]

\author[S.~Rubin]{Sasha Rubin\lmcsorcid{0000-0002-3948-129X}}[b]

\author[F.~Spegni]{Francesco Spegni\lmcsorcid{0000-0003-3632-3533}}[c]

\author[F.~Zuleger]{Florian Zuleger\lmcsorcid{0000-0003-1468-8398}}[a]

\address{Technical University of Vienna, Austria}
\address{University of Sydney, Australia}
\address{Universit\`a Politecnica delle Marche, Ancona, Italy}
\address{Universit\`a degli Studi di Roma La Sapienza, Italy}


\begin{abstract}
We study the complexity of the model-checking problem for parameterized discrete-timed systems with arbitrarily many anonymous and identical contributors, with and without a distinguished ``controller'', and communicating via asynchronous rendezvous.  Our work extends the seminal work from German and Sistla on untimed systems by adding discrete-time clocks to processes.

For the case without a controller, we show that the systems can be efficiently simulated --- and vice versa --- by systems of untimed processes that communicate via rendezvous and symmetric broadcast, which we call ``RB-systems''. Symmetric broadcast is a novel communication primitive that allows all processes to synchronize at once; however, it does not distinguish between sending and receiving processes.

We show that the parameterized model-checking problem for safety specifications is \PSPACE-complete, and for liveness specifications it is decidable in \EXPTIME. The latter result is proved using automata theory, rational linear programming, and geometric reasoning for solving certain reachability questions in a new variant of vector addition systems called ``vector rendezvous systems''. We believe these proof techniques are of independent interest and will be useful in solving related problems.

For the case with a controller, we show that the parameterized model-checking problems for RB-systems and systems with asymmetric broadcast as a primitive are inter-reducible. This allows us to prove that for discrete timed-networks with a controller the parameterized model-checking problem is undecidable for liveness specifications.

Our work exploits the intimate connection between parameterized discrete-timed systems and systems of processes communicating via broadcast, providing a rare and surprising decidability result for liveness properties of parameterized timed-systems, as well as extend work from untimed systems to timed systems.
\end{abstract}

\maketitle


\section{Introduction}

We systematically study the complexity of the model-checking problem for parameterized discrete-timed systems that communicate via synchronous rendezvous. These systems consist of arbitrarily many anonymous and identical contributors, with and without a distinguished ``controller" process. The parameterized model-checking problem asks whether a given specification holds no matter the number of identical contributors. This is in contrast to traditional model-checking that considers a fixed number of contributors.  

Our model subsumes the classic case of untimed systems~\cite{gs92} --- processes are finite-state programs with (discrete-time) clocks that guard transitions. Timed processes can be used to model more realistic circuits and protocols than untimed processes~\cite{alur99,chevallier09}. 

We study the computational complexity of the parameterized model-checking problem (PMCP) for safety and liveness specifications.\footnote{We follow Esparza et. al.~\cite{efm99} and refer to sets of finite executions as \emph{safety} properties, and sets of infinite executions as \emph{liveness} properties.} Safety properties are specified by formulas of linear-temporal logic over finite traces (\LTLf) and nondeterministic finite word automata (NFW), and liveness properties are specified by formulas of linear-temporal logic (\LTL) and nondeterministic B\"uchi automata (NBW).

We show that without a controller safety is \PSPACE-complete, while liveness is in \EXPTIME;  and with a controller safety is decidable, while liveness is undecidable.  In more detail:
\begin{enumerate}

\item For systems without a controller, we prove that the PMCP for safety specifications is \PSPACE-complete --- in fact, \PSPACE-hardness even holds for a fixed specification (known as program complexity) and for a fixed program (known as specification complexity).

\item For systems without a controller, we prove that the PMCP for liveness specifications can be solved in \EXPTIME.
This is a rare decidability result for liveness properties of any non-trivial model of parameterized timed systems. The algorithms presented make use of interesting and intricate combination of automata theory, rational linear programming, and geometric reasoning for solving certain reachability questions in a new variant of vector addition systems called `vector rendezvous systems'. We believe these techniques are of independent interest and will be useful in solving related problems.

\item For systems with a controller, we prove that the PMCP for liveness specifications is undecidable. This follows from a new reduction between timed-systems with a controller and systems with asymmetric broadcast, and the known undecidability of the PMCP of latter for liveness properties. The novel reduction also allows us to recover the known result that PMCP is decidable for discrete-time systems with a controller and safety specifications.
\end{enumerate}

Although this doesn't completely close the picture (the complexity without a controller for liveness properties is \PSPACE-hard and in \EXPTIME), we remark that the parameterized verification problem for liveness properties is notoriously hard to analyze: apart from a single simple cutoff result which deals with processes communicating using conjunctive and disjunctive guards~\cite{spalazzi2020parameterized,ss14vstte}, no decidability result for liveness specifications of timed systems was known before this work.

To solve the PMCP problem for these systems we introduce \emph{rendezvous-broadcast systems (RB-systems)} --- systems of finite-state processes communicating via \emph{rendezvous} and \emph{symmetric broadcast}. Unlike asymmetric broadcast which can distinguish between the sender and the receivers, with symmetric broadcast there is no designated sender, and thus it can naturally model the passage of discrete time, i.e., every symmetric broadcast can be thought of as a tick of the discrete-time clocks. 

We show that RB-systems and timed-networks with the same number of processes can efficiently simulate each other. Thus, in particular, the PMCP of RB-systems and timed networks are polynomial-time inter-reducible. Furthermore, we show that for the case with a controller, RB-systems (and thus timed-networks) are polynomial-time inter-reducible to systems with asymmetric broadcasts. We remark that this equivalence does not hold for the case without a controller (indeed, we show that without a controller PMCP for liveness specifications is decidable, whereas it is known to be undecidable for systems with asymmetric broadcast~\cite{efm99}).
We thus consider the introduction of the notion of a symmetric broadcast to be an interesting communication primitive in itself. We then study the PMCP for RB-systems and in fact establish the itemized results above for RB-systems.

This work extends the results and provides detailed versions of proofs and statements that already appeared in the conference proceedings version \cite{conf/icalp/AminofRZS15}.

\subsection{Techniques}
The bulk of the work concerns the case without a controller, in which we first establish the decidability of the PMCP for safety properties and prove this problem to be \PSPACE-complete (already for a fixed specification).
The decidability of the PMCP for safety properties of timed networks has already been known~\cite{ADM04}; however, it was obtained using well-structured transition systems and only gives a \emph{non-elementary} upper bound, which we improve to \PSPACE.
We obtain the result for safety properties by constructing a reachability-unwinding of the states of the processes of the parameterized system, where we compute precisely those states the system can be in after exactly $n$ broadcasts;
we show that the reachability-unwinding has a lasso-shape and can be constructed in \PSPACE, which allows us to obtain the upper complexity bound. We provide a matching lower-bound by reducing the termination problem of Boolean programs to non-reachability in RB-systems.

We then prove an \EXPTIME upper-bound for the PMCP for liveness properties. This result is considerably more challenging than the upper-bound for safety properties. One source of difficulty is the need to be able to tell whether some rendezvous transition can be executed a bounded (as opposed to unbounded) number of times between two broadcasts --- a property which is not $\omega$-regular.
In order to deal with this issue we work with B-automata~\cite{b10}, which generalize B\"uchi-automata by equipping them with counters.

A key step in constructing a B-automaton that recognizes the computations of the system that satisfy a liveness specification requires establishing the existence (or lack thereof) of certain cycles in the runs of the parameterized system. Alas, the intricate interaction between broadcasts and rendezvous transitions makes this problem very complicated. In particular, known classical results concerning pairwise rendezvous without broadcast~\cite{gs92} do not extend to our case.
We solve this problem in two steps: we first obtain a precise characterization (in terms of a set of linear equations) for reachability of configurations between two broadcasts; we then use this characterization in an iterative procedure for establishing the existence of cycles with broadcasts.
To obtain the characterization for reachability mentioned above we introduce vector rendezvous systems (VRS) and their continuous relaxation, called continuous vector rendezvous systems (CVRS). These systems are counter abstractions --- which view configurations as vectors of counters that only store the number of processes at every process state, but not the identity of the processes ---
and constitute a new variant of the classical notion of vector addition systems~\cite{journals/tcs/HopcroftP79}.


\subsection{Related Work}

In this work we parameterize by the number of processes (other choices are possible, e.g., by the spatial environment in which mobile-agents move~\cite{DBLP:journals/ai/AminofMRZ22}).  Our model assumes arbitrarily many finite-state anonymous identical processes, possibly with a distinguished controller. Finite-state programs are commonly used to model processes, especially in the parameterized setting~\cite{DBLP:reference/mc/AbdullaST18,PMCPbook2015,clarke2008proving,gs92}. The parameterized model-checking problem (PMCP) has been studied for large-scale distributed systems such as swarm robotics~\cite{LomuscioP22,KouvarosL16}, hardware design~\cite{DBLP:conf/charme/McMillan01}, and multi-threaded programs~\cite{DBLP:journals/toplas/0001KW14}.

The PMCP easily becomes undecidable for many types of communication primitives because of the unbounded number of processes (even for safety properties and untimed processes) \cite{Su88}. Techniques for solving the general case include identifying decidable/tractable subcases (as we do) and designing algorithms that are not guaranteed to terminate, e.g., with acceleration and approximation techniques~\cite{DBLP:reference/mc/AbdullaST18,DBLP:journals/cl/ZuckP04}. The border between decidability and undecidability for various models is surveyed in \cite{PMCPbook2015}, including token-passing systems~\cite{DBLP:conf/cade/AminofR16,AJKR14}, rendezvous and broadcast~\cite{gs92,efm99,aminof2018parameterized}, guarded protocols~\cite{jacobs2018analyzing}, ad hoc networks~\cite{delzanno2010parameterized}. 

The seminal work~\cite{gs92} shows that the PMCP of (untimed) systems that communicate via rendezvous is in \PTIME without a controller, and is \EXPSPACE-complete with a controller. We now compare our work against others that explored the PMCP for timed processes (and timed or untimed specifications).

In \cite{bertrand2013parameterized}, Bertrand and Fournier consider the case of dynamic networks of timed Markov decision processes (TMDPs), a model for agents mixing probabilistic and timed behavior, whose transitions can be guarded by simple conditions on a global real-valued clock variable and through broadcast messages (a broadcaster is distinguished from the receivers) that can force all processes to transition in one computational step. They explore the decidability of several variants of the parameterized probabilistic reachability problem in networks of TMDPs. Interestingly, they observe that in some settings the problem is undecidable if the number of processes is fixed but unknown, while it becomes decidable (but not primitive recursive) by allowing processes to join and leave the network with given probabilities along the executions of the system. Their undecidability results are based on the possibility, using message broadcasting,  of distinguishing a process that acts as controller of the network.

In \cite{aj03} Abdulla and Jonsson prove that model checking safety properties for timed networks with a controller process is decidable, provided that each process has at most one (real-valued) clock variable.  They assume several processes can take synchronous transitions at once by using a rendezvous primitive and a distinguished process can act as controller of the network. The decidability carries over to discrete timed networks with one clock variable per process, and in case all processes are equal, i.e. the controller process is just another copy of the user processes. Our work proves decidability of safety and liveness properties independently from the number of clock variables in the network, provided that clocks range over a discrete domain and there is no controller process in the network. 

In \cite{ADM04} Abdulla et al. extends the decidability result of \cite{aj03} to timed networks with rendezvous and a controller process, assuming each process has any finite number of discrete clocks. They prove decidability by exhibiting a non-elementary complexity upper bound for the problem, and of course this also proves decidability of the same problem for timed network in the  absence of a controller process. In our setting, without a controller process, we are able to prove a much smaller upper bound for the complexity of the model checking problem restricted to safety properties, as well as the decidability of liveness properties. In case there is a controller process, our results show that the PMCP for liveness properties for timed networks is undecidable.

In \cite{abdulla2016parameterized} Abdulla et al. consider the PMCP for a model of timed processes communicating through Ad Hoc Wireless Networks. In this model, an arbitrary number of processes use either real- or discrete- valued clocks and can connect among themselves in topologies that are defined by some given family of graphs such as bounded path graphs or cliques. Given a family of graphs, each timed process can communicate with its direct neighbors using broadcast or rendezvous messages. In the context of timed networks using rendezvous on a clique graph topology, which is the setting closest to our work, they focus on rendezvous communications and rephrase the decidability results already presented in \cite{aj03,ADM04}. 

In \cite{isenberg2017incremental} Isenberg provides techniques for finding invariants for safety properties and timed-networks consisting of processes with continuous clocks, a controller, shared global variables and broadcast communication. In contrast, our decision procedures are based on automata theory.

In \cite{AbdullaACMT18}, Abdulla et al. study the {PMCP for reachability specifications} for Timed Petri Nets. There the authors prove that the problem is PSPACE-complete provided that each process carry only one clock variable. Interestingly, this is the same complexity we prove for the discrete case, thus the extension of the problem to the continuous time setting (under the limitation of one clock variable per process) falls in the same complexity class as the discrete time setting. In this work we make a step further by providing an upper bound to the complexity of the PMCP for liveness specifications in the discrete time setting. 

In \cite{spalazzi2020parameterized,ss14vstte} Spalazzi and Spegni study the parameterized model-checking problem of {Metric Interval Temporal Logic} formulae against networks of conjunctive or disjunctive timed automata of arbitrary size. They prove that in case of timed networks with either all conjunctive or all disjunctive Boolean guards and a controller process, a cutoff exist allowing to reduce controlled timed networks of arbitrary size to timed networks of some known size, provided that process locations have no time invariants forcing progress. This implies that PMCP is decidable under such conditions. In contrast, in our setting time is discrete, there is no controller process, communication is by rendezvous, and specifications are qualitative (i.e. LTL).

Andr\`e et al. \cite{andre2023parameterized} prove that a cutoff for timed networks with all disjunctive Boolean guards does exist in presence of clock invariants in the case of processes with a single clock variable and suitable conditions ensuring that locations appearing in the clock invariants themselves are visited infinitely often.

Finally, we remark that simulations between parameterized systems with different communication primitives, including asymmetric broadcast and rendezvous, is systematically studied in~\cite{DBLP:conf/lpar/AminofRZ15}. We also contribute to that line of work, involving the newly introduced symmetric-broadcast primitive.  Our work establishes for the first time an intimate two-way connection between discrete-timed systems and systems communicating via broadcast: symmetric broadcast when there is no controller, and asymmetric broadcast when there is a controller. We are certain that this intimate connection will prove useful in transferring results between these two types of systems, and discovering new results also if one considers other communication primitives than rendezvous.

\section{Definitions and Preliminaries}\label{sec: definitions}

For the sake of self-consistency, let us now recall various notions from automata theory that will be used along this work.

\paragraph{Notation.}
Let $\Nat$ denote the set of positive integers, let $\NatZero = \Nat \cup \{0\}$, let $\Rat$ denote the set of rational numbers, $\RatGZ$ the set of positive rational numbers, and $\RatGEZ$ the non-negative ones. Let $[n,m]$, for $n<m \in \Nat$, denote the set $\{n,n+1,\dots,m\}$, and if $m=\infty$ then $[n,m] = \{ n,n+1,\ldots \}$. Let $[n]$ denote the set $[1,n]$.
Finally, for $m \le n \in \Nat$, we call $\mu:[m] \to 2^{[n]}$, a {\em partition} of $[n]$ if $i \neq j$ implies $\mu(i) \cap \mu(j) = \emptyset$, and $[n] = \cup_{i \in [m]} \mu(i)$.
{For an \emph{alphabet} $\Sigma$ we denote by $\Sigma^*$ (resp. $\Sigma^\omega$) the set of all finite (resp. infinite) \emph{words} over $\Sigma$.}
The {\em concatenation} of two words $u$ and $w$ is written $uw$ or $u \cdot w$.
The \emph{length} of a word $u$ is denoted by $|u|$, and if $u$ is infinite then we write $|u|=\infty$.

\subsection{Transition Systems} \label{sec:transition systems}
A \emph{labeled transition system (LTS)} is a tuple
\[L = \tup{\AP,\Sigma,S,I,R,\lambda}\]
where
\begin{itemize}
\item $\AP$ is a finite set of \emph{atomic propositions} (also called \emph{state labels}),
\item $\Sigma$ is an alphabet of {\em edge-labels},
\item $S$ is a set of {\em states} (in the following we assume that $S \subseteq \Nat$),
\item  $I \subseteq S$ is a set of {\em initial states},
\item $R \subseteq S \times \Sigma \times S$ is an {\em edge relation},
\item and $\lambda \subseteq S \times {\AP}$ is a \emph{labeling relation} that associates with each state the atomic propositions that hold in it. We will often use functional notation and write $\lambda(s)$ for the set of atoms $p$ such that $(s,p) \in \lambda$.
\end{itemize}
In case all components of $L$ are finite, we say that $L$ is a \emph{finite LTS}; and otherwise we say that it is an \emph{infinite LTS}.
An edge $e = (s,a,s') \in R$, is also called a \emph{transition}, and may be written $\trans{s}{s'}{a}$.
The element $s$ is called the {\em source} (denoted $\src(e)$) of $e$, and $s'$ is called its {\em destination} (denoted $\dst(e)$), and $a$ is called the \emph{label} of $e$.
Given $\sigma \in \Sigma$, and a state $s \in S$, we say that $\sigma$ is {\em enabled} in $s$ if there is some $s' \in S$ such that $\trans{s}{s'}{\sigma}$.
A {\em path} $\pi$ is a (finite or infinite) sequence $e_1 e_2 \dots$ of transitions such that for every $1 \leq i < |\pi|$ we have that $\dst(e_i) = \src(e_{i+1})$, where $|\pi| \in \Nat \cup \{\infty\}$ is the {\em length} of $\pi$.
We extend the notations $\src(\pi)$ and $\dst(\pi)$ to paths (the latter only for finite paths) in the natural way.
Extend $\lambda$ to paths as follows: if $\pi = e_1 e_2 \cdots e_k$ is finite then
$\lambda(\pi) = \lambda(\src(e_1)) \lambda(\src(e_2)) \cdots \lambda(\src(e_{k-1})) \lambda(\src(e_k)) \lambda (\dst(e_k))$, and if $\pi = e_1 e_2 \cdots$ is infinite then
$\lambda(\pi) = \lambda(\src(e_1)) \lambda(\src(e_2)) \cdots$. A \emph{run} is a path whose source is
an initial state. The set of runs of an LTS $L$ is written $\runs(L)$. A state $s \in S$ is {\em reachable} if it is the destination of some run.
The \emph{size} of a finite LTS is defined to be the sum of the number of states and number of transitions.

Let $L = \tpl{\AP,\Sigma,S,I,R,\lambda}$ and
$L'= \tpl{\AP,\Sigma,S',I',R',\lambda'}$ be two LTSs over the same set of atomic propositions $\AP$ and the same set of edge-labels $\Sigma$.
We now define a few notions of equivalence relating such LTSs.
A relation $\simrel \subseteq S \times S'$ is a \emph{simulation} if (i) for every $q \in I$ there is $q' \in I'$ such that $(q,q') \in \simrel$, (ii)
$(q,q') \in \simrel$ implies $\lambda(q) = \lambda'(q')$ and for every $(q,\sigma,r) \in R$ there exists $r'$ with $(q',\sigma,r') \in R'$
such that $(r,r') \in \simrel$. In this case we say that \emph{$L'$ simulates $L$}.
Say that $\simrel$ is a \emph{bisimulation} if $\simrel$ is a simulation and $\{(q',q) : (q,q') \in \simrel\} \subseteq Q' \times Q$ is a simulation.
We say that runs $\pi, \pi'$, of $L$ and $L'$ respectively, of the same length are \emph{equi-labeled} if for every $i < |\pi|$, if $\pi_i = (s,\sigma,t)$ and $\pi'_i = (s',\sigma',t')$ we have that $\lambda(s) = \lambda'(s')$,
$\lambda(t) = \lambda'(t')$, and $\sigma = \sigma'$. It follows immediately from the definitions that if $L'$ simulates $L$ then for every run in $L$ there exists an equi-labeled run in $L'$.

We will use the following operations:
Let $\AP$ be a set of atomic propositions.
Given a proposition $a \in \AP$ and a (finite or infinite) sequence $\xi \in (2^\AP)^* \cup (2^\AP)^\omega$, we denote by $(\xi)_a$ the subsequence of $\xi$ that consists of all sets that contain $a$.
Given a subset $\AP' \subseteq \AP$ and a (finite or infinite) sequence $\xi \in (2^\AP)^* \cup (2^\AP)^\omega$, we denote by $\xi|_{\AP'}$ the sequence that we obtain from $\xi$ by intersecting every set with $\AP'$.

\subsection{Automata} We will use nondeterministic automata with three types of acceptance conditions, i.e., ordinary reachability acceptance (on finite input words), B\"uchi acceptance (on infinite words), and a boundedness condition on a single counter (on infinite words).
Since automata are like LTSs (except that they include an acceptance condition, and exclude the state labeling function), we will use LTS terminology and notation that is independent of the labeling, e.g., source, destination, path and run. We remark that inputs to the automata will be edge-labeling of paths in certain LTSs, and thus the input alphabet for automata is also denoted $\Sigma$.

A {\em nondeterministic finite word automaton (NFW)} is a tuple
\[
\A = \tup{\Sigma,S,I,R,F}
\]
where
\begin{itemize}
\item $\Sigma$ is the \emph{input alphabet},
\item $S$ is the finite set of \emph{states},
\item $I \subseteq S$ are the \emph{initial states},
\item $R \subseteq S \times \Sigma \times S$ is the \emph{transition relation}, and
\item $F \subseteq S$ are the \emph{final states}.
\end{itemize}
Given a finite word $\alpha = \alpha_1 \alpha_2 \cdots \alpha_k$ over the alphabet $\Sigma$, we say that $\rho = \rho_1 \rho_2 \dots \rho_k$ is
a \emph{run of $\A$ over $\alpha$} if, for all $i \in [k]$, the label of the transition $\rho_i$ is $\alpha_i$. The run $\rho$ is \emph{accepting} if $\dst(\rho_{k}) \in F$. A word  is \emph{accepted} by $\A$ if there is an accepting run of $\A$ over it. The \emph{language} of $\A$ is the set of words that it accepts.

A {\em nondeterministic B\"uchi word automaton (NBW)} is a tuple $\A = \tup{\Sigma,S,I,R,\buchiset}$,
which is like an NFW except that $F$ is replaced by a \emph{B\"uchi set} $\buchiset$.
Unlike NFW which run over finite words, an NBW runs over infinite words. Hence, given an infinite word $\alpha = \alpha_1 \alpha_2 \cdots$ over the alphabet $\Sigma$, we say that $\rho = \rho_1 \rho_2 \dots$ is a \emph{run of $\A$ over $\alpha$} if, for all $i \in \Nat$, the label of the transition $\rho_i$ is $\alpha_i$.
The run $\rho$ induces a set $inf(\rho)$ consisting of those states $q \in S$ such that $q = \src(\rho_i)$ for infinitely many $i$. The run $\rho$ is \emph{accepting} if $inf(\rho) \cap \buchiset \neq \emptyset$. The definition when a word is \emph{accepted}, and of the \emph{language} of $\A$, are as for NFW.

An \emph{NBW with one counter}, or \emph{B-automaton} for short, is a tuple
\[
\tup{\Sigma,S,I,R,\buchiset,cc}
\]
which is like an NBW except that it has an additional \emph{counter command function} $cc:R \to \{\inc,\reset,\skp\}$ which associates with each transition a counter-update operation.
An infinite run $\rho = \rho_1 \rho_2 \cdots$ of a B-automaton induces a set $ctr(\rho) = \{c_i : i \in \Nat\}$, where $c_1 = 0$ and
\[
 c_{i+1} = \begin{cases}
            c_i & \mbox{ if } cc(\rho_i) = \skp\\
            c_i + 1 & \mbox { if } cc(\rho_i) = \inc\\
            0 & \mbox{ if } cc(\rho_i) = \reset.
           \end{cases}
\]
The run $\rho$ is \emph{accepting} if it satisfies the B\"uchi condition and its counter values are bounded. I.e., if $inf(\rho) \cap \buchiset \neq \emptyset$ and $\exists n \in \Nat$ s.t. $c < n$ for all $c \in ctr(\rho)$.

If $\buchiset = S$ (i.e., if there is effectively no B\"uchi acceptance condition), then we say that the B\"uchi set is \emph{trivial}.

B-automata were defined in~\cite{b10}, and in the general case may have multiple counters, some of which should be bounded and some of which should be unbounded. Since one can easily simulate a B\"uchi acceptance condition with a single counter (see \cite{b10}), our definition of a B-automaton given above is a special case of the B-automata of~\cite{b10} with two counters.
The proof of Lemma~\ref{lem: B automata emptiness} below, which also applies to these general multi-counter automata, was communicated to us by Nathana\"el Fijalkow (as far as we know it is a ``folk theorem'' for which we could not find a clear statement or proof in the literature).

\begin{lemma} \label{lem: B automata emptiness}
\label{lem:single-counter-b-automaton-emptiness}
Deciding whether the language of a B-automaton is not empty can be solved in \PTIME.
\end{lemma}

\begin{proof}
We reduce the problem to the emptiness problem for Streett automata~\footnote{We remind the reader that a Streett automaton is like an NBW except that its acceptance condition is not a single B\"uchi set, but a family of pairs of sets $\{(B_1, G_1), (B_2, G_2), \ldots, (B_k, G_k)\}$, and a run $\rho$ is accepting if for all $i \in [k]$ we have that $inf(\rho) \cap G_i \neq \emptyset$ implies $inf(\rho) \cap B_i \neq \emptyset$.}, which is in \PTIME~\cite{journals/scp/EmersonL87}.

Given a B-automaton $\A = \tup{\Sigma,S,I,R,\buchiset,cc}$, build a Streett automaton $\A'$ whose transition relation is like that of $\A$ except that it also stores the most recent counter command in the state, and whose acceptance condition encodes the following properties: `infinitely often see a state in $\buchiset$' and `infinitely many increments implies infinitely many resets'. Formally, $\A'$ has states $S \times \{\skp,\inc,\reset\}$; initial states $S \times \{\reset\}$; transitions of the form $((s,c),\alpha,(s',c'))$ where $(s,\alpha,s') \in R$ and $cc(s,\alpha,s') = c'$; and the acceptance condition containing the two pairs $(S,\buchiset)$ and $(S\times\{\inc\},S \times \{\reset\})$.

Then, the language of $\A$ is non-empty if and only if the language of $\A'$ is non-empty. To see this, note that accepting runs in $\A$ induce accepting runs in $\A'$, since a run with infinitely many "increments" that also has a bound on the counter must have infinitely many "resets". On the other hand, an accepting run $\rho_1$ of $\A'$ can be transformed, by "pumping out" loops, into another accepting run $\rho_2$ of $\A'$ in which the distance between two successive reset transitions is bounded (one only needs to ensure that in each infix starting and ending in a reset, if there is a B\"uchi state in the infix, then there is still one after pumping out).
Thus, the counter of $\rho_2$ is bounded, and so is also an accepting run of $\A$.\footnote{Observe that the Streett automaton does not, in general, accept the same language as the B-automaton. Indeed the latter's language may not even be $\omega$-regular.} 
\end{proof}

\subsection{Linear Temporal Logic}

For a set $AP$ of atomic propositions, \emph{formulas of \LTL  over $\AP$} are defined by the following BNF (where $p \in \AP$):
\[
\varphi \!::=\! p \!\mid\! \varphi \vee \varphi \!\mid\! \neg \varphi \!\mid\!  \nextX \! \varphi \!\mid \! \varphi \until \varphi
\]
We use the usual abbreviations, $\varphi \limp \varphi' = \neg \varphi \vee \varphi'$, $\true = p
\vee \neg p$, $\eventually \varphi = \true \until \varphi$ (read "eventually $\varphi$"), $\always \varphi = \lnot \eventually \lnot \varphi$ (read "always $\varphi$"). The \emph{size $|\varphi|$} of a formula $\varphi$
is the number of symbols in it.
A \emph{trace} $\tau$ is an infinite sequence over the alphabet $\Sigma = 2^{\AP}$, an infinite sequence of valuations of the
atoms. For $n \geq 0$, write $\tau_n$ for the valuation at position $n$; so, $\tau = \tau_0 \tau_1 \tau_2 \cdots$
Given a trace $\tau$, an integer $n$, and an LTL formula $\varphi$,
the satisfaction relation $(\tau,n) \models \varphi$, stating that
$\varphi$ holds at step $n$ of the sequence $\tau$, is defined as
follows:
\begin{itemize}
\item $(\tau,n) \models p$ iff $p \in \tau_n$;
\item $(\tau,n) \models \varphi_1 \vee \varphi_2$ iff $(\tau,n) \models \varphi_1$ or $(\tau,n) \models \varphi_2$;
\item 	$(\tau,n) \models \neg \varphi$ iff it is not the case that $(\tau,n) \models \varphi$;
\item 	$(\tau,n) \models \nextX \varphi$ iff $n + 1 < |\tau|$ and $(\tau,n+1) \models \varphi$;
\item 	$(\tau,n) \models \varphi_1 \until \varphi_2$ iff $(\tau,m) \models \varphi_2$ for some $n \leq m < |\tau|$, and $(\tau,j) \models \varphi_1$ for all $n \leq j < m$.
\end{itemize}
Write $\tau \models \varphi$ if $(\tau,0) \models \varphi$, read \emph{$\tau$ satisfies $\varphi$}.

We consider the variant \LTLf known as ``\LTL over finite traces"~\cite{BacchusK00,BaierM06,DegVa13}. It has the same syntax and semantics as \LTL except that $\tau$ is a finite sequence. Observe that the satisfaction of $\nextX$ and $\until$ on finite traces is defined ``pessimistically'', i.e., a trace cannot end before the promised eventuality holds.

The following states that one can convert \LTL/\LTLf formulas to NBW/NFW with at most an exponential blowup:
\begin{thmC}[\cite{Vardi:Banff95,DegVa13}] \label{thm:vardi-wolper}
Let $\varphi$ be an \LTL (resp. \LTLf) formula. One can build an NBW (resp. NFW), whose size is at most exponential in $|\varphi|$, accepting exactly the models $\varphi$.
\end{thmC}

\section{Parameterized systems}
\label{sec:parameterized-systems}
We first introduce \emph{systems with rendezvous and symmetric broadcast} (or \emph{RB-systems}, for short), a general formalism suitable for describing the parallel composition of $n \in \Nat$ copies of a process {\em template}. We identify two special cases: Rendezvous systems (or R-systems, for short) and Discrete Timed Systems.

\subsection{RB-systems} \label{sec: RB systems}
An RB-system is a certain LTS which evolves nondeterministically: either a $k$-wise rendezvous action is taken, i.e., $k$ different processes instantaneously synchronize on some rendezvous action $\msg{a}$, or the symmetric broadcast action is taken, i.e., all processes take an edge labeled by $\brd$.
Systems without the broadcast action are called R-systems. We will show that RB-systems (strictly) subsume discrete timed networks \cite{ADM04}, a formalism allowing to describe parameterized networks of timed processes with discrete value clocks. A discrete timed network also evolves nondeterministically: either a $k$-wise rendezvous action is taken by $k$ processes of the network, or all the clocks of all the processes advance their value by the same (discrete) amount.

In the rest of the paper the number of processes participating in a rendezvous will be denoted by $k$,  we let $\Actions$ denote a finite set of \emph{rendezvous actions}, and we call \emph{rendezvous alphabet} the set
$\Actionprts = \{ \msg{a}_i ~:~ \msg{a} \in \Actions, i \in [1,k] \}$.

\begin{definition}[\bf Process Template, RB-Template, R-Template] \label{dfn: rb-template}
{
A \emph{process template} is a finite LTS $\proctemp = \tup{\AP,\Actionprts \cup \{\brd\}, S,I,R,\lambda}$.
A process template $\proctemp$ is an {\em RB-template}, if for every state $s \in S$, we have that $\brd$ is enabled in $s$.We call edges labeled by $\brd$ {\em broadcast edges}, and the rest {\em rendezvous edges}.
A process template $\proctemp$ is an {\em R-template}, if $\proctemp$ does not contain any broadcast edges.}
\end{definition}

We now define the system $\sysinst{n}$ consisting of $n$ copies of a given template $\proctemp$:
\begin{definition}[\bf RB-System, R-System]
\label{def:rb-system}
Given an integer $n \in \Nat$ and an RB-Template (resp. R-Template) $\proctemp = \tup{\AP,\Actionprts \cup \{\brd\},S,I,R,\lambda}$, the {\em RB-system $\sysinst{n}$} (resp.  {\em R-system $\sysinst{n}$}) is defined as the finite LTS $\tup{\AP^n,\Comm^n, S^n,I^n,R^n,\lambda^n}$ where:

\begin{enumerate}
\item The set of atomic propositions $\AP^n$ is $\AP \times [n]$; intuitively, the atom $(p,i)$ denotes the fact that atom $p$ is currently true of process $i$.

\item The \emph{communication alphabet} $\Comm^n$ consists of $\brd$ and every tuple of the form $((i_1,\msg{a}_1), \ldots, \allowbreak (i_k,\msg{a}_k))$ where
$\msg{a} \in \Actions$ and $i_1,i_2,\cdots,i_k$ are $k$ different elements in $[n]$; intuitively, the system takes this action means that simultaneously for each $j \in [k]$, process $i_j$ transitions along an $\msg{a}_j$ edge.

\item $S^n$ is the set of functions (called {\em configurations}) of the form $f:[n] \to S$.
    We call $f(i)$ the {\em state of process $i$} in $f$. Note that we sometimes find it convenient to consider a more flexible naming of processes in which we let $S^n$ be the set of functions $f:X \to S$, where $X \subseteq \Nat$ is some set of size $n$.

\item The set of {\em initial configurations} $I^n = \{ f \in S^n \mid f(i) \in I \text{ for all } i \in [n]  \}$ consists of all configurations which map all processes to initial states of $\proctemp$.

\item The set of {\em global transitions} $R^n \subseteq S^n \times \Comm^n \times S^n$ contains transitions $\trans{f}{g}{\sigma}$ where one of the following two conditions hold:
    \begin{itemize}
    \item {\em (broadcast)} $\sigma = \brd$, and $\trans{f(i)}{g(i)}{\brd}$ in $R$, for every $i \in [n]$;

    \item {\em (rendezvous)} $\sigma = ((i_1, \msg{a}_1), \dots, (i_k, \msg{a}_k))$, and
    $\trans{f(i_j)}{g(i_j)}{\msg{a}_j}$ in $R$ for every $1 \leq j \leq k$; and $f(i) = g(i)$ for every $i \not \in \{i_1, \dots, i_k\}$. In this case we say that $\msg{a} \in \Actions$ is the {\em action taken}.

	\end{itemize}

\item The labeling relation $\lambda^n \subseteq S^n \times \AP^n$ consists of the pairs $(f,(p,i))$ such that $(f(i),p) \in \lambda$.
\end{enumerate}
\end{definition}

For every transition $t = (f,\sigma,g) \in R^n$, we define the set of \emph{active processes}, denoted by $\activeprocs(t)$, as follows: 
\begin{itemize}
\item if $\sigma = \brd$ define $\activeprocs(t) = [n]$, 
\item if $\sigma = ((i_1, \msg{a}_1), \dots, (i_k, \msg{a}_k))$ define $\activeprocs(t) = \{ i_1, \dots, i_k \}$. 
\end{itemize}

Let $t$ be a global transition $\trans{f}{g}{\sigma}$, and let $i$ be a process. We say that $i$ {\em moved} in $t$ if $i \in \activeprocs(t)$.
We write $edge_i(t)$ for the edge of $\proctemp$ taken by process $i$ in the transition $t$, i.e.,
\begin{itemize}
\item if
$\sigma = \brd$ then $edge_i(t)$ denotes $\trans{f(i)}{g(i)}{\brd}$;
\item if $\sigma = ((i_1, \msg{a}_1), \dots, (i_k, \msg{a}_k))$ then $edge_i(t)$ denotes $\trans{f(i)}{g(i)}{\msg{a}_j}$ if $\sigma(j) = (i, \msg{a}_j)$ for some $j \in [k]$;
\item otherwise $edge_i(t) := \bot$.
\end{itemize}
In case that $edge_i(t) \neq \bot$ we say that $edge_i(t)$ is \emph{taken} in $t$.
Given a run $\pi$ of $\proctemp^n$ and an edge $e$ of $\proctemp$, we say that $e$ \emph{appears} on $\pi$ if it is taken by some active process on some transition of $\pi$.

Given a process template $\proctemp$ define the {\em RB-system} $\psys$ as the following LTS:
\[
\tup{\AP^\infty,\Comm^\infty,S^\infty,I^\infty,R^\infty,\lambda^\infty}
\]
where
$\AP^\infty = \cup_{n \in \Nat} \AP^n$,
$S^\infty = \cup_{n \in \Nat} S^n$,
$I^\infty = \cup_{n \in \Nat} I^n$,
$R^\infty = \cup_{n \in \Nat} R^n$,
$\Comm^\infty = \cup_{n \in \Nat} \Comm^n$ and
$\lambda^\infty = \cup_{n \in \Nat} \lambda^n$.

\subsection{Discussion of our modeling choices} \label{sec:discussion}
Our definition of RB-systems allows one to model finitely many different process templates because a single process template $\proctemp$ can have multiple initial states  (representing the disjoint union of the different process templates).

We can easily transform a rendezvous action $\msg{a}$ involving $j < k$ processes (in particular where $j=1$, representing an internal transition taken by a single process) into a $k$-wise rendezvous action by simply adding, for every $j < i \leq k$, and every state in $\proctemp$, a self-loop labeled $\msg{a}_i$.
This transformation works when there are at least $k$ processes in the system. This is not a real restriction since all the systems with less than $k$ processes yield a single finite-state system which can be easily model-checked. In any case, all our results hold also if one specifically allows rendezvous actions involving $j < k$ processes.

The assumption that every state in an RB-template is the source of a broadcast edge means that for every configuration $f$ there is a broadcast global-transition with source $f$.

\subsection{Executions and Specifications}\label{sec:executions}

Take an RB-system $\sysinst{n} =\!\tup{\AP^n,\Comm^n,S^n,I^n,R^n, \allowbreak \lambda^n}$, a path $\pi = t_1 t_2 \dots$ in $\sysinst{n}$, and a process $i$ in $\sysinst{n}$. Define $\proj{\pi}(i) := edge_i(t_{j_1}) edge_i(t_{j_2}) \dots$, where $j_1 < j_2 < \dots$ are all the indices $j$ for which $edge_i(t_j) \neq \bot$. Thus, $\proj{\pi}(i)$ is the path in $\proctemp$ taken by process $i$ during the path $\pi$. Define the set of {\em executions} of $\psys$, denoted by $\exec$, to be the set of the runs of $\psys$ projected onto the state labels of a single process. Note that, due to symmetry, we can assume w.l.o.g. that the runs are projected onto process $1$. Formally,
\[\exec = \{ \lambda(\proj{\pi}(1)) \mid \pi \in \runs(\psys) \},\]
{where $\lambda$ is the labeling of the process template $P$.}
We denote by $\execfin$ (resp. $\execinf$) the finite (resp. infinite) executions in $\exec$.

\begin{figure}[t]
\centering
\begin{minipage}{.45\textwidth}
  \centering

\begin{tikzpicture}[->,>=latex,node distance=0.6cm,bend angle=25,auto]

\tikzset{every state/.style={circle,minimum size=.4cm,inner sep=0cm}}
\tikzset{every edge/.append style={font=\small}}

\node[initial,state] (t1) {$p$};
\node[state] (t2) [right= of t1] {$q$};

\path (t1) edge [loop above] node {{$\msg{a}_1$}} (t1);
\path (t1) edge [bend left] node {{$\msg{a}_2$}} (t2);

\end{tikzpicture}  \vspace{2cm}
	\caption{R-template}\label{fig: strict}
\end{minipage}
\begin{minipage}{.45\textwidth}
  \centering

\begin{tikzpicture}[->,>=latex,node distance=0.8cm,bend angle=25,auto]

\tikzset{every state/.style={circle,minimum size=.4cm,inner sep=0.cm}}
\tikzset{every edge/.append style={font=\small}}

\node[state] (t1) {$p$};
\node[initial,state] (t3) [below= of t1] {$r$};
\node[state] (t2) [right= of t3] {$q$};

\path (t1) edge [loop above] node {$\msg{a}_1$} (t1);
\path (t3) edge [bend left] node {$\msg{a}_1$} (t1);

\path (t3) edge [bend right, below] node {{$\msg{a}_2$}} (t2);

\path (t3) edge [loop below] node {$\brd$} (t3);
\path (t1) edge node {$\brd$} (t3);
\path (t2) edge [above] node {$\brd$} (t3);

\end{tikzpicture}

%
%
%
%
%
%
	\caption{RB-template}\label{fig: process not regular 2}
\end{minipage}
	\end{figure}

We present two examples. Note that
the letter in a state is both the name of that state as well as the (unique) atom that holds in that state.
\begin{example} \label{ex:R-template}
Consider the R-template $P$ in Figure~\ref{fig: strict}. Note that
$\execfin$ consists of all prefixes of words that match the regular expression $pp^*q$, and $\execinf = \emptyset$.
\end{example}

\begin{example} \label{ex:RB-template} 
Consider the RB-template $P$ in Figure~\ref{fig: process not regular 2}. In every run of $\sysinst{n}$, every process is involved in at most $n-1$ consecutive rendezvous transitions before a broadcast transition is taken, which resets all processes to the initial state. Thus, $\execfin$ consists of all prefixes of words that match
$(rr^*(pp^*+q))^*$ since $n$ may be arbitrarily large. Similarly, $\execinf$ is the set of words of the form
$(rr^*(pp^*+q))^*r^\omega$
and $r^{n_1}(p^{x_1} + q) \allowbreak r^{n_2}(p^{x_2} + q) \dots$, for some sequences $n_1, n_2, \cdots$ and $x_1,x_2,\cdots$ of positive integers such that $\{x_i\}_i$ is bounded. 
The first form occurs in case process $1$ is involved in only finitely many rendezvous transitions, and the second form occurs if process $1$ is involved in infinitely many rendezvous transitions. Indeed, in the latter case, the $n_i$s are unconstrained since a broadcast can occur any number of times before a rendezvous, and if $\{x_i\}_i$ is bounded by $B \in \Nat$ then the execution can be realised in the RB-System $\proctemp^{B+1}$.
\end{example}

The definitions above imply the following easy lemma.

\begin{lemma} \label{lem:bisim}
Let $P,P'$ be two RB-templates with the same atomic propositions and edge-labels alphabet.
\begin{enumerate}
 \item If $\pi \in \runs(P^\infty)$ and $\pi' \in \runs(P'^\infty)$ are equi-labeled then $\lambda(\proj{\pi}(j)) = \lambda'(\proj{\pi'}(j))$ for every process $j$.\label{lem:bisim:equi}
 \item If $P$ and $P'$ simulate each other then $\exec[P^\infty] = \exec[P'^\infty]$. \label{lem:bisim:bisim}
\end{enumerate}
\end{lemma}

\begin{proof}
Let $\pi = e_1 e_2 \cdots$ and $\pi' = e'_1 e'_2 \cdots$.  For the first item note that for every $i < |\pi|$
the label of the edge $e_i$ is equal to the label of the edge $e'_i$.
It follows that the same processes are active in $e_i$ and $e'_i$.
Thus, the run $\proj{\pi}(i)$ of the template $P$ is equi-labeled with the run $\proj{\pi'}(i)$ of the template $P'$, and in particular they induce the same sequence of sets of atomic propositions.

For the second item, suppose $P'$ simulates $P$ via $\simrel \subseteq S \times S'$. Derive the relation $C \subseteq S^{\infty} \times (S')^\infty$ from $\simrel$ point-wise, i.e.,
$(f,f') \in C$ iff there is $n \in \Nat$ such that i) $f \in S^n, f' \in (S')^n$ and ii) for every $i \leq n$, $(f(i),f'(i)) \in \simrel$. It is routine to check that $C$ is a simulation, and thus $P'^\infty$ simulates $P^\infty$. By a symmetric argument, $P^\infty$ simulates $P'^\infty$. Thus, for every run $\pi$ in $P^\infty$ there is an equi-labeled run $\pi'$ in $P'^\infty$, and vice versa. Now apply the first item. 
\end{proof}

\subsection{Parameterized Model-Checking Problem.} \label{subsec: PMCP}
Specifications represent sets of finite or infinite sequences over the alphabet $2^{\AP}$. In this work we will consider specifications of finite executions to be given by nondeterministic finite word automata (NFW) and specifications of infinite executions to be given by nondeterministic B\"uchi word automata (NBW). Standard translations allow us to present specifications in linear temporal logics such as \LTL\ and $\LTL_f$, see Theorem~\ref{thm:vardi-wolper}.

We now define the main decision problem of this work.
{
\begin{definition}[\bf PMCP] \label{dfn:PMCP}
Let $\Pspec$ be a specification formalism for sets of infinite (resp. finite) words over the alphabet $2^{\AP}$.
The {\em Parameterized Model Checking Problem} for $\Pspec$, denoted $PMCP(\Pspec)$, is to decide, given a process-template $\proctemp$, and a set $W$ of infinite (resp. finite) words specified in $\Pspec$, if all executions in the set $\execinf$ (resp. $\execfin$) are in $W$.
\end{definition}
}

Just as for model-checking~\cite{Vardi:Banff95}, we have three ways to measure the complexity of the PMCP problem. If we measure the complexity in the size of the given template and specification, we have the (usual) complexity, sometimes called \emph{combined complexity}. If we fix the template and measure the complexity with respect to the size of the specification we get the \emph{specification complexity}. If we fix the specification and measure the complexity with respect to the size of the template we get the \emph{program complexity} (we use "program complexity" instead of "template complexity" in order to be consistent with the model-checking terminology). Moreover, if $C$ is a complexity class, we say that the specification complexity of the PMCP-problem is \emph{$C$-hard} if there is a fixed template such that the induced PMCP problem (that only takes a specification as input) is $C$-hard in the usual sense; and it is \emph{$C$-complete} if it is in $C$ and $C$-hard. Symmetric definitions hold for program complexity.

\subsection*{Results}
Our main results solve PMCP for RB-templates for specifications of finite and infinite executions. In both cases we use the automata theoretic approach: given an RB-template $\proctemp$, we show how to build an automaton $\M$ accepting exactly the executions of the RB-system $\psys$. Model checking of a specification given by another automaton $\M'$ is thus reduced to checking if the language of $\M$ is contained in the language of $\M'$. The automaton $\M$ will be based on what we call the reachability-unwinding of the given RB-template. In the finite-execution case $\M$ will be an NFW that is almost identical to the reachability-unwinding. In the infinite execution case $\M$ will be more complicated. Indeed, classic automata over infinite words (e.g., NBW) will not be powerful enough to capture the system (see Lemma~\ref{lem: infinite exec no regular} below) and so we use B-automata; the automaton $\M$ will be based on three copies  of the reachability-unwinding (instead of one copy) where from each copy certain edges will be removed based on a classification of edges into different types.

Our first result classifies the complexity for NFW/\LTLf specifications:
\begin{theorem}  \label{thm:PSPACE-complete}
Let $\Pspec$ be specifications of sets of finite executions expressed as NFW or \LTLf formulas. Then the complexity of $PMCP(\Pspec)$ for RB-systems is \PSPACE-complete, as is the program complexity and the specification complexity.
\end{theorem}

To see that classical acceptance conditions (e.g. B\"{u}chi, Parity) are not strong enough for the case of infinite executions, consider the following lemma.

\begin{lemma}\label{lem: infinite exec no regular}
  The process template $\proctemp$ in Figure~\ref{fig: process not regular 2} has the property that the set $\execinf$ is not $\omega$-regular.
\end{lemma}
\begin{proof}

The following pumping argument shows that this language is not $\omega$-regular. Assume by way of contradiction that an NBW $\A$ accepts $\execinf$, and consider an accepting run of $\A$ on the word $(rp^{n+1})^\omega$, where $n$ is the number of states of $\A$. It follows that for each $i \in \Nat$, while reading the $i$'th block of $p$'s, $\A$ traverses some cycle $c_i$. Hence, by correctly pumping the cycle $c_i$, e.g., $i$ times for every $i \in \Nat$, we can obtain an accepting run of $\A$ on a word $w'$ which is not $\execinf$ since it contains blocks of consecutive $p$'s of ever increasing length, contradicting our assumption.
\end{proof}

On the other hand, there is a $B$-automaton (with a trivial B\"uchi set) recognizing this language (the counter is incremented whenever $p$ is seen and reset whenever $r$ is seen). This is no accident: we will prove (Theorem~\ref{thm: BSW correctness}) that for every RB-template $\proctemp$ one can build a $B$-automaton (with a trivial B\"uchi set) recognizing the infinite executions of $\psys$. Combining this with an NBW for the specification, we reduce the model-checking problem to the emptiness problem of a $B$-automaton. Hence, our second main result provides an \EXPTIME upper bound for NBW/\LTL specifications:

\begin{theorem} \label{thm: main dec}
Let $\Pspec$ be specifications of sets of infinite executions expressed as NBW or LTL formulas. Then $PMCP(\Pspec)$ of RB-systems can be solved in \EXPTIME. 
\end{theorem}

\subsection{Variants with a controller and asymmetric broadcast}

We now give two variants of RB-Systems, i.e., one that incorporate a distinguished "controller" process, and another that allows for asymmetric broadcasts~\cite{efm99}.

Given two process templates $\proctemp_C$ and $\proctemp$ the RB-System with a controller (RBC-System) $\proctemp_C \cup \sysinst{n}$ is the finite LTS $\tup{\AP^{n+1},\Comm^{n+1}, S^{n+1},I^{n+1},R^{n+1},\lambda^{n+1}}$, which is defined exactly as in Definition~\ref{def:rb-system}, with the only difference that for process 1 we use the process template $\proctemp_C$ and for processes 2 to $n+1$ we use process template $\proctemp$.
The RBC-System $\proctemp_C \cup \sysinst{\infty}$ is then defined analogously.
We now need to adjust the definitions of executions to differentiate between the projection to a controller resp. non-controller processes;
we set,
\[\exec[\proctemp_C \cup \sysinst{\infty}]_C = \{ \lambda(\proj{\pi}(1)) \mid \pi \in \runs(\proctemp_C \cup \sysinst{\infty}) \},\]
and
\[\exec[\proctemp_C \cup \sysinst{\infty}] = \{ \lambda(\proj{\pi}(2)) \mid \pi \in \runs(\proctemp_C \cup \sysinst{\infty}) \},\]
where, because of symmetry, we can always project to process 2 for a non-controller process.

In order to capture asymmetric broadcasts we need to enhance the edge-labels alphabet of a process template. This is done by introducing a set of broadcast actions $\Broadcasts$ (disjoint from $\Actions$), and setting the edge-labels alphabet to be $\Actionprts \bigcup \cup_{\msg{b} \in \Broadcasts} \{ \msg{b}_\snd, \msg{b}_\rcv\}$. Such a process template  is called an {\em RBA-template} if, in addition, for every state $s \in S$ we have that $\msg{b}_\rcv$ is enabled in $s$ for every $\msg{b} \in \Broadcasts$.
Given an RBA-template $\proctemp$, the RBA-System $\sysinst{n}$ is the finite LTS $\tup{\A^n,\Comm^n, S^n,I^n,R^n,\lambda^n}$, which is defined as in Definition~\ref{def:rb-system} except for the definition of the global transition relation, where we support asymmetric broadcasts instead of symmetric broadcasts as follows:
    \begin{itemize}
    \item {\em (asymmetric broadcast)} $\sigma = \tup{c_1,\ldots,c_n}$ is an $n$-tuple such that there is some $\msg{b} \in \Broadcasts$ and some $i$ such that $c_i =\msg{b}_\snd$ and $c_j =\msg{b}_\rcv$ for all $j \neq i$, and $\trans{f(i)}{g(i)}{c_i}$ in $\proctemp$, for every $i \in [n]$;
    \end{itemize}
The RBA-System $\sysinst{\infty}$ is then defined analogously.
The set of executions is defined as for RB-Systems.

We remark that we define RBA-Systems without a controller for technical convenience.
It would be straight-forward to define an RBA-System with a controller in the same way as we did above.
However, it is easy to verify that RBA-Systems with a controller are not more powerful than RBA-Systems that lack a controller.
That is because having a controller can be simulated through an initial asymmetric broadcast that makes the sender process taking over the role of the controller and the receiver processes continuing as non-controller processes. We now make this statement precise. Recall the notation $(\xi)_a$ and $\xi|_\AP$ from Section~\ref{sec:transition systems}.

\begin{theorem} \label{thm:RBC equiv RBA}
RBC-Systems and RBA-Systems are equally powerful, more precisely,
\begin{enumerate}
\item for each RBC-System, given by process templates $\proctemp_C$ and $\proctemp$ over atomic propositions $\AP$, we can construct in linear time an RBA-System $\proctemp'$ over atomic propositions $\AP \cup \{c,p\}$, with $c,p \not\in\AP$ such that $\exec[\proctemp_C \cup \proctemp^\infty]_C = \{ (\xi)_{c}|_\AP \mid \xi \in \exec[\proctemp'^\infty] \}$ and $\exec[\proctemp_C \cup \proctemp^\infty] = \{ (\xi)_{p}|_\AP \mid \xi \in \exec[\proctemp'^\infty]  \}$;

\item for each RBA-System, given by process template $\proctemp$ over atomic propositions $\AP$, we can construct  in linear time an RBC System, given by process templates $\proctemp'_C$ and $\proctemp'$ over atomic propositions $\AP \cup \{p\}$ such that
     $\exec[\proctemp'^\infty] = \{ (\xi)_{p}|_\AP \mid \xi \in \exec[\proctemp'_C \cup \proctemp'^\infty]  \}$
    (the executions of the controller are not important for this statement).
\end{enumerate}
\end{theorem}

\begin{proof}
For the first item, the symmetric broadcast can be easily implemented by an asymmetric broadcast where the sender simply behaves like the receivers, and a controller can be elected using an initial asymmetric broadcast. Here are the details. The broadcast alphabet $\Broadcasts$ consists of two symbols, $\brd$, b: the first for modeling the symmetric broadcast of the RBC-System, and the second to be used for electing the controller.
Define the RBA-template $\proctemp' = \tup{\AP \cup \{c,p\},\Actionprts \cup \{\msg{b}_\rcv, \msg{b}_\snd, \brd_\rcv, \brd_\snd\}, S',\{\iota\},R',\lambda'}$ as follows. Let $S'$ consist of the states of $\proctemp_C$, the states of $\proctemp$ (assumed to be disjoint from the states of $\proctemp_C$), and a new initial state $\iota$.
The transitions relation $R'$ includes all rendezvous transitions of $\proctemp_C$ and $\proctemp$, as well as the following new transitions:
\begin{itemize}
\item $\trans{s}{t}{\brd_\rcv}$ and $\trans{s}{t}{\brd_\snd}$ for every transition $\trans{s}{t}{\brd}$ in $\proctemp_C$ and $\proctemp$;
\item $\trans{\iota}{\iota}{\brd_\rcv}$ and $\trans{\iota}{\iota}{\brd_\snd}$;
\item $\trans{\iota}{\iota_C}{\msg{b}_\snd}$ and $\trans{\iota}{\iota_P}{\msg{b}_\rcv}$ where $\iota_C$ and $\iota_{P}$ are the initial states of $\proctemp_C$ and $\proctemp$ respectively;
\item $\trans{s}{s}{\msg{b}_\rcv}$ for every $s \neq \iota$ (these transitions can never be taken but are added to satisfy the constraint that broadcast can always be received, as required by the definition of RBA-templates).
\end{itemize}
Finally, let $\lambda_C$ and $\lambda_P$ be the labeling functions of $\proctemp_C$ and $\proctemp$ respectively. Define the labeling function $\lambda'$ as follows. If $s$ is a state of $\proctemp$ then $\lambda'(s) = \lambda_P(s) \cup \{p\}$; if $s$ is a state of $\proctemp_C$ then $\lambda'(s) = \lambda_C(s) \cup \{c\}$; and $\lambda'(\iota) = \emptyset$.
It is not hard to verify that $\exec[\proctemp_C \cup \proctemp^\infty]_C = \{ (\xi)_{c}|_\AP \mid \xi \in \exec[\proctemp'^\infty] \}$ and $\exec[\proctemp_C \cup \proctemp^\infty] = \{ (\xi)_{p}|_\AP \mid \xi \in \exec[\proctemp'^\infty]  \}$.

For the second item in the statement of the theorem, the idea is depicted in Figure~\ref{fig:simulation_rba_rbc}.
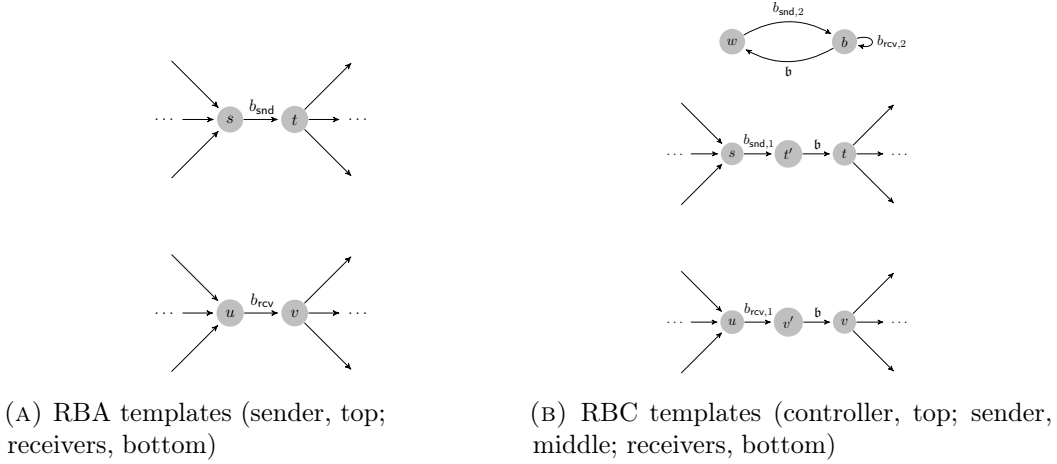
\begin{figure}
      \centering
      \begin{subfigure}{.45\textwidth}
        \centering

        \resizebox{0.45\linewidth}{!}{\begin{tikzpicture}[->,>=stealth',shorten >=2pt,auto,node distance=1.5cm,
                    semithick]
    \tikzstyle{invisible}=[]

    \tikzstyle{vertex}=[circle,fill=black!25]


    \node[vertex]       (s)                    { $s$ };
    \node[vertex]       (t) [right of=s]        { $t$ };

    \node[invisible]    (v2) [left of=s]                  {\ldots};
    \node[invisible]    (v1) [above of=v2]     {};
    \node[invisible]    (v3) [below of=v2]     {};

    \node[invisible]    (v5) [right of=t]       {\ldots};
    \node[invisible]    (v4) [above of=v5]      {};
    \node[invisible]    (v6) [below of=v5]      {};

    \path (v1) edge (s);
    \path (v2) edge (s);
    \path (v3) edge (s);
    \path (s) edge node {$\msg{b}_\snd$} (t);
    \path (t) edge (v4);
    \path (t) edge (v5);
    \path (t) edge (v6);

    \node[invisible]  (space2) [below of=s] {};
    \node[invisible]  (space3) [below of=space2] {};


    \node[vertex]       (u) [below of=space3]       { $u$ };
    \node[vertex]       (v) [right of=u]         { $v$ };

    \node[invisible]    (v11) [left of=u]      {\ldots};
    \node[invisible]    (v10) [above of=v11]     {};
    \node[invisible]    (v12) [below of=v11]     {};

    \node[invisible]    (v14) [right of=v]       {\ldots};
    \node[invisible]    (v13) [above of=v14]      {};
    \node[invisible]    (v15) [below of=v14]      {};

    \path (v10) edge (u);
    \path (v11) edge (u);
    \path (v12) edge (u);
    \path (u) edge node {$\msg{b}_\rcv$} (v);
    \path (v) edge (v13);
    \path (v) edge (v14);
    \path (v) edge (v15);

\end{tikzpicture}}
      \caption{RBA templates (sender, top;\\ receivers, bottom)}

      \end{subfigure}
      \begin{subfigure}{.45\textwidth}
        \centering

      \resizebox{0.5\linewidth}{!}{\begin{tikzpicture}[->,>=stealth',shorten >=2pt,auto,node distance=1.5cm,
                    semithick]
    \tikzstyle{invisible}=[]

    \tikzstyle{vertex}=[circle,fill=black!25]



    \node[vertex]       (w)    {$w$};
    \node[invisible]    (z1) [right of=w]   {};
    \node[vertex]       (z) [right of=z1]        {$b$};

    \path   (w) edge [bend left] node  {$\msg{b}_{\snd,2}$} (z);
    \path   (z) edge [bend left] node {$\brd$} (w);
    \path   (z) edge [loop right] node  {$\msg{b}_{\rcv,2}$}(z);

    \node[invisible]    (space) [below of=w] {};


    \node[vertex]       (s) [below of=space]       { $s$ };
    \node[vertex]       (t1) [right of=s]        { $t'$ };
    \node[vertex]       (t) [right of=t1]        { $t$ };

    \node[invisible]    (v2) [left of=s]                  {\ldots};
    \node[invisible]    (v1) [above of=v2]     {};
    \node[invisible]    (v3) [below of=v2]     {};

    \node[invisible]    (v5) [right of=t]       {\ldots};
    \node[invisible]    (v4) [above of=v5]      {};
    \node[invisible]    (v6) [below of=v5]      {};

    \path (v1) edge (s);
    \path (v2) edge (s);
    \path (v3) edge (s);
    \path (s) edge node {$\msg{b}_{\snd,1}$} (t1);
    \path (t1) edge node {$\brd$} (t);
    \path (t) edge (v4);
    \path (t) edge (v5);
    \path (t) edge (v6);

    \node[invisible]  (space2) [below of=s] {};
    \node[invisible]  (space3) [below of=space2] {};


    \node[vertex]       (u) [below of=space3]       { $u$ };
    \node[vertex]       (v1) [right of=u]         { $v'$ };
    \node[vertex]       (v) [right of=v1]         { $v$ };

    \node[invisible]    (v11) [left of=u]      {\ldots};
    \node[invisible]    (v10) [above of=v11]     {};
    \node[invisible]    (v12) [below of=v11]     {};

    \node[invisible]    (v14) [right of=v]       {\ldots};
    \node[invisible]    (v13) [above of=v14]      {};
    \node[invisible]    (v15) [below of=v14]      {};

    \path (v10) edge (u);
    \path (v11) edge (u);
    \path (v12) edge (u);
    \path (u) edge node {$\msg{b}_{\rcv,1}$} (v1);
    \path (v1) edge node {$\brd$} (v);
    \path (v) edge (v13);
    \path (v) edge (v14);
    \path (v) edge (v15);

\end{tikzpicture}}
      \caption{RBC templates (controller, top; sender, middle; receivers, bottom)}
      \end{subfigure}
            \caption{\label{fig:simulation_rba_rbc}{An illustration of how RBA processes can be simulated by RBC processes.}}
      \end{figure}
Intuitively, to simulate an asymmetric broadcast, we use a symmetric broadcast preceded by a rendezvous with the controller, which ensures that there is single sender and multiple receivers.
For an RBA template $\proctemp = \tup{\AP, \Actionprts \bigcup \cup_{\msg{b} \in \Broadcasts} \{ \msg{b}_\snd, \msg{b}_\rcv\}, S,I,R,\lambda}$, we construct two process templates
\[
\proctemp'_C = \tup{\AP \cup \{p\}, \Actionprts' \cup \{\brd\}, S'_C, I'_C, R'_C, \lambda'_C}
\]
and
\[
\proctemp' = \tup{\AP \cup \{p\}, \Actionprts' \cup \{\brd\}, S', I', R', \lambda'},
\]
as follows. The new rendezvous actions consist of the old rendezvous actions as well as $\msg{b}_{\snd},\msg{b}_{\rcv}$ for every broadcast action $b \in \Broadcasts$; thus, we have
\[
\Actionprts' = \Actionprts \bigcup \cup_{\msg{b} \in \Broadcasts} \{ \msg{b}_{\snd,1},\msg{b}_{\snd,2},\msg{b}_{\rcv,1},\msg{b}_{\rcv,2}\}
\] (recall the remarks in Section~\ref{sec:discussion} that allow us to have rendezvous actions involving any number of processes).

We now define the components of the controller $\proctemp'_C$:
\begin{itemize}
\item $S'_C = \{w,\bot_C\} \cup \Broadcasts$, where $w$ (the ``waiting'' state) and $\bot_C$ (the ``dead'' state) are new states;
\item $I'_C = \{w\}$;
\item $\lambda'_C(s) = \emptyset$ for all $s \in S'_C$ (the labeling function of $P'_C$ is not important);
\item the transitions in $R'_C$ consist of rendezvous transitions $\trans{w}{\msg{b}}{\msg{b}_{\snd,2}}$ and $\trans{\msg{b}}{\msg{b}}{\msg{b}_{\rcv,2}}$ for every $b \in \Broadcasts$, and the symmetric broadcast transitions $\trans{\msg{b}}{w}{\brd}$ and $\trans{w}{\bot_C}{\brd}$ and $\trans{\bot_C}{\bot_C}{\brd}$.
\end{itemize}

We now define the components of the process template $\proctemp'$:
\begin{itemize}

\item let $S'$ consist of the states in $S$, a new state $\bot$, and new (''intermediate'') states of the form $q'$ where $q$ varies over the states of $\proctemp$ that are the destination of some asymmetric broadcast transition in $\proctemp$;
\item $I' = I$;
\item $\lambda'(s) = \lambda(s) \cup \{p\}$ for $s \in S$, and $\lambda'(q') = \emptyset$ for $q' \in S' \setminus S$ (in particular, $\lambda'(\bot) = \emptyset$);
\item the transitions $R'$ include all the rendezvous transitions of $\proctemp$, the rendezvous transition
$\trans{s}{t'}{\msg{b}_{\snd,1}}$ for every asymmetric send-broadcast transition $\trans{s}{t}{\msg{b}_\snd}$ of $\proctemp$, and the rendezvous transition
$\trans{u}{v'}{\msg{b}_{\rcv,1}}$ for every asymmetric receive-broadcast transition $\trans{u}{v}{\msg{b}_\rcv}$ of $\proctemp$. Further, $R'$ includes,  for every $q' \in S'$, the symmetric broadcast transition $\trans{q'}{q}{\brd}$, and for every state $s \in S$, the symmetric broadcast transitions $\trans{s}{\bot}{\brd}$, and the symmetric broadcast transition $\trans{\bot}{\bot}{\brd}$.
\end{itemize}

The states $\bot_C,\bot$ are not drawn in Figure~\ref{fig:simulation_rba_rbc}.

We now argue that $\exec[\proctemp^\infty] = \{ (\xi)_p|_\AP \mid \xi \in \exec[\proctemp'_C \cup \proctemp'^\infty] \}$. In what follows,
although the $n$ processes over the template $\proctemp'$ in the RBA-system $\proctemp'_C \cup \proctemp'^n$ are numbered $2, \cdots, n+1$, we will number them $1, \cdots, n$, and refer to the process over template $\proctemp'_C$ as ``the controller'' instead of as process $1$.

We first argue the inclusion $\exec[\proctemp^\infty] \subseteq \{ (\xi)_p|_\AP \mid \xi \in \exec[\proctemp'_C \cup \proctemp'^\infty]  \}$:
Given a run $\pi$ of $\proctemp^n$ (for some $n$), we will construct a corresponding run $\pi'$ of $\proctemp'_C \cup \proctemp'^n$.
Every rendezvous transition in $\pi$ is simulated in $\pi'$ by the corresponding rendezvous transition of $\proctemp'^n$. Every asymmetric broadcast transition with action $\msg{b} \in \Broadcasts$ in $\pi$ is simulated in $\pi'$ by a sequence of transitions of $\proctemp'_C \cup \proctemp'^{\infty}$, as follows. Suppose process $j$ in the RBA-system is sending the asymmetric broadcast by taking the edge $\trans{s}{t}{\msg{b}_\snd}$ in $\proctemp$. Then, process $j$ in the RBC-system rendezvouses with the controller on the action $\msg{b}_\snd$, i.e., process $j$ takes the edge $\trans{s}{t'}{\msg{b}_{\snd,1}}$ in $\proctemp'$, and the controller takes the edge $\trans{w}{b}{\msg{b}_{\snd,2}}$ in $\proctemp'_C$. Further, consider the remaining processes in the RBA-system that are receiving the broadcast (in some arbitrary order), and suppose process $i$ it takes the edge $\trans{u}{v}{\msg{b}_\rcv}$ in $\proctemp$. Then, that process in the RBC-system rendezvouses with the controller on action $\msg{b}_\rcv$, i.e., process $i$ takes the edge $\trans{u}{v'}{\msg{b}_{\rcv,1}}$ in $\proctemp'$, and the controller takes the edge $\trans{b}{b}{\msg{b}_{\rcv,2}}$ in $\proctemp'_C$. Once all the receiving processes in the RBA-system have been accounted for in this way, the RBC-system does a symmetric broadcast, so the controller returns to the waiting state $w$ (i.e., it takes the edge $\trans{b}{w}{\brd}$ in $\proctemp'_C$), and all the non-controller processes go to the target of the broadcast transition (i.e., a process in state $v'$ takes the edge $\trans{v'}{v}{\brd}$ in $\proctemp'$). In this way, the runs $\pi,\pi'$ agree on the projection onto a single (non-controller) process when the intermediate states, whose label does not include the atom $p$, are removed.

We now argue $\exec[\proctemp^\infty] \supseteq \{ (\xi)_p|_\AP \mid \xi \in \exec[\proctemp'_C \cup \proctemp'^\infty]\}$: given a run $\pi'$ of $\proctemp'_C \cup \proctemp'^n$ (for some $n$), we show there is a corresponding run $\pi$ of $\proctemp^n$. Unlike the previous inclusion, this time our construction will be tailored to faithfully simulate the behavior of only a single process, i.e., $\pi$ will depend on the process of interest. Recall from Section~\ref{sec:executions} that, due to symmetry, we can restrict our attention to process $1$ (which, by our note above, is a process with template $\proctemp$ resp. $\proctemp'$, and not the controller).

We claim that, w.l.o.g., $\pi'$ is structured as a sequence of blocks of the form: first, zero or more rendezvous transitions that do not involve the controller; then, unless there is no further symmetric broadcast, $k$ many rendezvous transitions that do involve the controller for some $k$ with $1 \leq k \leq n$ (the first of these is on an action $\msg{b}_\snd$, and the remainder on $\msg{b}_\rcv$), at least one of which involves process $1$, followed by a symmetric broadcast.  To see why this claim holds, first note that we can swap the order of successive transitions if the first of these is a rendezvous between the controller and a process, say process $j$, and the second is a rendezvous without the controller, say amongst processes $X$ (the reason for this is that the rendezvous with the controller puts process $j$ into an intermediate state, from which there is no rendezvous edge, and thus $j$ is not in $X$, and thus the global state, as well as the projection onto process $1$, after these transitions is not changed by the swap). Second, consider the transitions between two symmetric broadcasts. Repeatedly swap as above until all rendezvous transitions that do not involve the controller (if any) are before all rendezvous transitions that do (if any). This completes the block, unless process $1$ does not occur in a rendezvous with the controller. But in this case we can safely trim $\pi'$ after the last transition in which process $1$ is active since the subsequent symmetric broadcast puts process $1$ in the sink $\bot$ that has an empty label. Third, if there are finitely many symmetric broadcasts, consider the (possibly infinitely many) transitions after the last symmetric broadcast (if there are no symmetric broadcasts, then consider all the transitions). Note that there are at most $n$ such transitions that rendezvous with the controller. Since a process that rendezvous with the controller is put into a state where it cannot rendezvous with any other processes, we can repeatedly remove the last such transition until there are none left (note that even if process $1$ is involved in such a transition, it is put into an intermediate state by such a transition, which has empty label, and thus can safely be removed).

It should be now quite clear how to simulate each block of $\pi'$: every rendezvous transition that does not involve the controller is simulated by a corresponding rendezvous transition in $\proctemp^n$; the rest of the block, if any, can be simulated by a single asymmetric broadcast transition of $\proctemp^n$. Note, however, that processes in $\proctemp'^n$ that do not rendezvous with the controller in this block will no longer be faithfully simulated since such processes will transition to the dead state $\bot$ on the symmetric broadcast. This is not a problem since these processes will never again participate in any rendezvous transitions, and thus will have no influence on the ability of the projection of $\pi$ on process $1$ to faithfully capture the projection of $\pi'$ on process $1$. Hence, we are free to define the target states in $\pi$ of these processes during future asymmetric broadcasts, if any (note that by assumption at least one target state exists for every asymmetric broadcast).
\end{proof}

The PMCP for RBA-Systems is quite well understood~\cite{efm99}, i.e., it is undecidable for $\omega$-regular specifications, and decidable for regular specifications. We thus get:

\begin{theorem} \hfill
\begin{enumerate}
\item Let $\Pspec$ be specifications of sets of infinite executions expressed as NBW or LTL formulas. Then $PMCP(\Pspec)$ of RBC-systems is undecidable.
\item Let $\Pspec$ be specifications of sets of finite executions expressed as NFW or LTLf formulas. Then $PMCP(\Pspec)$ of RBC-systems is decidable.
\end{enumerate}
\end{theorem}
\begin{proof}
By Theorem~\ref{thm:RBC equiv RBA}, one can transfer verification tasks between RBA- and RBC-Systems:
This can be achieved by a formula/automaton transformation, and the verification of the transformed formula/automaton on the transformed system.
We exemplify how to transfer LTL/LTLf specifications:
Given an LTL/LTLf specification $\phi$, we implement the projection operation $(\cdot)_p$ by replacing every occurrence of an atomic proposition $X$ in $\phi$ by $\lnot p \until (p \land X)$, resulting in a formula $\phi'$;
in order to only consider executions that do not reach the dead state we can then consider the specification $(\always \lnot p \until p) \rightarrow \phi'$.
Similar ideas can be implemented through automata transformations.
In particular, the undecidability (resp. decidability) of $\omega$-regular (resp. regular) specifications for RBA-systems transfers to RBC-systems.
\end{proof}

\subsection{Large systems simulating small systems}
We provide a simple but useful property of RB-systems that will be used throughout the rest of this paper. Intuitively, a large RB-system can, using a single run, partition its processes into several groups, each one simulating a run of a smaller RB-system, \emph{as long as all the simulated runs have the same number of broadcasts}. In order to state and prove this result, we need the following.

\paragraph{Notation.}
Let $X \subseteq [n]$ be a set of processes.
For a configuration $f:[n] \to S$ of $\sysinst{n}$ define $\restr{f}{X}$ to be the restriction of $f$ to the domain $X$. Similarly, for a global transition $t$ of $\sysinst{n}$, say $\trans{f}{g}{\sigma}$, if $t$ is a broadcast transition (i.e., $\sigma=\brd$), or a rendezvous transition whose active processes are all in $X$ (i.e., $\sigma \neq \brd$ and $\activeprocs(t) \subseteq X$), then we define $\restr{t}{X}$ to be $\trans{\restr{f}{X}}{\restr{g}{X}}{\sigma}$; otherwise (i.e., if $\sigma \neq \brd$ and $\activeprocs(t) \not \subseteq X$), then $\restr{t}{X}$ is undefined.  Finally, given a path $\pi$ in $\sysinst{n}$, if for every $1 \leq i \leq |\pi|$ we have that $\activeprocs(\pi_i) \subseteq X$ or $\activeprocs(\pi_i) \subseteq [n] \setminus X$, then the restriction $\restr{\pi}{X} := \restr{\pi_{i_1}}{X} \restr{\pi_{i_2}}{X} \dots$ is defined by taking $i_1 < i_2 < \dots$ to be exactly the indices $1 \leq j \leq |\pi|$ for which $\restr{\pi_{j}}{X}$ is defined; otherwise (i.e., if there is a transition on $\pi$ in which some of the active processes are in $X$ and some are not in $X$)  $\restr{\pi}{X}$ is undefined.

We will implicitly rename processes as follows. Let $rename:X \to [|X|]$ be a bijection. Consider configurations $f$, transitions $t$, and paths $\pi$ of $\sysinst{n}$. By renaming the processes using $rename$ we can think of $\restr{f}{X}$ as a configuration of $\sysinst{|X|}$, and $\restr{t}{X}$ (if defined) as the transition of $\sysinst{|X|}$ obtained by restricting the configurations $f$ and $g$ in $t$ to $X$, and $\restr{\pi}{X}$ (if defined) as a path of $\sysinst{|X|}$.

For a process template $\proctemp$, paths $\pi_1, \pi_2, \cdots, \pi_h$ in $\psys$ (possibly using different numbers of processes), and pairwise disjoint subsets $X_1, X_2, \cdots, X_h$ of $\Nat$, we say that a path $\pi$ in $\psys$ \emph{simulates} $\pi_1, \cdots, \pi_h$ (with $X_1, X_2, \cdots, X_h$) if $\restr{\pi}{X_i} = \pi_i$ for every $i$.
Observe that, if $\pi_1, \cdots, \pi_h$ do not have the same number of broadcasts then there is no $\pi$ that can simulate them. The next lemma shows that this condition is not only necessary but also sufficient.

\begin{lemma}[Composition]\label{lem: rb-system composition}
Given an integer $b$, paths (resp. runs) $\pi_1, \dots, \pi_h$ in RB-systems $\sysinst{n_1}, \dots, \sysinst{n_h}$  each with exactly $b$ broadcast transitions: for every $n \geq \Sigma_{i=1}^h n_i = m$, every configuration $f$ in $\sysinst{n}$ and pairwise disjoint subsets $X_1, X_2, \cdots, X_h$ of $\Nat$ such that $\restr{f}{X_i} = \src(\pi_i)$ for every $i$, there exists a path (resp. run) $\pi$ in $\sysinst{n}$ starting in $f$ that simulates $\pi_1, \cdots, \pi_h$ with $X_1, X_2, \cdots, X_h$.
\end{lemma}

\begin{proof}
 We begin by proving the lemma in the special case of R-systems.
 For every $j \in [h]$, we will have the $n_j$ processes in the set $X_j$ simulate $\pi_j = e_{j,1}\ e_{j,2} \ldots$.
 The extra processes (between $m+1$ and $n$) do not move.
 Note that all transitions on $\pi_1, \dots, \pi_h$ are rendezvous involving $k$ processes. Whenever a rendezvous appearing on $\pi_j$ is performed in $\sysinst{n}$ only $k$ processes in $X_j$ move, leaving the others unaffected. Thus, $\pi$ can be obtained by any interleaving of the rendezvous appearing on $\pi_1, \dots, \pi_h$ as long as the relative internal ordering of rendezvous on each of these paths is maintained (e.g., round-robin $e_{1,1}\ e_{2,1} \dots e_{h,1}\ e_{1,2}\ e_{2,2} \dots e_{h,2} \dots$).

 Now, we consider the case of general RB-systems. As before, for every $j \in [h]$, we will have the $n_j$ processes in the set $X_j$ simulate $\pi_j$. If $n > m$, the extra processes are ignored (however, they do move whenever there is a broadcast). Each path $\pi_1, \dots, \pi_h$ is cut into $b+1$ segments (numbered $0, \ldots, b$), each containing only rendezvous transitions and followed by a broadcast transition. Thus, the $i$'th segment of each path is followed by the $(i+1)$'th broadcast. The path $\pi$ is constructed in $b+1$ phases: in phase $i$, the $i$'th segment of all the paths $\pi_1, \dots, \pi_h$ are simulated as was done in the R-systems case, followed (if $i < b$) by a single broadcast transition that forces the simulation of the $i$'th broadcast on all of these paths at once. 
\end{proof}

We now present a more flexible form of simulation in which the processes that are assigned to simulate a given path are not fixed throughout the simulation (this will be used in the proof of Theorem~\ref{thm: BSW correctness}).

\begin{definition}
We say that $\pi_0$ \emph{weakly-simulates} $\pi_1, \cdots, \pi_h$ if there exists an integer $l$ and a decomposition of each of these paths into $l$ segments, the $i$'th segment of $\pi_j$ is denoted $\pi_j^i$ for $1 \leq i \leq l, 0 \leq j \leq h$, and pairwise disjoint sets $X_1^i, \cdots, X_h^i$ for $1 \leq i \leq l$, such that for every $i$ we have that $\pi_0^i$ simulates $\pi_1^i, \cdots, \pi_h^i$ (with $X_1^i, X_2^i, \cdots, X_h^i$).
\end{definition}
The difference between weak-simulation and simulation is that the set of processes simulating each path may be changed at the end of each segment. The following observation follows immediately from the definition above.

\begin{remark} \label{rem: weakly-simulates}
If $\pi_0$ weakly-simulates \emph{cycles} $\pi_1, \cdots, \pi_h$, then $\restr{\dst(\pi_0)}{X_j^l} = \dst(\pi_j) = \src(\pi_j)$ for every $j$. In words: for every simulated cycle $\pi_j$, the destination configuration of the weakly simulating path $\pi_0$ restricted to the set of processes $X_j^l$ (the states used in simulating the last segment of $\pi_j$) is equal to the destination configuration of $\pi_j$ and thus, since $\pi_j$ is a cycle, also to its source configuration.
\end{remark}

\section{Discrete Timed Networks}

In this section we give the formal definition of a {\em discrete timed network}, with minor changes compared with~\cite{ADM04}. In particular, we first describe the form of a process template and later the operational semantics defining how networks of such processes evolve.
In this work, unless stated otherwise, we only consider timed networks \emph{without} a controller, and always assume a discrete time model $\NatZero$.

\begin{definition}
\label{def:tn}
A {\em timed-network (TN) template} is a tuple $\tpl{A,C,\grd,\rst,\CP}$ where $A = \tpl{\AP,\Actionprts,S,I,R,\lambda}$ is a finite LTS, $C$ is a finite set of \emph{clock variables} (also called \emph{clocks}),
each transition $t \in R$ is associated with a {\em guard} $\grd(t)$ and a {\em reset command} $\rst(t)$, and $\CP$ is a set of \emph{clock predicates}, i.e., predicates of the form $x \bowtie c$ where $x \in \clocks$, $c \in \NatZero$ is a constant, and ${\bowtie} \in \{>,=\}$. A guard is a Boolean combination of clock predicates. A reset command is a subset of $C$.
\end{definition}

{The \emph{size} of a TN template is the size of the LTS $A$ (i.e., the number of states plus the number of transitions) plus the sizes of all the guards, reset commands, and clock predicates where the constants in the clock predicates are represented in \emph{unary}.\footnote{The unary representation is chosen in order to elicit the relation to RB-systems, i.e., this representation allows us to show that the PMCP for timed-networks and RB-systems is polynomial-time inter-reducible.}

A timed network $\sysinsttimed{n}$ consists of the parallel composition of $n \in \Nat$ template processes, each running a copy of the template.
Each copy has a \emph{local configuration} $(q,K)$, where $q \in S$ and $K:\clocks \to \NatZero$ is a clock evaluation mapping each clock to its (discrete) value.
We say that an evaluation $K:C \to \NatZero$ \emph{satisfies} a Boolean combination of clock predicates $\phi$, if $\phi$ evaluates to true when every occurrence of clock $x$ in $\phi$ is replaced by the value $K(x)$.
A rendezvous action $\msg{a}$ is \emph{enabled} if there are $k$ processes $i_1, \cdots, i_k$ such that for every $j \in [k]$ process $i_j$ is
in a local configuration $(q_j,K_j)$  for which there is an edge $\trans{q_j}{q'_j}{\msg{a}_j}$, say $t_j$, and the clock evaluation $K_j$ satisfies the guard $\grd(t_j)$.
The rendezvous action is \emph{taken} means that the $k$ processes change their local configurations to $(q'_i,K'_i)$, where $K'_i$ is obtained from $K_i$ after setting the values of the clocks in $\rst(t_i)$ to $0$. Besides these rendezvous transitions, the system can evolve by taking timed-transitions in which all clocks of all processes advance by one time unit (so every $K(x)$ increases by one).\footnote{Alternatively, as in \cite{ADM04}, one can let time advance by any amount.} Runs of $T^n$ projected onto a single process induce sequences over the alphabet $2^{\AP \cup \CP}$ of the atomic predicates and clock predicates that hold at each local configuration. Specifications (for the behavior of a single process) can be given as automata or linear-temporal properties over the alphabet $2^{\AP \cup \CP}$.

To formally define a timed network as an LTS, its executions, and its corresponding PMCP, one can proceed by instantiating the intuitive description given above, along the lines of, e.g., \cite{ADM04}. Alternatively, one can give an equivalent definition (in the sense that it yields exactly the same LTS for the timed network, and thus also the same set of executions and PMCP) by observing that timed networks are essentially RB-systems whose RB-template $P_T$ is induced by the given TN-template $T$ by viewing local configurations as states of $P_T$, and thinking of timed transitions as symmetric broadcast transitions.
Notice that following this approach, the obtained RB-template $P_T$ would be infinite, due to clocks potentially increasing without bounds. In order to make it finite, one can simply truncate clock values up to an appropriate upper bound.
In the following we give a detailed construction.

\paragraph{Defining Timed systems as RB-systems.}\footnote{While Definition~\ref{dfn: rb-template} requires an RB-template to be finite, for the purpose of this section we lift this restriction.}
Let $T$ be a TN-template $\tpl{A,C,\grd,\rst,\CP}$ where $A = \tpl{\AP,\Actionprts,S,I,R,\lambda}$. Define the infinite RB-template
\[P_T = \tup{\AP \cup \CP,\Actionprts \cup \{\brd\}, S_T,I_T,R_T,\lambda_T}\]
where
\begin{itemize}
	\item $S_T = S \times \NatZero^C$ is the set of \emph{local configurations},
 \item $I_T$ consists of all pairs $(q,K)$ where $q \in I$ and $K(x) = 0$ for all $x \in C$,
 \item $R_T$ consists of two types of transitions:
 \begin{itemize}
    \item \emph{timed transitions} of the form $\trans{(q,K)}{(q,K')}{\brd}$, for every $q \in S$ and $K : C \to \NatZero$, and such that $K'(x) = K(x) + 1$ for all $x \in C$; or

    \item \emph{rendezvous transitions} of the form  $\trans{(q,K)}{(q',K')}{\sigma}$, for every transition $t = (q,\sigma,q') \in R$, for every evaluation $K : C \to \NatZero$ satisfying the guard $\grd(t)$, and such that for every $x \in C$, if $x \in \rst(t)$ then $K'(x) = 0$ and otherwise $K'(x) = K(x)$.

\end{itemize}
\item $\lambda_T \subseteq S_T \times (\AP\ \cup\ \CP)$ consists of all pairs $((q,K),p)$ such that either $p \in \AP$ and $p \in \lambda(q)$, or $p \in \CP$ and $K$ satisfies the clock
predicate $p$.
\end{itemize}

Given a TN-template $T = \tpl{A,C,\grd,\rst,\CP}$ with $A = \tpl{\AP,\Actionprts,S,I,R,\lambda}$, and $n \in \Nat$, we define the {\em timed network $\sysinsttimed{n}$}, composed of $n$ processes, to be the RB-system $(P_T)^n$, and the \emph{timed network $\psystimed$} to be the RB-system $(P_T)^\infty$.

{
\begin{definition}[\bf PMCP for Timed-Networks]
Let $\Pspec$ be a specification formalism for sets of infinite (resp. finite) words over the alphabet $2^{\AP \cup \CP}$.
The {\em Parameterized Model Checking Problem for Timed-Networks} for $\Pspec$, denoted $PMCP(\Pspec)$, is to decide, given a
TN-template $T$, and a set $L$ of infinite (resp. finite) words specified in $\Pspec$, if all executions in the set $\execinf[T^\infty]$ (resp. $\execfin[T^\infty]$) are in $L$.
\end{definition}
}

In Sections~\ref{sec:solving finite} and~\ref{sec:solving infinite} we show how to solve PMCP for finite RB-templates for specifications of finite and infinite executions respectively.
This cannot be used directly to solve the PMCP for timed networks since given a timed template $T$ the RB-template $P_T$ is infinite. However, the next Lemma shows that given $T$, there is a finite RB-template $U$ such that $\exec[T^\infty] = \exec[U^\infty]$.  The template $U$ is obtained from $P_T$ by clipping the clock values to be no larger than $1$ plus the maximal constant appearing in the clock predicates $\CP$.

\begin{lemma}
\label{lem:timed-networks-limited}
Let $T$ be a TN-template and let $d = \max\{c : x \bowtie c \in \CP\}+1$.
One can construct in time polynomial the size of $T$ a finite RB-template $U$ such that $\exec[T^\infty] = \exec[U^\infty]$.
\end{lemma}
\begin{proof}
We use the following clipping operation: for $d \in \Nat$ and $K:C \to \NatZero$ let $clip_d(K):C \to \{ 0, \ldots, d \}$ map $x$ to $\min\{K(x),d\}$. For a local configuration $(q,K)$ define $clip_d(q,K)$ to be $(q,clip_d(K))$, and extend this to sets of configurations point-wise. Let $t$ be any transition $\trans{(q,K)}{(q',K')}{\sigma}$, define $clip_d(t)$ to be $\trans{clip_d(q,K)}{clip_d(q',K')}{\sigma}$, and extend this to sets of transitions point-wise.

Note that, by our choice of $d$, an evaluation $K$ satisfies a Boolean combination of clock predicates $\phi$ iff the evaluation $clip_d(K)$ satisfies $\phi$.

Let $T$ be a TN-template $\tpl{A,C,\grd,\rst,\CP}$ where $A = \tpl{\AP,\Actionprts,S,I,R,\lambda}$, and $P_T = \tup{\AP \cup \CP,\Actionprts \cup \{\brd\}, S_T,I_T,R_T,\lambda_T}$.
Then let
\[ U = \tup{\AP \cup \CP, \Actionprts, S',I',R',\lambda'}\]
where
\begin{itemize}
\item $S' = clip_d(S_T)$,
\item $I' = clip_d(I_T)$,
\item $R' = clip_d(R_T)$, and
\item $\lambda' = \{((q,clip_d(K)),p) : ((q,K),p) \in \lambda_T\}$.
\end{itemize}
We claim that
$\exec[T^\infty] = \exec[U^\infty]$. To see this note that $P_T$ and $U$ are bisimilar using the relation $B$ defined by letting
$((q,K),(q',K')) \in B$ iff $(q',K') = clip_d(q,K)$. It is not hard to see by following the definitions that $B$ is a bisimulation relation. To finish apply
Lemma~\ref{lem:bisim} item \ref{lem:bisim:bisim}. 
\end{proof}

{The construction used in Lemma~\ref{lem:timed-networks-limited} is illustrated in Figure~\ref{fig:ex_dtn_rbs}.
We note that the polynomial-time result crucially depends on constants represented in unary (note that the construction polynomially depends on $d = \max\{c : x \bowtie c \in \CP\}+1$).
We leave the investigation of complexity-theoretic consideration when numbers are represented in binary for future work.}

\begin{figure}[tb]
\begin{minipage}{0.29\linewidth}
\begin{tikzpicture}[->,>=stealth',shorten >=1pt,auto,node distance=2cm,
                    semithick]

  \node[initial,state] (P)                  {$p$};
  \node[state]         (Q) [below of=P]     {$q$};
  \node[state]         (W) [below of=Q]     {$r$};

  \path (P) edge    node[align=left] {$\msg{a_1}$\\$x:=0$} (Q)
        (P) edge [loop above] node {$\msg{a_2}$} (P)
        (Q) edge    node[below, pos=0.1, xshift=-15,align=right] {$\msg{a_1}$\\$x \le 2$} (W)
        (Q) edge[bend left=30]    node[align=right] {$\msg{a'_2}$\\$x > 2$} (P)
        (W) edge[bend right=75]    node[pos=0.2] {$\msg{a'_1}$} (P);
\end{tikzpicture}
\end{minipage}
\begin{minipage}{0.69\linewidth}

\begin{tikzpicture}[->,>=stealth',shorten >=1pt,auto,node distance=2cm,
                    semithick]

  \node[initial,state] (P0)                  {$p,0$};
  \node[state]         (Q0) [below of=P0]     {$q,0$};
  \node[state]         (W0) [below of=Q0]     {$r,0$};

  \node[state]         (P1) [right of=P0]                 {$p,1$};
  \node[state]         (Q1) [below of=P1]     {$q,1$};
  \node[state]         (W1) [below of=Q1]     {$r,1$};

  \node[state]         (P2) [right of=P1]                 {$p,2$};
  \node[state]         (Q2) [below of=P2]     {$q,2$};
  \node[state]         (W2) [below of=Q2]     {$r,2$};

  \node[state]         (P3) [right of=P2]                 {$p,3$};
  \node[state]         (Q3) [below of=P3]     {$q,3$};
  \node[state]         (W3) [below of=Q3]     {$r,3$};

  \path (P0) edge node[below,xshift=-10] {$\msg{a_1}$} (Q0)
        (P0) edge [loop above] node {$\msg{a_2}$} (P0)
        (Q0) edge node[below,pos=0.2,xshift=-10] {$\msg{a_1}$} (W0)
        (W0) edge[bend right=60] node[pos=0.2] {$\msg{a'_1}$} (P0);

  \path (P0) edge node[above,pos=0.3] {$\brd$} (P1)
        (Q0) edge node[below,pos=0.3] {$\brd$} (Q1)
        (W0) edge node[below,pos=0.3] {$\brd$} (W1);

  \path (P1) edge node[pos=0.2,below,xshift=-5,yshift=10] {$\msg{a_1}$} (Q0)
        (P1) edge [loop above] node {$\msg{a_2}$} (P1)
        (Q1) edge node[below,pos=0.2,xshift=-10] {$\msg{a_1}$} (W1)
        (W1) edge[bend right=60] node[pos=0.2] {$\msg{a'_1}$} (P1);

  \path (P1) edge node[above,pos=0.3]{$\brd$} (P2)
        (Q1) edge node[below,pos=0.3]{$\brd$} (Q2)
        (W1) edge node[below,pos=0.3]{$\brd$} (W2);

  \path (P2) edge node[pos=0.15,below,xshift=-5,yshift=10] {$\msg{a_1}$} (Q0)
        (P2) edge [loop above] node {$\msg{a_2}$} (P2)
        (Q2) edge node[below,pos=0.2,xshift=-10] {$\msg{a_1}$} (W2)
        (W2) edge[bend right=60] node[pos=0.2] {$\msg{a'_1}$}   (P2);

  \path (P2) edge node[above,pos=0.3]{$\brd$} (P3)
        (Q2) edge node[below,pos=0.3]{$\brd$} (Q3)
        (W2) edge node[below,pos=0.3]{$\brd$} (W3);

  \path (P3) edge node[pos=0.1,below,xshift=-5,yshift=10] {$\msg{a_1}$} (Q0)
        (P3) edge [loop above] node {$\msg{a_2}$} (P3)
        (Q3) edge[bend left=33] node[below, xshift=10,yshift=5] {$\msg{a'_2}$} (P3)
        (W3) edge[bend right=63] node[pos=0.2] {$\msg{a'_1}$}  (P3);

  \path (P3) edge[loop right] node{$\brd$} (P3)
        (Q3) edge[loop right] node[xshift=-2pt]{$\brd$} (Q3)
        (W3) edge[loop right] node{$\brd$} (W3);

\end{tikzpicture}

\end{minipage}
\caption{\label{fig:ex_dtn_rbs}Construction from Lemma~\ref{lem:timed-networks-limited}. A TN-template $T$ (left) with one clock $x$, and the RB-template $U$ (right) with $d = 3$. For readability, atomic predicates and clock predicates are not drawn.}
\end{figure}
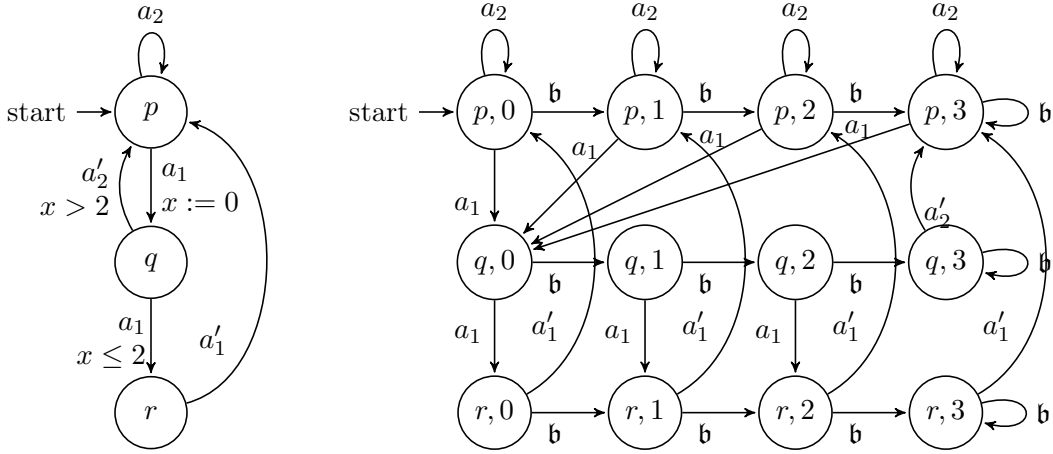

We next show that RB-systems are not more powerful than timed networks.
We show that, by allowing for operations that take a subsequence and remove atomic propositions, RB-systems and timed networks can define the same languages of (finite and infinite) executions. Recall the notation $(\xi)_a$ and $\xi|_\AP$ from Section~\ref{sec:transition systems}.

\begin{lemma}
\label{lem:rb-systems-timed-networks-reduction}
Let $P$ be a process template with atomic propositions $\AP$.
Then, one can construct in linear time a TN-template $T$ with a singleton set of clocks $C = \{c\}$, clock predicates $\CP = \{c=0,c=1\}$, and atomic proposition $\AP \cup \CP$, such that
$\exec[P^\infty] = \{ (\pi)_{c=0}|_\AP \mid \pi \in \exec[T^\infty] \}$.
\end{lemma}
\begin{proof}
We consider the process template $P = \tpl{\AP,\Actionprts  \cup \{\brd\},S,I,R,\lambda}$ and construct the TN-template $T = \tpl{A,C,\grd,\rst,\CP}$,
where $C = \{c\}$ and $\CP = \{c=0,c=1\}$, and the LTS $A = \tpl{\AP,\Actionprts  \cup \{\brd, \sharp\},S,I,R_A,\lambda}$ is obtained from $P$ as follows:
\begin{enumerate}
\item  $R_A$ contains every rendezvous transition $t$ of $P$, where we set $\grd(t):= c=0$ for the guard and $\rst(t) := \{\}$ for the reset.
\item $R_A$ contains an internal transition $\trans{s}{s'}{\sharp}$ for every broadcast transition $\trans{s}{s'}{\brd}$ of $P$, where we set $\grd(t):= c=1$ for the guard, and $\rst(t) := \{c\}$ for the reset (internal transitions are $1$-rendezvous transitions where a single process changes state).
\end{enumerate}

We claim that $\exec[\psys] = \{ (\xi)_{c=0}|_\AP \mid \xi \in \exec[T^\infty] \}$. To show that, we will need the following notation:
given a configuration $f$ of $\sysinst{n}$, for some $n$, and a clock value $v \in \{0,1\}$, we denote by $f^v$ the configuration of $T^n$ obtained from $f$ by defining the clock value of each of the processes to be $v$ (i.e., {by letting $f^v(j) = (f(j), (v))$, } for every $j \in [n]$). Also, it will be convenient to call a transition $\trans{f}{g}{(j,\sharp)}$ of $T^n$ an internal transition (of process $j$).

We first show that $\exec[\psys] \subseteq \{ (\xi)_{c=0}|_\AP \mid \xi \in \exec[T^\infty] \}$.
Given $\pi' \in \runs(\psys)$, let $n$ be such that $\pi' \in \runs(\sysinst{n})$.
We will construct a corresponding run $\pi \in \runs(T^n)$.
Intuitively, a rendezvous transition of $\pi'$ is simulated directly by a single corresponding transition in $\pi$ whose source and destination have all clocks at zero; whereas a broadcast transition  of $\pi'$ is simulated in $\pi$ by a sequence of transitions: one timed transition (that increases all clocks from zero to one), followed by one internal transition of each process. Observe that, by incrementing the clocks to one, the timed transition enables the guards of the internal transitions of each of the processes, and that each such internal transition, once taken, resets the clock back to zero.

Formally, we construct $\pi$ by considering the transitions of $\pi'$ in order. Let $i' \geq 1$, assume that we already constructed a prefix of length $i-1$ of $\pi$ that simulates the first $i'-1$ transitions of $\pi'$ (initially, this prefix is obviously empty). The construction will maintain the invariant $\dst(\pi'_{i'})^0 = \dst(\pi_i)$; here we use the notation $f^v$ introduced above. Consider the two cases for the transition $\pi'_{i'}$.
(i) if $\pi'_{i'} = (f',\sigma,g')$ is a rendezvous transition, then extend $\pi$ by letting $\pi_i = (f'^0, \sigma, g'^0)$.
(ii) if $\pi'_{i'} = (f',\brd,g')$ is a broadcast transition, then extend $\pi$ as follows: first add a timed transition $\pi_i = (f'^0, \brd, f'^1)$; then add a series of $n$ transitions, one for each process $j \in [n]$, in which process $j$ takes the internal transition from $f'(j)$ to $g'(j)$ (i.e., let $f_i = f'^1$, and let $\pi_{i+j} = (f_{i+j-1}, ((j,\sharp)), f_{i+j})$, where $f_{i+j}$ is identical to $f_{i+j-1}$ except that $f_{i+j}(j) = (g'(j), 0)$); Observe that at the end of this sequence we reach the configuration $g'^0$.
It is easy to verify that $\lambda(\proj{\pi'}(i)) = (\lambda(\proj{\pi}(i)))_{c=0}|_\AP$ for all processes $i$.
Hence, $\exec[P^\infty] \subseteq \{ (\pi)_{c=0}|_\AP \mid \pi \in \exec[T^\infty] \}$.

We now show the reverse inclusion $\{ (\xi)_{c=0}|_\AP \mid \xi \in \exec[T^\infty] \} \subseteq \exec[\psys]$.
Given $\pi \in \runs(T^\infty)$, let $n$ be such that $\pi \in \runs(T^n)$.
We will construct a corresponding run $\pi' \in \runs(\sysinst{n})$. Unlike the previous inclusion, this time our construction will be tailored to faithfully simulate the behavior of only a single process, i.e., $\pi'$ will depend on the process of interest. Recall from Section~\ref{sec:executions} that, due to symmetry, we can restrict our attention to process $1$, i.e., it is enough to show that $\dagger$: $\lambda(\proj{\pi'}(1)) = (\lambda(\proj{\pi}(1)))_{c=0}|_\AP$.

We first claim that, w.l.o.g., $\pi$ is structured as a sequence of blocks of the form: one or more multi-process rendezvous transitions, followed (unless $\pi$ has no more timed transitions) by a timed transition and $1 \leq k \leq n$ internal transitions of $k$ different processes, one of which is process $1$. To see why this claim holds, let $i_1 < i_2 < \cdots$ be the positions of timed transitions of $\pi$. Observe that, for every $h$, all the clocks in $\dst(\pi_{i_h})$ are at least $1$ and thus, due to the way the transitions of $T$ are guarded, for every $j \in [n]$, the first transition in which process $j$ can be active after $\pi_{i_h}$ and before $\pi_{i_{h+1}}$ must be an internal transition of $j$ (in which, obviously, no other process is active). Hence, we can push all the internal transitions to immediately follow the timed transition $\pi_{i_h}$ without changing the projection of $\pi$ on any single process. Finally, observe that if process $1$ did not make an internal transition between $\pi_{i_h}$ and $\pi_{i_{h+1}}$ then we can safely trim $\pi$ and ignore everything after position $i_h$. Indeed, in this case at $\dst(\pi_{i_{h+1}})$ process $1$ has a clock value $> 1$, and thus (since all non-timed transitions are guarded by $c=0$ or $c=1$), this value will never go below $1$ again. Now, recall that $\dagger$ restricts our attention to the projection of $\pi$ on process $1$ transitions whose sources have clock value $0$.

It should be now quite clear how to simulate each block of $\pi$: every multi-process rendezvous transition can be easily simulated by a corresponding rendezvous transition in $\sysinst{n}$, whereas the rest of the block (i.e., the timed transition followed by a sequence of internal transitions) can be simulated by a single broadcast transition of $\sysinst{n}$. Note, however, that processes which do not make an internal transition in the block will not be faithfully simulated since such processes never change their state in this block (only their clock changes during by the timed transition), whereas the simulating broadcast may unfaithfully change their state.
However, as we noted above, these processes will not include process $1$, nor can they rendezvous with any process (again, because their guards will always be false), and thus will have no influence on the ability of the projection of $\pi'$ on process $1$ to faithfully capture the projection of $\pi$ on process $1$.

Given a configuration $f$ of $T^n$, denote by $state(f) \in S^n$, the configuration in $\sysinst{n}$ obtained from it by dropping the clock values of all processes.
Formally, we construct $\pi'$ by considering the transitions of $\pi$ in order, block by block. Let $i \geq 1$, and assume that we already constructed a prefix of length $i'-1$ of $\pi'$ that simulates the first $i-1$ transitions of $\pi$. Whenever $\pi_i$ is a multi-process rendezvous transition, or is the last transition in a block, the construction will maintain the invariant that the state of process $j$ in $\dst(\pi_i)$  is equal to its state in $\dst(\pi'_{i'})$, for every process $j$ whose clock value at $\dst(\pi_i)$ is $0$. Consider the following two cases:
(i) $\pi_i = (f,\sigma,g)$ is a multi-process rendezvous transition. Then, extend $\pi'$ by letting $\pi'_{i'} = (state(f), \sigma, state(g))$.
(ii) $\pi_i$ is a timed transition (let $f' = state(\dst(\pi_i))$), and the rest of the transitions in the block are $k \leq n$ internal transitions of different processes, i.e., there is a set of processes $X = \{x_1, x_2, \cdots, x_k\} \subseteq [n]$, and for every $j \in [k]$ there is an internal transition $\trans{f'(j)}{q_j}{\sharp}$ of the template $T$, such that $\pi_{i+j} = \trans{f_{i+j-1}}{f_{i+j}}{(x_j,\sharp)}$, where $f_i = \dst(\pi_i)$, and $f_{i+j}$ is identical to $f_{i+j-1}$ except that $f_{i+j}(x_j) = (q_j, 0)$. Then, extend $\pi'$ by letting $\pi'_{i'} = (f', \brd, g')$ be a transition of $\sysinst{n}$ satisfying $g'(x_j) = q_j$ for every $j \in [k]$ (note that the exact broadcast transitions taken by processes outside $X$ are unimportant).
It is not hard to verify that $\dagger$ holds, as promised.
\end{proof}

We note that Lemma~\ref{lem:rb-systems-timed-networks-reduction} implies that for every specification $\varphi$ (given as LTL formula resp. finite automaton) one can construct a specification $\varphi'$ such that all executions $\execinf[P^\infty]$ (resp. $\execfin[P^\infty]$) satisfy $\varphi$ iff all executions $\execinf[T^\infty]$ (resp. $\execfin[T^\infty]$) satisfy the specification $\varphi'$; this is because the operation $(\pi)_a$ can be implemented as a formula resp. automaton transformation.

{
From Lemma~\ref{lem:timed-networks-limited} and Lemma~\ref{lem:rb-systems-timed-networks-reduction} we immediately obtain the main result of this section:
\begin{theorem} \label{thm: RB and Timed}
The Parameterized Model Checking Problems for RB-Systems and Timed-Networks are polynomial-time inter-reducible;
in particular, (lower as well as upper) bounds on the program/specification/combined complexity transfer.
\end{theorem}
}

\paragraph{Timed Networks with a Controller.}
The inter-reducibility between RB-Systems and Timed Networks extends to systems with a controller as we sketch in the following.
Given two TN-templates $T_C$ and $T$, the TN-System with a controller (TNC-System) $T_C \cup T^n$ is then defined as the RBC-System $\proctemp_{T_C} \cup \proctemp_{T}^n$, where $\proctemp_{T_C}$ and $\proctemp_T$ are the induced (infinite) RB-templates introduced above.
The TNC-System $T_C \cup T^\infty$ is then defined analogously.
Lemma~\ref{lem:timed-networks-limited} and Lemma~\ref{lem:rb-systems-timed-networks-reduction} can then straightforwardly be extended to the relationship between TNC-Systems and RBC-Systems.
This gives us the following results:
\begin{theorem}
The Parameterized Model Checking Problems for RBC- and TNC-Systems are polynomial-time inter-reducible;
in particular, (lower as well as upper) bounds on the program/specification/combined complexity transfer.
\end{theorem}
\begin{corollary} \hfill
\begin{enumerate}
\item Let $\Pspec$ be specifications of sets of infinite executions expressed as NBW or LTL formulas.
    Then $PMCP(\Pspec)$ of TNC-Systems is undecidable.
\item Let $\Pspec$ be specifications of sets of finite executions expressed as NFW or LTLf formulas. Then $PMCP(\Pspec)$ of TNC-Systems is decidable.
\end{enumerate}
\end{corollary}

\section{Solving PMCP for Specifications over Finite Executions} \label{sec:solving finite} \label{sec: unwinding} 

In this section we solve the PMCP problem for specifications given as nondeterministic finite word automata (NFW), and prove that it is \PSPACE-complete. Following the automata-theoretic approach outlined in Section~\ref{subsec: PMCP}, given an RB-template $\proctemp$ we will build an NFW $\execNFW$ that accepts exactly the executions in $\execfin$.  Model checking of a regular specification given by an NFW $\A'$ is thus reduced to checking containment of the language of $\execNFW$ in that of $\A'$. The structure of $\execNFW$ is based on the reachability-unwinding of $\proctemp$, which we now introduce.

Given a template $\proctemp = \tup{AP, \Actionprts \cup \{\brd\}, S,I,R,\lambda}$, we show how to construct a new process template $\puwd = \tup{AP, \Actionprts \cup \{\brd\}, S^\uwd,I^\uwd,R^\uwd,\lambda^\uwd}$, called the {\em reachability-unwinding}
of $\proctemp$, see Figure~\ref{fig:lasso}. Intuitively, $\puwd$ is obtained by alternating the following two operations: \emph{(i)} taking a copy of $\proctemp$ and removing from it all the unreachable rendezvous edges, i.e., all transitions of $P$ that cannot be taken by any process in any run of $\psys$; and \emph{(ii)} unwinding on broadcast edges.
This is repeated until a copy is created which is equal to a previous one; we then stop and close the unwinding back into the old copy, forming a high-level lasso structure.

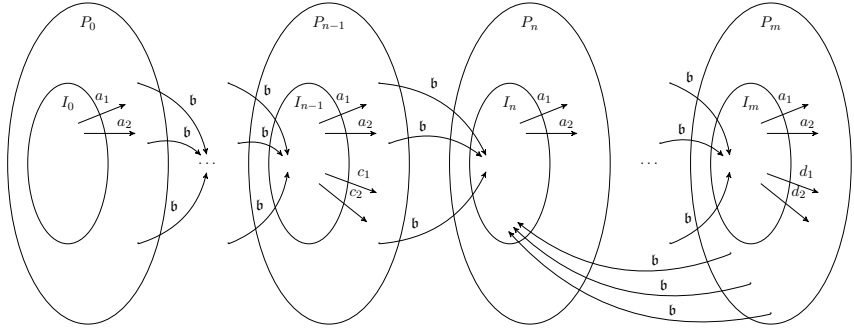
\begin{figure}[t]
	\resizebox{0.8\linewidth}{!}{             
		\begin{tikzpicture}[->,>=stealth',shorten >=2pt,auto,
                    semithick]

    \tikzstyle{invisible}=[]
    \tikzstyle{vertex}=[circle,fill=black!25]


    \draw (0,0) ellipse (2cm and 4cm);
    \node[invisible]    at (0,3.5)    {$P_0$};
    
    \draw (-0.5,0) ellipse (1cm and 2cm);
    \node[invisible]    at (-0.5,1.5)    {$I_0$};

    \node[invisible]    (P0btarget) at (3,0) {\ldots};
    \node[invisible]    (P1btarget) at (5,0) {};

    \node[invisible]    (I01) at (-0.25,1)      {};
    \node[invisible]    (I02) at (-0.1,0.75)    {};
    \node[invisible]    (I03) at (-0.1,-0.25)   {};
    \node[invisible]    (I04) at (-0.25,-0.5)   {};

    \node[invisible]    (P01) at (1,1.5)        {};
    \node[invisible]    (P02) at (1.25,0.75)    {};
    \node[invisible]    (P03) at (1.25,-0.75)   {};
    \node[invisible]    (P04) at (1,-1.5)       {};

    \path (I01.center) edge node [pos=0.75] {$a_1$} (P01.center);
    \path (I02.center) edge node [pos=0.75] {$a_2$} (P02.center);


    \node[invisible]    (P05) at (1.25,2)      {$\cdot$};
    \node[invisible]    (P06) at (1.5,0.5)      {$\cdot$};
    \node[invisible]    (P07) at (1.25,-2)      {$\cdot$};

    \path (P05.center) edge [bend left] node {$\brd$} (P0btarget);
    \path (P06.center) edge  [bend left] node  {$\brd$} (P0btarget);
    \path (P07.center) edge [bend right] node {$\brd$} (P0btarget);

    \node[invisible]    (P08) at (3.5,2)      {$\cdot$};
    \node[invisible]    (P09) at (3.75,0.5)      {$\cdot$};
    \node[invisible]    (P0A) at (3.5,-2)      {$\cdot$};

    \path (P08.center) edge [bend left] node [pos=0.33] {$\brd$} (P1btarget);
    \path (P09.center) edge  [bend left] node [pos=0.33]  {$\brd$} (P1btarget);
    \path (P0A.center) edge [bend right] node [pos=0.5] {$\brd$} (P1btarget);


    \draw (6,0) ellipse (2cm and 4cm);
    \node[invisible]    at (6,3.5)    {$P_{n-1}$};
    
    \draw (5.5,0) ellipse (1cm and 2cm);
    \node[invisible]    at (5.5,1.5)    {$I_{n-1}$};

    \node[invisible] (P2btarget)    at (10,0) {};

    \node[invisible]    (I01) at (5.75,1)      {};
    \node[invisible]    (I02) at (5.9,0.75)    {};
    \node[invisible]    (I03) at (5.9,-0.25)   {};
    \node[invisible]    (I04) at (5.75,-0.5)   {};

    \node[invisible]    (P01) at (7,1.5)        {};
    \node[invisible]    (P02) at (7.25,0.75)    {};
    \node[invisible]    (P03) at (7.25,-0.75)   {};
    \node[invisible]    (P04) at (7,-1.5)       {};

    \path (I01.center) edge node [pos=0.75] {$a_1$} (P01.center);
    \path (I02.center) edge node [pos=0.75] {$a_2$} (P02.center);
    \path (I03.center) edge node {$c_1$} (P03.center);
    \path (I04.center) edge node {$c_2$} (P04.center);


    \node[invisible]    (P15) at (7.25,2)      {$\cdot$};
    \node[invisible]    (P16) at (7.5,0.5)      {$\cdot$};
    \node[invisible]    (P17) at (7.25,-2)      {$\cdot$};

    \path (P15.center) edge [bend left] node [pos=0.33] {$\brd$} (P2btarget);
    \path (P16.center) edge  [bend left] node [pos=0.33]  {$\brd$} (P2btarget);
    \path (P17.center) edge [bend right] node [pos=0.33] {$\brd$} (P2btarget);


    \draw (11,0) ellipse (2cm and 4cm);
    \node[invisible]    at (11,3.5)    {$P_n$};
    
    \draw (10.5,0) ellipse (1cm and 2cm);
    \node[invisible]    at (10.5,1.5)    {$I_n$};

    \node[invisible]    at (14,0) {\ldots};
    \node[invisible] (Pnbtarget) at (16,0) {};

    \node[invisible]    (I01) at (10.75,1)      {};
    \node[invisible]    (I02) at (10.9,0.75)    {};
    \node[invisible]    (I03) at (10.9,-0.25)   {};
    \node[invisible]    (I04) at (10.75,-0.5)   {};

    \node[invisible]    (P01) at (12,1.5)        {};
    \node[invisible]    (P02) at (12.25,0.75)    {};
    \node[invisible]    (P03) at (12.25,-0.75)   {};
    \node[invisible]    (P04) at (12,-1.5)       {};

    \path (I01.center) edge node [pos=0.75] {$a_1$} (P01.center);
    \path (I02.center) edge node [pos=0.75] {$a_2$} (P02.center);


    \node[invisible]    (Pn5) at (14.5,2)      {$\cdot$};
    \node[invisible]    (Pn6) at (14.25,0.5)      {$\cdot$};
    \node[invisible]    (Pn7) at (14.5,-2)      {$\cdot$};

    \path (Pn5.center) edge [bend left] node [pos=0.1] {$\brd$} (Pnbtarget);
    \path (Pn6.center) edge  [bend left] node [pos=0.3]  {$\brd$} (Pnbtarget);
    \path (Pn7.center) edge [bend right] node [pos=0.3] {$\brd$} (Pnbtarget);
    

    \draw (17,0) ellipse (2cm and 4cm);
    \node[invisible]    at (17,3.5)    {$P_m$};
    
    \draw (16.5,0) ellipse (1cm and 2cm);
    \node[invisible]    at (16.5,1.5)    {$I_m$};

    \node[invisible] (Pmbtarget1) at (10.5,-1.3) {};
    \node[invisible] (Pmbtarget2) at (10.4,-1.4) {};
    \node[invisible] (Pmbtarget3) at (10.3,-1.5) {};

    \node[invisible]    (I01) at (16.75,1)      {};
    \node[invisible]    (I02) at (16.9,0.75)    {};
    \node[invisible]    (I03) at (16.9,-0.25)   {};
    \node[invisible]    (I04) at (16.75,-0.5)   {};

    \node[invisible]    (P01) at (18,1.5)        {};
    \node[invisible]    (P02) at (18.25,0.75)    {};
    \node[invisible]    (P03) at (18.25,-0.75)   {};
    \node[invisible]    (P04) at (18,-1.5)       {};

    \path (I01.center) edge node [pos=0.75] {$a_1$} (P01.center);
    \path (I02.center) edge node [pos=0.75] {$a_2$} (P02.center);
    \path (I03.center) edge node {$d_1$} (P03.center);
    \path (I04.center) edge node {$d_2$} (P04.center);


    \node[invisible]    (Pm5) at (16,-2.25)      {$\cdot$};
    \node[invisible]    (Pm6) at (16.5,-3)      {$\cdot$};
    \node[invisible]    (Pm7) at (17,-3.75)      {$\cdot$};

    \path (Pm5.center) edge [bend left] node [above,pos=0.33] {$\brd$} (Pmbtarget1);
    \path (Pm6.center) edge  [bend left] node [above,pos=0.33]  {$\brd$} (Pmbtarget2);
    \path (Pm7.center) edge [bend left] node [above,pos=0.33] {$\brd$} (Pmbtarget3);

\end{tikzpicture}
	}
	\caption{\label{fig:lasso}A high level view of the reachability-unwinding lasso.}
\end{figure}

\paragraph{Intuition.} Technically, it is more convenient to first calculate all the desired copies and then to arrange them in the lasso. Thus, for $0 \leq i \leq m$ (for an appropriate $m$), we first compute an R-template $\proctemp_i = \tup{AP,\Actionprts,S_i,I_i,R_i,\lambda_i}$ which is a copy of $\proctemp$ with initial states redesignated and all broadcast edges, plus some rendezvous edges, removed.
Second, we take $\proctemp_0, \dots, \proctemp_m$ and combine them, to create the single process template $\puwd$. We do this by connecting, for $i < m$, the states in $\proctemp_i$ with the initial states of $\proctemp_{i+1}$ by means of broadcast edges, as induced by transitions in $\proctemp$. In case $i=m$, then $\proctemp_i$ is connected to the copy $\proctemp_n$, for some $n \leq m$, as determined by the lasso structure. 

\paragraph{Constructing $P_i$ via a Saturation Algorithm.} Recursively construct the R-template 
\[ \proctemp_i = \tup{AP,\Actionprts,S_i,I_i,R_i,\lambda_i} 
\] 
(called the $i$'th {\em component} of $\puwd$) as follows. For $i=0$, we let $I_0 := I$; and for $i > 0$ we let $I_i := \{ s \in S \mid (h,\brd,s) \in R \text{ for some } h \in S_{i-1} \}$ be the set of states reachable from $S_{i-1}$ by a broadcast edge. The sets $S_i$ and $R_i$ are obtained using the following {\em saturation} algorithm:
start with $S_i := I_i$ and $R_i := \emptyset$; at each round of the algorithm, consider in turn each edge $e \in R \setminus R_i$ of the form $\trans{s}{t}{\msg{a}_h}$ such that $s \in S_i$: if for every $l \in [k] \setminus \{ h \}$ there is some other edge $\trans{s'}{t'}{\msg{a}_l}$ with $s' \in S_i$, then add $e$ to $R_i$ and add $t$ (if not already there) to $S_i$. The algorithm ends when a fixed-point is reached. Finally, let $\lambda_i$ be the restriction of $\lambda$ to $S_i \times AP$. Observe the following property of this algorithm: 
\begin{remark}
\label{rem:saturation}
If $\trans{s}{t}{\msg{a}_h}$ in $R_i$ then for all $l \in [k]$ there exist $s',t' \in S_i$ such that $\trans{s'}{t'}{\msg{a}_l}$ in $R_i$.
\end{remark}

Now, $\proctemp_i$ is completely determined by $I_i$ (and $\proctemp$), and so there are at most $2^{|S|}$ possible values for it. Hence, there must exist $n,m$ with $n \leq m < 2^{|S|}$ such that $\proctemp_n = \proctemp_{m+1}$. We stop calculating $\proctemp_i$'s when this happens since for every 
$i \geq n$ it must be that $\proctemp_i = \proctemp_{n + ((i-n)\mod r)}$, where $r=m+1-n$. We call $n$ the {\em prefix length} of $\puwd$ and call $r$ its {\em period}, i.e., $n$ is the number of components on the prefix of the lasso and $r$ is the number of components on the noose of the lasso. For $i \in \NatZero$, let $\comp(i)$ denote the associated component number of $i$, i.e., the position in the lasso after $i$ moves between components. Formally, $comp(i) = \min(i,n + ((i-n) \mod r))$.

We now construct $\puwd$.

\begin{definition}[Reachability-unwinding] \label{dfn: unwinding}
Given $\proctemp_0, \dots, \proctemp_m$, define the RB-template
\[ \puwd = \tup{AP,\Actionprts \cup \{\brd\},S^\uwd,I^\uwd,R^\uwd,\lambda^\uwd}\]
as follows:
\begin{itemize}
 \item $S^\uwd := \cup_{i=0}^m \{ (s, i) \mid s \in S_i \}$;
 \item $I^\uwd := \{ (s,0) \mid s \in I \}$;
 \item $R^\uwd$ contains the following transitions:
 \begin{itemize}
    \item the rendezvous transitions $\cup_{i=0}^m \{ \trans{(s,i)}{(t,i)}{\varsigma} \mid \trans{s}{t}{\varsigma} \in R_i \}$, and

    \item the broadcast transitions  $\cup_{i=0}^{m-1} \{ \trans{(s,i)}{(t,i+1)}{\brd} \mid  \trans{s}{t}{\brd} \in R\ \text{and}\ s \in S_i  \}$ and $\{ \trans{(s,m)}{(t,n)}{\brd} \mid  \trans{s}{t}{\brd} \in R\ \text{and}\ s \in S_m \}$.
\end{itemize}
\item $\lambda^\uwd = \cup_{i=0}^m \{ ((s,i),p) : (s,p) \in \lambda_i \}$.
\end{itemize}
\end{definition}

For $0 \leq i \leq m$, we denote by $\puwd_i$ the restriction of $\puwd$ to the states in $\{ (s, i) \mid s \in S_i \}$ and the rendezvous transitions between them. We call $\puwd_i$ the \emph{$i$'th component of $\puwd$}. Observe that $\puwd_i$ can be written as an R-template $\puwd_i = \tup{AP,\Actionprts \cup \{\brd\},S^\uwd_i,I^\uwd_i,R^\uwd_i,\lambda^\uwd_i}$ which is obtainable from the component $\proctemp_i$ by simply attaching $i$ to every state. We will sometimes find it convenient to speak of the component $\puwd_i$, for $i$ greater than $m$, in which case we identify $\puwd_i$ with $\puwd_{comp(i)}$.
We say that a configuration $f$ of $\puwdsys$ is \emph{in $\puwd_i$} iff all its processes are in states of the component $\puwd_i$, i.e., iff $f(j) \in S_i^\uwd$ for all $j$ in the domain of $f$.

\begin{example}
\hspace{1cm}
\begin{itemize}
\item If $P$ is the template in Figure~\ref{fig: strict} then $P_0=P$, and $P_1$ is the empty process (because there are no broadcast edges). Thus $\puwd$ is a copy of $P$. 
\item If $P$ is the template in Figure~\ref{fig: process not regular 2} then $P_0$ is equal to $P$ without the broadcast edges, and $P_1 = P_0$. Thus $\puwd$ is a copy of $P$.
\item If $P$ is the template on the right hand side of Figure~\ref{fig:ex_dtn_rbs}, then $\puwd$ contains four components (prefix length $3$ and period $1$). The four components are drawn in Figure~\ref{fig:puwd-timed-ex}.
 \end{itemize}

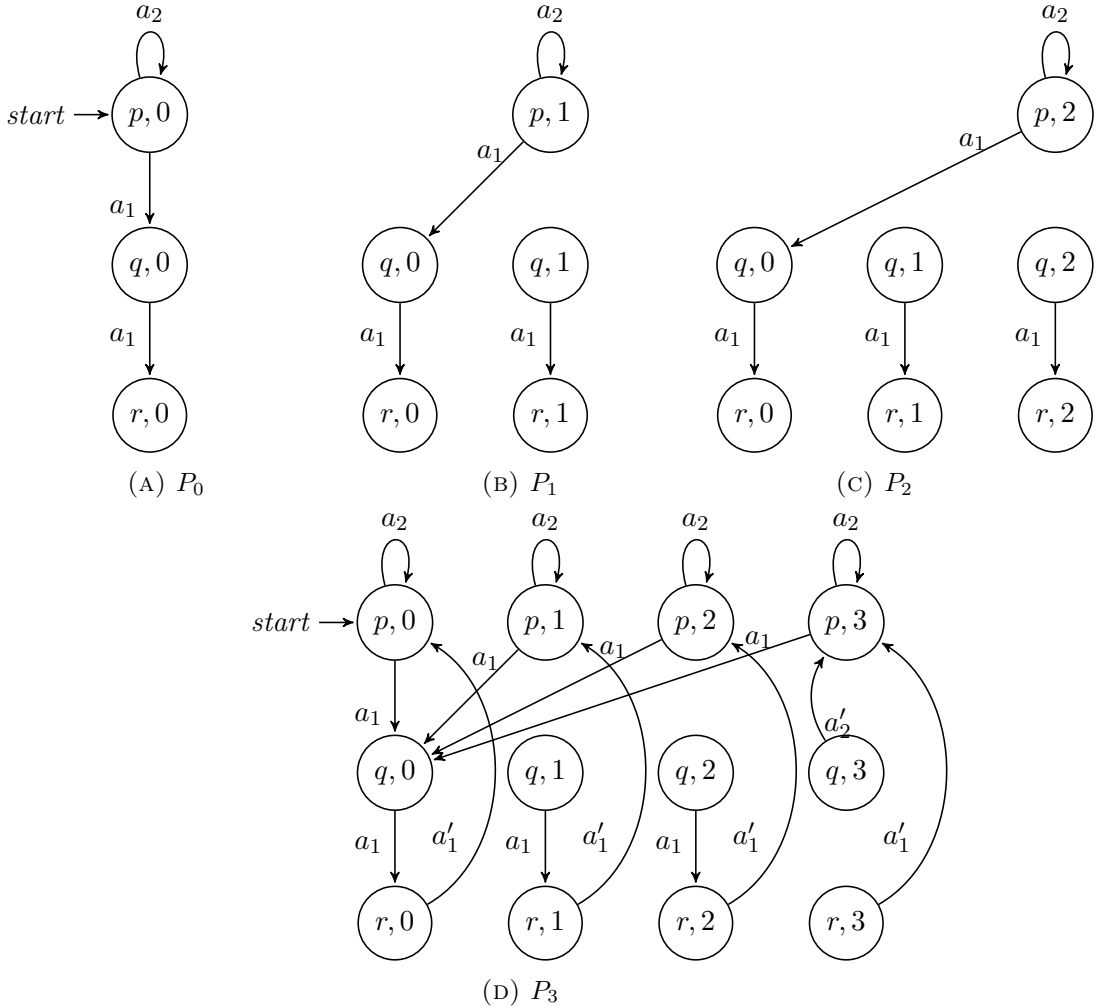
\begin{figure}
\centering
\hspace*{-20pt}
\begin{subfigure}{0.3\textwidth}

\begin{tikzpicture}[->,>=stealth',shorten >=1pt,auto,node distance=2cm,
                    semithick]

  \node[initial,state] (P0)                  {$p,0$};
  \node[state]         (Q0) [below of=P0]     {$q,0$};
  \node[state]         (W0) [below of=Q0]     {$r,0$};

  \path (P0) edge node[below,xshift=-10] {$\msg{a_1}$} (Q0)
        (P0) edge [loop above] node {$\msg{a_2}$} (P0)
        (Q0) edge node[below,pos=0.2,xshift=-10] {$\msg{a_1}$} (W0);

\end{tikzpicture}
\caption{$P_0$}
\end{subfigure}
\begin{subfigure}{0.3\textwidth}
\begin{tikzpicture}[->,>=stealth',shorten >=1pt,auto,node distance=2cm,
                    semithick]

  \node[state]         (Q0) [below of=P0]     {$q,0$};
  \node[state]         (W0) [below of=Q0]     {$r,0$};

  \node[state]         (P1) [right of=P0]                 {$p,1$};
  \node[state]         (Q1) [below of=P1]     {$q,1$};
  \node[state]         (W1) [below of=Q1]     {$r,1$};

\path
        (Q0) edge node[below,pos=0.2,xshift=-10] {$\msg{a_1}$} (W0);

\path (P1) edge node[pos=0.2,below,xshift=-5,yshift=10] {$\msg{a_1}$} (Q0)
        (P1) edge [loop above] node {$\msg{a_2}$} (P1)
        (Q1) edge node[below,pos=0.2,xshift=-10] {$\msg{a_1}$} (W1);

\end{tikzpicture}
\caption{$P_1$}
\end{subfigure}
\begin{subfigure}{0.28\textwidth}
\begin{tikzpicture}[->,>=stealth',shorten >=1pt,auto,node distance=2cm,
                    semithick]

  \node[state]         (Q0) [below of=P0]     {$q,0$};
  \node[state]         (W0) [below of=Q0]     {$r,0$};

  \node[state]         (Q1) [below of=P1]     {$q,1$};
  \node[state]         (W1) [below of=Q1]     {$r,1$};

  \node[state]         (P2) [right of=P1]                 {$p,2$};
  \node[state]         (Q2) [below of=P2]     {$q,2$};
  \node[state]         (W2) [below of=Q2]     {$r,2$};

\path
        (Q0) edge node[below,pos=0.2,xshift=-10] {$\msg{a_1}$} (W0);

\path
	(Q1) edge node[below,pos=0.2,xshift=-10] {$\msg{a_1}$} (W1);

  \path (P2) edge node[pos=0.15,below,xshift=-5,yshift=10] {$\msg{a_1}$} (Q0)
        (P2) edge [loop above] node {$\msg{a_2}$} (P2)
        (Q2) edge node[below,pos=0.2,xshift=-10] {$\msg{a_1}$} (W2);

\end{tikzpicture}
\caption{$P_2$}
\end{subfigure}
\begin{subfigure}{0.49\textwidth}
\begin{tikzpicture}[->,>=stealth',shorten >=1pt,auto,node distance=2cm,
                    semithick]

  \node[initial,state] (P0)                  {$p,0$};
  \node[state]         (Q0) [below of=P0]     {$q,0$};
  \node[state]         (W0) [below of=Q0]     {$r,0$};

  \node[state]         (P1) [right of=P0]                 {$p,1$};
  \node[state]         (Q1) [below of=P1]     {$q,1$};
  \node[state]         (W1) [below of=Q1]     {$r,1$};

  \node[state]         (P2) [right of=P1]                 {$p,2$};
  \node[state]         (Q2) [below of=P2]     {$q,2$};
  \node[state]         (W2) [below of=Q2]     {$r,2$};

  \node[state]         (P3) [right of=P2]                 {$p,3$};
  \node[state]         (Q3) [below of=P3]     {$q,3$};
  \node[state]         (W3) [below of=Q3]     {$r,3$};

  \path (P0) edge node[below,xshift=-10] {$\msg{a_1}$} (Q0)
        (P0) edge [loop above] node {$\msg{a_2}$} (P0)
        (Q0) edge node[below,pos=0.2,xshift=-10] {$\msg{a_1}$} (W0)
        (W0) edge[bend right=60] node[pos=0.2] {$\msg{a'_1}$} (P0);

  \path (P1) edge node[pos=0.2,below,xshift=-5,yshift=10] {$\msg{a_1}$} (Q0)
	(Q1) edge node[below,pos=0.2,xshift=-10] {$\msg{a_1}$} (W1)
	        (P1) edge [loop above] node {$\msg{a_2}$} (P1)
	        (W1) edge[bend right=60] node[pos=0.2] {$\msg{a'_1}$} (P1);

  \path (P2) edge node[pos=0.15,below,xshift=-5,yshift=10] {$\msg{a_1}$} (Q0)
        (Q2) edge node[below,pos=0.2,xshift=-10] {$\msg{a_1}$} (W2)
                (P2) edge [loop above] node {$\msg{a_2}$} (P2)
                (W2) edge[bend right=60] node[pos=0.2] {$\msg{a'_1}$}   (P2);

  \path (P3) edge node[pos=0.1,below,xshift=-5,yshift=10] {$\msg{a_1}$} (Q0)
        (P3) edge [loop above] node {$\msg{a_2}$} (P3)
        (Q3) edge[bend left=33] node[below, xshift=10,yshift=5] {$\msg{a'_2}$} (P3)
        (W3) edge[bend right=60] node[pos=0.2] {$\msg{a'_1}$}  (P3);

\end{tikzpicture}
\caption{$P_3$}

\end{subfigure}
\caption{Components $P_0,P_1,P_2,P_3$ of $\puwd$ for $P$ from Figure~\ref{fig:ex_dtn_rbs}.}\label{fig:puwd-timed-ex}
\end{figure}

\end{example}

\begin{definition}[Legal configuration/path] \label{dfn: legal}
 A configuration $f$ of $\puwdsys$ is \emph{legal} iff it is in $\puwd_i$ for some $i$; a path in $\puwdsys$ is \emph{legal} iff its source configuration is.
\end{definition}

\begin{remark} \label{rem: broadcast component} If $\pi$ is a finite path of $\puwdsys$, with $b$ broadcast transitions, with source $f$ and destination $f'$, then if $f$ is in $\puwd_i$ then $f'$ is in $\puwd_{i+b}$, and if $\pi$ is a run then $f'$ is in $\puwd_b$. In particular, any configuration of $\puwdsys$ that is reachable from an initial configuration is legal.
\end{remark}

Recall that we introduced $\puwd$ in order to define an automaton recognizing $\execfin$. Before doing so, we have to understand the relationship between $\puwd$ and $\proctemp$.

\paragraph{On the relation between $\puwd$ and $\proctemp$.}
Observe that by projecting out the component numbers from states in $\puwd$ (i.e., by replacing $(s,i) \in S^\uwd$ with $s \in S$), states and transitions in $\puwd$ are transformed into states and transitions of $\proctemp$. Similarly, paths and runs in $\puwdsys$ can be transformed into paths and runs in $\psys$. Note, however, that this operation does not induce a bisimulation between $\puwd$ and $\proctemp$, nor does it induce a bisimulation relation between $\psys$ and $\puwdsys$, since not all states and transitions of $\proctemp$ appear in every component of $\puwd$. These missing states and edges are also the reason that a path in $\psys$ that is not a run (i.e., that does not start at an initial configuration) may not always be lifted to a path in $\puwdsys$.
Nonetheless, our construction of $\puwd$ is such that runs of $\proctemp$ (resp. $\psys$) can be lifted to runs of $\puwd$ (resp. $\puwdsys$) by simply adding the correct component numbers (based on the number of preceding broadcasts) to the states of the transitions of the run.

\paragraph{Winding and Unwinding Notation}
More formally, projecting out of component numbers (which we call ``winding'' and denote by $\rwd$) is defined as follows: if $(s,j)$ is a state of $P^\uwd$ define $(s,j)^\rwd := s$, which is a state of $P$; if $t$ is a transition $\trans{s}{s'}{\varsigma}$ of $P^\uwd$ define $t^\rwd := (s^\rwd,\varsigma,s'^\rwd)$, which is a transition of $P$; if $\pi = t_1 t_2 \dots \in \runs(P^\uwd)$ define $\pi^\rwd := t_1^\rwd t_2^\rwd \dots \in \runs(P)$. Similarly, if $\mathfrak{f}$ is a configuration in $\puwdsys$ define $\mathfrak{f}^\rwd$, a configuration of $\psys$, by $\mathfrak{f}^\rwd(i) := \mathfrak{f}(i)^\rwd$ where $i$ is in the domain of $\mathfrak{f}$; if $e$ is a global transition $\trans{\mathfrak{f}}{\mathfrak{g}}{\sigma}$ of $\puwdsys$ then define $e^\rwd := (\mathfrak{f}^\rwd,\sigma,\mathfrak{g}^\rwd)$; and if $\rho \in \runs(\puwdsys)$ define $\rho^\rwd = \rho_1^\rwd \rho_2^\rwd \dots \in \runs(\psys)$. Finally, we apply this to sets: if $X \subseteq \runs(\puwdsys)$ then $X^\rwd = \{\rho^\rwd : \rho \in X\}$.

We define the reverse transformation of ``unwinding'' only with respect to runs of $\psys$ (a similar definition can be given for the unwinding of runs of $\proctemp$) as follows: given a configuration $f$ of $\psys$, and a component number $j$, denote by $f^j$ the function defined by $f^j(i) := (f(i), j)$ for every $i$ in the domain of $f$;
given $\pi \in \runs(\psys)$, for $i \in \Nat$ let $\brd^{<i}$ be the number of broadcast transitions on $\pi$ preceding $\pi_i$.
The \emph{unwinding} $\pi^\uwd$ of $\pi$ is defined to be the sequence $\pi_1^\uwd \pi_2^\uwd \dots$ obtained by taking for every $1 \leq i \leq |\pi|$ the transition  $\pi_i^\uwd := (f^{\comp(\brd^{<i})}, \sigma, g^{\comp(\brd^{<i})})$ if $\pi_i = (f,\sigma,g)$ is a rendezvous transition, and otherwise taking $\pi_i^\uwd := (f^{\comp(\brd^{<i})}, \brd, g^{\comp(\brd^{<i+1})})$ if $\pi_i = (f, \brd, g)$ is a broadcast transition.


The next lemma says that we may work with template $\puwd$ instead of $P$.
\begin{lemma}\label{lem: runs can be lifted to unwinding}
  For every $n \in \Nat$, we have that $\runs(\sysinst{n}) = \runs((\puwd)^n)^\rwd$.
\end{lemma}
\begin{proof}
 Let us fix any $n \in \Nat$. The direction $\{ \rho^\rwd \mid \rho \in \runs((\puwd)^n) \} \subseteq \runs(\sysinst{n})$ follows from the fact that $\puwd$ is obtained from $\proctemp$ by an unwinding process. The reverse inclusion requires more care as $\puwd$ misses edges and states of $\proctemp$.

  Let $\pi \in \runs(\sysinst{n})$. We prove that $\pi^\uwd \in \runs((\puwd)^n)$ by induction on the length $i$ of each prefix of $\pi$.
   Let $b$ be the number of broadcast edges on $\xi :=\pi_1 \dots \pi_{i-1}$.
   For the base case $|\pi| = 0$, there is nothing to prove. For the induction step, observe that by the inductive hypothesis
   the unwinding of $\xi$ is a run of $(\puwd)^n$. It remains to show that $\pi_i^\uwd$  is a transition of $(\puwd)^n$.
   Observe that by Remark~\ref{rem: broadcast component} $f = \dst(\pi_{i-1}) = \src(\pi_i)$ is in $\puwd_{\comp(b)}$. Consider first the case that $\pi_i$ is the broadcast transition $\trans{f}{g}{\brd}$. By the definition of the broadcast edges in $\puwd$ we have that $\pi_i^\uwd$ is a transition $\trans{f^{\comp(b)}}{g^{\comp(b+1)}}{\brd}$ of $(\puwd)^n$. Consider now the case that $\pi_i$ is the rendezvous edge $\trans{f}{g}{\sigma}$, and let $\sigma = ((j_1, \msg{a}_1), \dots, (j_k, \msg{a}_k))$. Since $f$ is in $\puwd_{\comp(b)}$, for every $h \in [k]$, the algorithm used to construct the states and transitions of the component $P_{\comp(b)}$ must have added the edge $\trans{f(j_h)}{g(j_h)}{\msg{a}_h}$ to $R_{\comp(b)}$. It follows that $\pi_i^\uwd = (f^{\comp(b)}, \sigma, g^{\comp(b)})$  is a transition of $(\puwd)^n$.
 \end{proof}

The state labeling of a run $\rho \in \runs(\puwdsys)$ and its winding $\rho^\rwd$ are equal. Thus we have the following:
\begin{corollary}\label{cor:exec psys and puwdsys}
For every template $P$, we have that
 $\exec[\psys] = \exec[\puwdsys]$.
\end{corollary}

The following lemma says that for every component $\puwd_b$, there is a run of $\puwdsys$ that ``loads'' arbitrarily many processes into every state of it.

\begin{lemma}[Loading]
\label{lem: loading}
For all $b,n \in \Nat$ there is a finite run $\pi$ of $\puwdsys$ with $b$ broadcasts, s.t., $|\dst(\pi)^{-1}(s)| \geq n$ for every state $s$ of $\puwd_b$.
\end{lemma}

\begin{proof}
By Lemma~\ref{lem: rb-system composition} (Composition) applied to $\puwd$ it is sufficient to prove the following: for every $b \in \Nat$, and every state $q$ in $\puwd_b$, there exists a finite run $\pi$ of $\puwdsys$, with $b$ broadcast transitions, such that $|\dst(\pi)^{-1}(q)| \geq 1$. Recall that, by definition, $\puwd_b = \puwd_{comp(b)}$.

The proof is by induction on $b$.
For the base case $b = 0$, proceed by induction on the round $j \geq 1$ of the saturation algorithm at which $q$ is added to $S_{\comp(b)}$ (i.e., $S_0$). The state $q$ is added at round $j$, due to some edge $(s_h,\msg{a}_h,q)$ of $R$, only if for every $l \in [k]\setminus\{h\}$ there are edges $(s_l, \msg{a}_l,q_l)$ of $R$ and, either (i) $j = 1$ and $s_l \in I_{\comp(b)}$ or, (ii) $j > 1$ and $s_l$ is already in $S_{\comp(b)}$ (i.e., it was added to $S_{\comp(b)}$ at a round before $j$).
By the inductive hypothesis on round $j$, for every $l \in [k] \setminus \{h\}$ there exists $\rho_l \in \runs(\puwdsys)$, with $b$ broadcasts, which ends with at least one process in the state $s_l$. By Lemma~\ref{lem: rb-system composition} (Composition) there exists $\rho \in \runs(\puwdsys)$ in which there are $k$ different processes $i_1, \ldots, i_k$ such that, for every $l \in [k]$, the process $i_l$ ends in the state $s_l$. Extend $\rho$ by a global rendezvous transition in which, for $l \neq h$, process $i_l$ takes the edge $(s_l,\msg{a}_l,q_l)$, and process $i_h$ takes the edge $(s_h,\msg{a}_h,q)$. This extended run has $b$ broadcast transitions, and at least one process in state $q$, as required.

For the inductive step ($b > 0$), suppose it holds for all values $\leq b$, and let us prove it for $b + 1$ (i.e. take $q \in S_{\comp(b+1)}$). First consider the case $q \in I_{\comp(b+1)}$: there is an edge $(s,\brd,q)$ in $\puwd$ and by the inductive hypothesis (on $b$) there is a run of $\puwdsys$ with $b$ broadcasts in which some process $i$ ends in state $s$. Extend this run by a global broadcast transition in which process $i$ takes the edge $(s,\brd,q)$. This extended run has $b+1$ broadcast transitions, and at least one process in state $q$. Second, suppose $q \in S_{\comp(b+1)} \setminus I_{\comp(b+1)}$. Then proceed as in the base case.
\end{proof}

The first part of the following proposition states that the set of finite executions of the RB-system $\psys$ is equal to the set of state labels of the finite runs
of $\puwd$. This is very convenient since $\puwd$ is finite, whereas $\psys$ is infinite. The second part of the proposition gives a weaker result for the infinite case.

\begin{proposition}\label{prop: unwinding captures executions}
For every template $P$, the following holds:
\begin{enumerate}
\item  $\execfin[\psys] = \{\lambda^{\uwd}(\pi) \mid \pi \in \runs(\puwd), |\pi| \in \Nat \}$.
\item $\execinf[\psys] \subseteq \{\lambda^{\uwd}(\pi) \mid \pi \in \runs(\puwd), |\pi| = \infty \}$.
\end{enumerate}

\end{proposition}
\begin{proof}

    We first prove the inclusion $\exec \subseteq \{\lambda^{\uwd}(\pi) \mid \pi \in \runs(\puwd)\}$.
  Every execution of $\psys$ is, by definition, of the form $\lambda(\proj{\xi}(1))$ for some $\xi \in \runs(\sysinst{n})$ and some $n$. By Lemma~\ref{lem: runs can be lifted to unwinding}, $\xi = \rho^\rwd$ for some $\rho \in  \runs((\puwd)^n)$. Observe that $\xi$ and $\rho$ are equi-labeled. Thus, by Lemma~\ref{lem:bisim} part~\ref{lem:bisim:equi} we have that $\lambda(\proj{\xi}(1)) = \lambda^\uwd(\proj{\rho}(1))$.

We now prove the inclusion
\[
\{\lambda^\uwd(\pi) \mid \pi \in \runs(\puwd), |\pi| \in \Nat\} \subseteq \execfin[\psys].
\]
 Observe that since $\lambda^\uwd(\pi) = \lambda(\pi^\rwd)$ it is enough to prove the following by induction on the length $i$ of $\pi$: there is a run $\rho \in \runs(\psys)$ such that $\proj{\rho}(1) = \pi^\rwd$.

For the base case $i = 0$ there is nothing to prove. For the inductive step $i > 0$: first apply the inductive hypothesis to get $\rho \in \runs(\psys)$ such that $\proj{\rho}(1) = (\pi_1 \pi_2 \cdots \pi_{i-1})^\rwd$. There are two cases depending on $\pi_i$.

If $\pi_i$ is a broadcast edge then extend $\rho$ by a global broadcast transition $t$ in which process $1$ takes $\pi_i^\rwd$, i.e., $edge_1(t) = \pi_i^\rwd$, to obtain the run $\rho \cdot t \in \runs(\psys)$ whose projection on process $1$ equals $\pi^\rwd$.

If $\pi_i = (s,\msg{a}_h,t)$ is a rendezvous edge then proceed as follows. Let $b$ be the number of broadcast transitions in $\pi_1 \cdots \pi_{i-1}$. So $\pi_i$, being a rendezvous edge, is in $\puwd_b$. Thus, by Remark \ref{rem:saturation}, 
after the saturation algorithm, for all $l \in [k]$ there exist an edge $(s_l,\msg{a}_l,t_l)$ in $\puwd_b$. By Lemma~\ref{lem: loading} (Loading) there exists $\rho' \in \runs(\puwdsys)$ with $b$ broadcast transitions that loads at least one process into every state $s$ of $\puwd_b$. By Lemma~\ref{lem: runs can be lifted to unwinding}, $\rho'^\rwd \in \runs(\psys)$. By Lemma~\ref{lem: rb-system composition} (Composition) compose $\rho$ and $\rho'^\rwd$ to get $\rho'' \in \runs(\psys)$ such that $\restr{\rho''}{\{1\}} = \proj{\rho}(1) =  (\pi_1 \pi_2 \cdots \pi_{i-1})^\rwd$ and at the end of $\rho''$ there is at least one process (different from process $1$) in every state of $\puwd_b$.  Now extend $\rho''$ by the rendezvous transition $t$ for which $edge_1(t) = \pi_i^\rwd$ and for each $l \in [k] \setminus \{h\}$ some process takes the transition $(s_l,\msg{a}_l,t_l)$ to obtain the run $\rho \cdot t \in \runs(\psys)$ whose projection on process $1$ equals $\pi^\rwd$.
\end{proof}

\begin{remark}
Unfortunately, the containment in Proposition~\ref{prop: unwinding captures executions} Part 2 is sometimes strict. For example, consider the $R$-template $\proctemp$ in Figure~\ref{fig: strict}. Observe that $\proctemp$ equals $\puwd$, and that $p^\omega$ is the state label of the run of $\puwd$ that self-loops in the initial state forever, but $p^\omega$ is not an execution of $\psys$. In the next section we will use B-automata to capture $\execinf$.
\end{remark}

We introduced $\puwd$ in order to define an automaton recognizing $\execfin$. This automaton is formed from the LTS $\puwd$ by adding the input alphabet $2^\AP$ and having each transition read as input the label of its source state.

\begin{definition}[NFW $\execNFW$] \label{dfn: automaton for finite executions}
Given an RB-template $\proctemp = \tup{\AP,\Actionprts \cup \{\brd\}, S,I,R,\lambda}$, consider the reachability-unwinding $\puwd = \tup{AP,\Actionprts \cup \{\brd\},S^\uwd,I^\uwd,R^\uwd,\lambda^\uwd}$.
Define $\execNFW$ to be the NFW $\tup{\Sigma,S',I',R',F}$ with
\begin{itemize}
 \item input alphabet $\Sigma = 2^{AP}$,
 \item state set $S' = S^\uwd$,
 \item initial-states set $I' = I^\uwd$,
 \item transition relation $R'$ consisting of transitions $(s,\lambda^\uwd(s),t)$ for which there is a $\sigma$ such that $(s,\sigma,t) \in R^\uwd$,
 \item final-states set $F = S^\uwd$.
\end{itemize}
\end{definition}

The following is immediate from Proposition~\ref{prop: unwinding captures executions} Part 1:

\begin{corollary} \label{cor: execNFW recognises execfin}
The automaton $\execNFW$ recognizes the language $\execfin$.
\end{corollary}

Applying a standard automata-theoretic technique, we get the following upper bound:

\begin{theorem}\label{thm: finite RB PMCP is in pspace}
Let $\Pspec$ be specifications of finite executions expressed as NFWs or \LTLf formulas.
Then $PMCP(\Pspec)$ for RB-systems is in \PSPACE.
\end{theorem}
\begin{proof}
Let $\proctemp = \tup{\AP,\Actionprts \cup \{\brd\}, S,I,R,\lambda}$ be a process template, and let $\execNFW$ be the NFW from Definition~\ref{dfn: automaton for finite executions}. The fact that words accepted by $\execNFW$ are exactly the executions in $\execfin$ is by Corollary~\ref{cor: execNFW recognises execfin}. Analyzing the construction of the unwinding template $\puwd$ (before Definition~\ref{dfn: unwinding}), we get that $\execNFW$ is of size at most exponential in the size of $\proctemp$.

We describe a \PSPACE algorithm for checking the containment of the language accepted by $\execNFW$ in the language of some specification NFW $\A'$. This is done by solving the non-containment problem in nondeterministic polynomial space, and using the fact that \NPSPACE = \PSPACE = co-\PSPACE.\footnote{Our nondeterministic algorithm will be allowed to diverge. Such an algorithm can be simulated by one that never diverges and still runs in polynomial space by incrementing a counter at every step, and rejecting the computation-branch if the counter ever exceeds the original number of configurations of the possibly-diverging algorithm.}
The algorithm constructs on the fly: (1) a finite word $\rho \in (2^\AP)^*$, and an accepting run of $\execNFW$ on $\rho$; and (2) checks that $\rho$ is not accepted by $\A'$.
Item (2) can be done, as usual, simply by storing the subset of states of $\A'$ that are reachable by reading the prefix of $\rho$ constructed thus far, and validating that, at the end, this set does not contain an accepting state.
For item (1), the algorithm does not store all of (the exponentially large $\execNFW$). Instead, at each point in time, it only stores a single component $\puwd_i$ of $\puwd$ (which can be calculated in \PTIME for every $i$), from which it can deduce the part of $\execNFW$ corresponding to it. The algorithm starts by constructing $\puwd_0$, and sets $\rho$ to be the empty word, and the run of $\execNFW$ on $\rho$ to be the initial state of $\execNFW$. At each step, it can either declare the guess as finished (if the run constructed thus far ends in an accepting state of $\execNFW$) or extend $\rho$ and the run. Extending $\rho$ is a trivial guess. Extending the guessed run is done by either guessing a transition of $\execNFW$ inside  the component induced by the currently stored $\puwd_i$; or by guessing a transition that moves to the next component in the lasso, at which point the algorithm also discards $\puwd_i$ and replaces it with $\puwd_{i+1}$.

In case the specification is given as an \LTLf formula $\varphi$, we let $\A'$ be the NFW from Theorem~\ref{thm:vardi-wolper} corresponding to $\lnot \varphi$ and replace (2) above by a check that $\rho$ is accepted by $\A'$, which can be done, as usual, simply by storing the subset of states of $\A'$ that are reachable by reading the prefix of $\rho$ constructed thus far, and validating that, at the end, this set does contain an accepting state.
\end{proof}

The next theorem gives a corresponding lower bound.
Interestingly, its proof shows that the problem is already \PSPACE hard for \emph{safety} specifications (i.e., that a bad state is never visited).

\begin{theorem}\label{thm: lower-bound-finite-PMCP}
Let $\Pspec$ be specifications of finite executions expressed as NFW or \LTLf formulas. Then $PMCP(\Pspec)$ for RB-systems is \PSPACE-hard. Moreover, this is true even for a fixed specification, and thus the program-complexity is \PSPACE-hard.
\end{theorem}
\begin{proof}
  The proof proceeds by a reduction from the reachability problem for Boolean programs, known to be \PSPACE-complete~\cite{books/daglib/0090683}.

  A Boolean program consists of $m$ Boolean variables $X_1, \ldots, X_m$ (for some $m$) and $n$ instructions (for some $n$), referred to by their \emph{program location} $l \in [n]$, of two types: \emph{(i)} \emph{conditionals} of the form $l: \texttt{if } X_i  \texttt{ then } l_{if}  \texttt{ else } l_{else}$; \emph{(ii)} \emph{toggles} of the form $l: X_i := \neg X_i$. The semantics of the first type of instruction is to move from location $l$ to location $l_{if}$ if $X_i$ is true and to location $l_{else}$ otherwise; instructions of this type are thus conditional jumps that do not change the values of any of the Boolean variables.
  The semantics of the second type of instruction is to negate the value of the Boolean variable $X_i$;
  the execution continues from location $l+1$ (unless that was the last instruction).
  All Boolean variables are initialized to $\texttt{false}$ and execution begins with instruction $1$.
  We remark that the Boolean programs considered here are deterministic.
  The reachability problem for Boolean programs is to decide whether the execution of a Boolean program ever reaches its last program location $n$. Note that we can assume, without loss of generality, that the last instruction of a Boolean program is a conditional instruction.

  Given a Boolean program $B$, we build a process template $P$, and a specification NFW, such that $\psys$ satisfies the specification iff the execution of $B$ does not reach its last instruction.

  Formally, $\proctemp = \tup{AP,\Actionprts \cup \{\brd\}, S,\{\iota\},R,\lambda}$, where:
  \begin{itemize}
    \item $AP = \{\textrm{done}\}$, i.e., there is a single atom;
    \item $\Actionprts = \cup_{\msg{a} \in \Actions} \{ \msg{a}_1, \msg{a}_2\}$  where
    \[
    \Actions = \cup_{i \in [m]} \{ \msg{(\text{protect},i)}, \msg{(\text{if},i)}, \msg{(\text{else},i)}, \msg{(\text{toggle},i)} \}
    \]

    \item The set of states $S := \{\iota, {\tt sink}\} \cup S_{instr} \cup S_{var}$, where $\iota$ is an initial state, $\tt{sink}$ is a sink state, and:
    \begin{enumerate}
        \item $S_{instr} := \cup_{l \in [n]} \{l,l'\}$,
        \item $S_{var} := \cup_{i \in [m]} \{X_i, \neg X_i, X'_i, \neg X'_i\}$,
    \end{enumerate}
    \item $R$ will be defined later;
    \item $\lambda = \{(n,\textrm{done})\}$, i.e., the atom is true in state $n$, and false in all other states.
  \end{itemize}

  The specification says that the last program location $n$ is never visited. This can be expressed, for instance, by the \LTLf formula $\always \lnot \textrm{done}$ (read "it is always the case that atom \textrm{done} does not occur"), and by an NFW consisting of a single state. In what follows we call the specification $\varphi$.

  Before describing the transition relation $R$, we briefly describe the way the states are used in the simulation of the Boolean program $B$  by runs of $\psys$.
  At the beginning of every run, every process nondeterministically moves (on a broadcast) from the initial state either to the state $1$, or to one of the states $\neg X_1, \cdots, \neg X_m$. A process that moves to the state $1$ will keep track of (i.e., encode) the program location of the Boolean program, and from this point on will only be in states from the set $S_{instr}$; whereas a process that moves to a state of the form $\neg X_i$ will keep track of (i.e., encode) the value of variable $X_i$, and from this point on will only be in states from the set $\{X_i, \neg X_i, X'_i, \neg X'_i\}$. Observe that multiple processes may decide to encode the program location or the same variable. However,
  the transition relation $R$ will be defined in such a way as to enforce the invariant that, right after every broadcast, the following holds:

\begin{description}
\item [($\dagger$)] all processes that encode the same object (i.e., the program location or the value of some variable) agree on its value.
\end{description}

Moreover, between every two broadcasts, at most one instruction of the Boolean program can be simulated, namely, the instruction referenced by the processes that track the program location. The primed versions of the states will be used in order to enforce this round structure, as well as the invariant $\dagger$, as follows: in each round, every rendezvous transition moves a process from an unprimed state to a primed state, from which it can only move on a broadcast; whereas a broadcast takes a process in a primed state back to an unprimed state, and processes in an unprimed state to $\tt{sink}$.

  Let $var(l)$ denote the index $i$ of the variable used (i.e., tested or toggled) in the instruction in program location $l$.
  We now define the transition function $R$ of the template $\proctemp$. It consists of the following transitions:
\begin{itemize}
\item $\trans{\iota}{1}{\brd}$; $\trans{\iota}{\neg X_i}{\brd}$ for $i \in [m]$,
\item $\trans{\tt{sink}}{\tt{sink}}{\brd}$; $\trans{l}{\tt{sink}}{\brd}$,  $\trans{X_i}{\tt{sink}}{\brd}$, $\trans{\neg X_i}{\tt{sink}}{\brd}$, for $l \in [n], i \in [m]$,
\item $\trans{l'}{l}{\brd}$, $\trans{X'_i}{X_i}{\brd}$, and $\trans{\neg X'_i}{\neg X_i}{\brd}$.

\item $\trans{l}{l}{\msg{(\text{protect},i)}_1}$ for $l \in [n]$ and $i \in [n] \setminus var(l)$.
\item $\trans{X_i}{X'_i}{\msg{(\text{protect},i)}_2}$ and
$\trans{\neg X_i}{\neg X'_i}{\msg{(\text{protect},i)}_2}$ for $i \in [m]$.
\item $\trans{l}{l'_{if}}{\msg{(\text{if},var(l))}_1}$ and $\trans{l}{l'_{else}}{\msg{(\text{else},var(l))}_1}$ for all conditional instructions $l$.
\item $\trans{X_i}{X'_i}{\msg{(\text{if},i)}_2}$ and $\trans{\neg X_i}{\neg X'_i}{\msg{(\text{else},i)}_2}$ for $i \in [m]$.
\item $\trans{l}{(l+1)'}{\msg{(\text{toggle},var(l))}_1}$ for all toggle instructions $l$.
\item $\trans{X_i}{\neg X'_i}{\msg{(\text{toggle},i)}_2}$ and $\trans{\neg X_i}{X'_i}{\msg{(\text{toggle},i)}_2}$ for $i \in [m]$.
\end{itemize}

We now prove that the reduction is correct.

Suppose that the infinite run $\rho$ of the Boolean program visits its last instruction. We build a run $\pi$ of $\sysinst{m+1}$ witnessing the fact that
$\psys$ does not satisfy $\varphi$. The run $\pi$ simulates $\rho$ as follows. Start with a broadcast, which takes process $m+1$ (called the \emph{controller process}) to state $1$, and for every $i \in [m]$ takes process $i$ (called the \emph{$i$'th variable process}) to $\neg X_i$.
Repeatedly extend the run $\pi$ by the following sequence of global transitions (below, $l$ denotes the current state of the controller process):
\begin{enumerate}
 \item For every $i \neq var(l)$, the controller rendezvous with the $i$'th memory process on the action $\msg{(\text{protect},i)}$.
 \item The controller rendezvous with the $var(l)$'th memory process as follows: if $l$ is a toggle instruction then the rendezvous action is $\msg{(\text{toggle},i)}$; otherwise, it is $\msg{(\text{if},i)}$ if the $i$'th memory process is in state $X_i$, and it is $\msg{(\text{else},i)}$ if it is in state $\neg X_i$.
 \item A broadcast.
\end{enumerate}

It is easy to see that $\pi$ simulates $\rho$. In particular, the state of the controller after $z \geq 1$ broadcasts is equal to the program location of the Boolean program after $z$ steps.

For the other direction, we argue as follows. We say that a configuration $f$ of $\psys$ is \emph{consistent} if it satisfies the invariant $\dagger$ stated earlier. For such an $f$, let $pl(f) \in [n]$ be the program location encoded by $f$, or $\bot$ if there are no processes in $f$ tracking the program location; and for every $i \in [m]$, let $val_i(f) \in \{\true, \false\}$ be the value of $X_i$ encoded by $f$, or $\bot$ if there are no processes in $f$ tracking the value of $X_i$. Given $z \in \mathbb{N}$, and a run $\pi$ of $\psys$ with at least {$z$} broadcasts, write $\pi(z)$ for the configuration in $\pi$ immediately following the $z$'th broadcast. 
Observe that it is enough to show that $\pi$ simulates the run $\rho$ of the Boolean program in the following sense:
\begin{enumerate}
 \item $\pi(z)$ is consistent, 
 \item if $pl(f_z) \neq \bot$ then the program location in $\rho_z$ is equal to $pl(f_z)$, 
 \item for every $i \in [m]$, if $val_i(f_z) \neq \bot$ then the value of variable $X_i$ in $\rho_z$ is equal to $val_i(f_z)$. 
\end{enumerate}

We prove the items above by induction on $z$.
For $z = 1$, i.e., after the first broadcast (which must be the first transition on any run), processes assume different roles. Any process that moves to state $1$ is called a \emph{controller}, and
any process that moves to state $\neg X_i$ (for some $i \in [m]$) is called an \emph{$i$'th variable processes}. Clearly the induction hypothesis holds. For the inductive step, note that by the inductive hypothesis $\pi(z-1)$ is consistent, let $l := pl(\pi(z-1))$, and observe that the only rendezvous transitions on $\pi$ between the $z-1$ and $z$ broadcasts are of a controller process that rendezvous with a variable process on an action of the form described in items $1$ and $2$ in the proof of the first direction (in particular, if $l = \bot$ then there are no rendezvous between the $z-1$ and $z$ broadcasts). Thus, it must be that, just before the $z$'th broadcast, processes in a primed state that are encoding the same object are in the same state. Combining this with the fact that any process in an unprimed state will move to $\tt{sink}$ on the $z$'th broadcast one can see that the inductive hypothesis holds also after the $z$'th broadcast.   
\end{proof}

\begin{remark} \label{rem: spec complexity PSPACE-hard}
We now show that specification complexity of the PMCP for NFW and \LTLf specifications is also \PSPACE-hard. We do this by reducing from the standard model-checking problem.

Recall that the standard model-checking problem is, given an LTS $L$ without edge labels (aka `finite state program' or `Kripke structure') and a specification $\varphi$, to decide if all finite executions of $L$ satisfy $\varphi$. The specification complexity of the model-checking problem for \LTLf formulas or for NFW specifications is \PSPACE-hard. To see that, given an alphabet $\Sigma$, take a single state Kripke structure $K$ that generates all words in $\Sigma^*$, and note that model-checking $K$ and a given NFW specification is equivalent to deciding the universality problem for NFWs which is (even over a fixed alphabet) \PSPACE-hard~\cite{DBLP:books/fm/GareyJ79}. Similarly, model-checking $K$ and a given \LTLf formula is equivalent to deciding \LTLf satisfiability (using the negation of the original formula) which is again PSPACE-hard even for a fixed alphabet~P~\cite{DegVa13}.

To reduce the model-checking problem (for a fixed LTS $L$ without edge labels) to the PMCP problem with a fixed RB-template, simply build an RB-template $\proctemp$ from $L$ by adding the edge label $\brd$ to every transition, i.e., every transition in $L$ becomes a broadcast transition. Clearly, then, $\execfin$ is exactly the sequences of the form $\lambda(\pi)$ where $\pi$ is a finite run of $L$. Thus, $L \models \varphi$ iff all executions in $\execfin$ satisfy $\varphi$.
\end{remark}

Theorem~\ref{thm:PSPACE-complete} from Section~\ref{subsec: PMCP} now follows: the upper bound is in Theorem~\ref{thm: finite RB PMCP is in pspace}, the lower-bound on the program complexity (and thus the combined complexity) is in Theorem~\ref{thm: lower-bound-finite-PMCP}, and the lower-bound on the specification complexity is in Remark~\ref{rem: spec complexity PSPACE-hard}.

\section{Solving PMCP for Specifications over Infinite Executions}\label{sec:solving infinite}

The main step in our automata-theoretic approach to solve the PMCP for infinite executions is the construction, given an RB-template $\proctemp$, of a B-automaton $\execBSW$ (with a trivial B\"uchi set) that accepts the language $\execinf$. In this section we describe the construction of this automaton.

In order to construct the B-automaton $\execBSW$, it is helpful to recall the source of difficulty in dealing with infinite executions as opposed to finite executions. Recall from Section~\ref{sec:solving finite} that the finite executions were dealt with by simply turning the reachability-unwinding $\puwd$ into a nondeterministic automaton $\execNFW$ (by having each transition read as input the label of its source state), and that this worked because of the equality between the state-labels of finite runs of $\puwd$ and the finite executions of $\psys$, as stated in the first part of Proposition~\ref{prop: unwinding captures executions}. Also, recall that the second part of the same proposition, which deals with the infinite case, only states a containment (instead of equality), which may be strict --- as illustrated by the template $\proctemp$ in Example~\ref{ex:R-template}. Indeed, looking at this template again, one can see that in order to allow process $1$ to trace $p^z$ for $z \in \Nat$, we can use a system with $z+1$ processes that rendezvous with process $1$ one after the other. However, no finite amount of processes can allow process $1$ to trace $p^\omega$, since once a process rendezvous with process $1$ it cannot do so ever again. Thus, while the self loop on the initial state can be taken infinitely often in a path in the template $\proctemp$ (and hence also in a run of the automaton $\execNFW$), it can not be taken infinitely often in a run of $\psys$.

The key to modifying $\execNFW$ to obtain $\execBSW$ is to treat edges of $\puwd$ differently based on the conditions under which they can (or cannot) be taken infinitely often in runs of $\puwdsys$. In particular, one has to distinguish between edges that can appear infinitely often on a run with finitely or infinitely many broadcasts, and among the latter between ones that can or cannot appear unboundedly many times between two consecutive broadcasts. Note that the fact that an edge is only taken a bounded number of times between consecutive broadcasts can be naturally tracked by the acceptance condition of a single counter.

The rest of this section is organized as follows. We formally present the classification of edges along the lines outlined above, and prove a couple of easy lemmas about this classification. We then give the definition of the automaton $\execBSW$ and prove the correctness of the construction.

\begin{definition}[Edge Types]
\label{def:edge-types}
An edge $e$ of $\puwd$ is
 \begin{itemize}
 \item \emph{\locr} iff it appears infinitely many times on some run of $\puwdsys$ with finitely many broadcasts.
 \item \green iff it appears infinitely many times on some run of $\puwdsys$ with infinitely many broadcasts.
 \item \good iff it appears unboundedly many times between broadcasts on some run $\pi = \pi_0 \pi_1 \ldots$ of $\puwdsys$ with infinitely many broadcasts, i.e., if for every $n \in \Nat$ there are $i<j \in\Nat$ such that $\pi_i \dots \pi_j$ contains $n$ transitions using this edge and no broadcast edges.
 \item \sgood iff it is \green but not \good.

\end{itemize}
\end{definition}

Note that:
\begin{itemize}
\item\good edges are also \green,
\item \sgood edges are exactly those \green edges which satisfy that for every run of $\puwdsys$ there is a bound on the number of times they appear between any two consecutive broadcasts,
\item \green edges only belong to components of $\puwd$ that are on the loop of the lasso,
\item {broadcast edges can only be \sgood.}
\end{itemize}

\begin{example}
Neither edge in $\puwd$ for the template $\proctemp$ in Figure~\ref{fig: strict} is \locr or \green.

On the other hand, every edge in $\puwd$ for the template $\proctemp$ in Figure~\ref{fig: process not regular 2} is \sgood (and none are \locr).
\end{example}

It turns out that determining the type of an edge is decidable;  this is a non-trivial problem, and we dedicate Section~\ref{sec: deciding types of edges} to solving it.

We now characterize the edge types in terms of witnessing cycles in $\puwdsys$. Recall the definition of legal configuration and path (Definition~\ref{dfn: legal}).

\begin{lemma}\label{lem: edge types and cycles}
  An edge $e$ of $\puwd$ is:
\begin{enumerate}[i.]
  \item \locr iff it appears on a legal cycle $C_e$ of $\puwdsys$ that has no broadcast.
  \item \green iff it appears on a legal cycle $C_e$ of $\puwdsys$  that has broadcasts.
  \item \good\ iff it appears on a legal cycle $D_e$ of $\puwdsys$ that has no broadcast, that is contained in a legal cycle $C_e$ with broadcasts;
  \item \sgood\ iff it appears on a legal cycle $C_e$ of $\puwdsys$ that has broadcasts, but not on any cycle without broadcasts that is contained in a cycle with broadcasts.
\end{enumerate}
\end{lemma}

\begin{proof}
Observe that it is enough to prove the first three items.

For the `if' directions, let $n \in \Nat$ be the number of processes in $C_e$ (i.e., $C_e$ is a cycle in  $(\puwd)^n$),
and recall that since $C_e$ is legal, its source configuration $f$ is in $(\puwd_i)^n$ for some $i$. Hence, by Lemma~\ref{lem: loading} (Loading), a configuration $g$ such that $\restr{g}{[n]} = f$
(i.e., the first $n$ processes of $g$ form the configuration $f$) can be reached from an initial configuration of $\puwdsys$; then, for items (i) and (ii), we can simply pump $C_e$ forever (with the extra processes in $g$ moving only on broadcasts). For item (iii), $C_e$ is pumped in the following way: for every $i \in \Nat$, at the $i$'th repetition of the outer cycle $C_e$ we pump the inner cycle $D_e$ for $i$ times.

The `only if' directions follow from the observation that every run in $\puwdsys$ involves only finitely many processes, and thus only finitely many distinct configurations, all of which are legal (by Remark~\ref{rem: broadcast component}).
\end{proof}

Lemma~\ref{lem: edge types and cycles} implies that every \good\ edge is \locr, whereas a \sgood\ edge may or may not be \locr.\footnote{This overlap is the reason that we decided to use the term ``\locr'' instead of naming these edges by another color.} The following lemma states that we can assume that the cycles in Lemma~\ref{lem: edge types and cycles} that witness broadcasts all have the same number of broadcasts.

\begin{lemma}\label{lem: cycle implies one with K broadcasts}
There is a number $K$ such that for every 
{\green}, \good, or \sgood edge $e$, the cycle $C_e$ mentioned in items ii), iii) and iv) in Lemma~\ref{lem: edge types and cycles} can be taken to contain exactly $K$ broadcasts.
\end{lemma}

\begin{proof}
Apply Lemma~\ref{lem: edge types and cycles} to all the relevant edges in $\puwd$ and obtain cycles, say $C_{e_1}, C_{e_2}, \cdots, \allowbreak C_{e_l}$. Suppose $C_{e_i}$ has $k_i$ broadcasts. Let $K$ be the least-common-multiple of the $k_i$s. By repeating cycle $C_{e_i}$ for $K/k_i$ times, we obtain a witnessing cycle with exactly $K$ broadcasts.
\end{proof}

We now informally describe the structure of the automaton $\execBSW$. It is made of three copies of $\execNFW$ (called ${\execBSWinit}$, ${\execBSWbrd}$, ${\execBSWnobrd}$)
as follows: $\execBSWinit$ is an exact copy of $\execNFW$; the copy $\execBSWbrd$ has only the \green
edges left; and $\execBSWnobrd$ has only the \locr edges left (and in particular does not have broadcast edges). Furthermore, for every edge $(s, \sigma, s')$ in ${\execBSWinit}$ we add two new edges, both with the same source as the original edge, but one going to the copy of $s'$ in $\execBSWbrd$, and one to the copy of $s'$ in $\execBSWnobrd$. The initial states of $\execBSW$ are the initial states of $\execBSWinit$. The (single) counter increments at every transition in $\execBSWinit$, increment at every \sgood rendezvous edge in $\execBSWbrd$, and resets at every broadcast edge in $\execBSWbrd$.

Here is the formal definition.

\begin{definition}[B-Automaton $\execBSW$]
Let us introduces the process template:
\[
\proctemp = \tup{\AP,\Actionprts \cup \{\brd\}, S,I,R,\lambda}
\]
and let $\puwd = \tup{AP,\Actionprts \cup \{\brd\},S^\uwd,I^\uwd,R^\uwd,\lambda^\uwd}$ be its unwinding.
Define the B-automaton $\execBSW = \tup{\Sigma,S',I',R',\buchiset, cc}$ as follows:

\begin{itemize}
\item $\Sigma = 2^{AP}$,
\item $S' = \{\init,\grn,\loc\} \times S^\uwd$, 
\item $\buchiset = S'$, i.e., the B\"uchi condition is always satisfied,
\item $I' = \{\init\} \times I^\uwd$,
\item The transition relation $R'$ is $\delta_{\init} \cup \delta_\grn \cup \delta_\loc$ and the counter function $cc$ are defined as follows. For every transition $e = (s,\sigma,t) \in R^\uwd$
\begin{enumerate}
 \item $\delta_{\init}$ contains the transitions $\tau = ((init,s),\lambda^{\uwd}(s),(i,t))$ for every $i \in \{\init,\grn,\loc\}$;
 and $cc(\tau) = \inc$.

  \item $\delta_\grn$ contains the transition $\tau = ((\grn,s),\lambda^{\uwd}(s), (\grn,t))$ only if $e$ is $\green$; $cc(\tau) = \inc$ if $e$ is a \sgood rendezvous edge, $cc(\tau) = \reset$ if $e$ is a broadcast edge, and otherwise $cc(\tau) = \skp$;

 \item $\delta_\loc$ contains the transition $\tau = ((\loc,s),\lambda^{\uwd}(s),(\loc,t))$ only if $e$ is \locr; $cc(\tau) = \skp$.
\end{enumerate}
\end{itemize}
\end{definition}

Since transitions of $\execBSW$ are induced by transitions of $\puwd$, we call transitions of $\execBSW$ broadcast, $\good$, etc., based on the classification of the corresponding transition of $\puwd$. Figure~\ref{fig:b_automaton} depicts a high level view of the B-automaton constructed from a process template $P$, assuming that $(s,\sigma_1,t)$ is a $\sgood$ rendezvous edge, $(s,\sigma_2,u)$ is a $\green$ broadcast edge, $(s,\sigma_3,v)$ is a $\good$ rendezvous edge, and $(s, \sigma_4, w)$ is a locally-reusable edge (for some $\sigma_1,\sigma_2,\sigma_3$, and $\sigma_4$).

\begin{figure}
	\resizebox{0.65\linewidth}{!}{\begin{tikzpicture}[->,>=stealth',shorten >=2pt,auto,node distance=2cm,
                    semithick]

    \tikzstyle{invisible}=[]
    \tikzstyle{vertex}=[circle,fill=black!25]


    \draw (0,0) ellipse (2cm and 4cm);

    \node[vertex] (S1) at (-1,0.25)        {$s$};
    \node[vertex] (T1) at (0.5,2.5)     {$t$};
    \node[vertex] (U1) at (1,1)         {$u$};
    \node[vertex] (V1) at (1,-0.5)      {$v$};
    \node[vertex] (W1) at (0.5,-2)      {$w$};
    \node[invisible] (dots1) at (-0.5, -3) {\ldots};

    \node[invisible] (Linit) at (0, 4.5) {$\execBSWinit$};

    \path (S1) edge [bend left=10]  node [pos=0.75,xshift=6pt,yshift=-1pt] {inc} (T1);
    \path (S1) edge                 node [pos=0.75,yshift=-5pt] {inc} (U1);
    \path (S1) edge     node[yshift=-5pt] {inc} (V1);
    \path (S1) edge [bend right=10]    node [pos=0.75,yshift=-6pt] {inc} (W1);


    \draw (4.5,0) ellipse (2cm and 4cm);

    \node[vertex] (S2) at (3.5,1)       {$s$};
    \node[vertex] (T2) at (5,2.5)       {$t$};
    \node[vertex] (U2) at (5.5,1)       {$u$};
    \node[vertex] (V2) at (5.5,-0.5)    {$v$};
    \node[vertex] (W2) at (5,-2)        {$w$};
    \node[invisible] (dots2) at (5, -3) {\ldots};

    \node[invisible] (Lgrn) at (4.5, 4.5) {$\execBSWbrd$};

    \path (S2) edge [bend left]     node {inc}      (T2);
    \path (S2) edge                 node {reset}    (U2);
    \path (S2) edge [bend right]    node {skip}     (V2);

    \path (S1) edge [bend left=80]     node [pos=0.6] {inc}      (T2);
    \path (S1) edge [bend left=45]     node [pos=0.55] {inc}      (U2);
    \path (S1) edge [bend right=45]     node [pos=0.47] {inc}      (V2);
    \path (S1) edge [bend right=80]     node [pos=0.6] {inc}      (W2);


    \draw (9,0) ellipse (2cm and 4cm);

    \node[vertex] (S3) at (8,1) {$s$};
    \node[vertex] (W3) at (9.5,-2) {$w$};

    \node[invisible] (dots3) at (8.5, -3) {\ldots};

    \node[invisible] (Lgrn) at (9, 4.5) {$\execBSWnobrd$};

    \path (S3) edge [bend right]    node {skip}     (W3);

    \path (S1) edge [bend right=100]    node [pos=0.75] {inc}     (W3);

\end{tikzpicture}}
  \vspace*{-20pt}
	\caption{\label{fig:b_automaton}A high level view of the B-automaton construction.}
\end{figure}
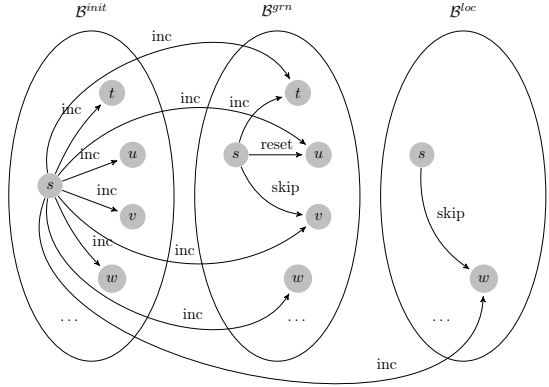

The rest of this section is concerned with proving that the construction is correct:
\begin{theorem}\label{thm: BSW correctness}
For every RB-template $P$ the language of B-automaton $\execBSW$ is $\execinf$.
\end{theorem}

\paragraph{Right-to-left direction.} Take $\alpha \in \execinf$. By Corollary~\ref{cor:exec psys and puwdsys} we have that
$\exec[\psys] = \exec[\puwdsys]$. Thus,
there is some $\pi \in \puwd$ such that $\alpha = \lambda^\uwd(\pi)$.
We now show that there is an accepting run of $\execBSW$ on $\alpha$. First note that by Definition~\ref{dfn: automaton for finite executions} and Proposition~\ref{prop: unwinding captures executions}, for every prefix of $\alpha$ there is a finite run on that prefix that remains in $\execBSWinit$.  There are two cases: either $\pi$ contains infinitely many broadcast transitions or not. If it does not, then from some point on all edges on $\pi$ are \locr. Thus, at that point, the automaton can move from $\execBSWinit$ to $\execBSWnobrd$. The resulting run is accepting since the counter is never incremented in $\execBSWnobrd$. On the other hand, if $\pi$ has infinitely many broadcast transitions, then from some point on, all its edges are \green. Thus, at that point, the automaton can move from $\execBSWinit$ to $\execBSWbrd$. Observe that the counter is reset on every broadcast edge and it is only incremented on \sgood edges, which by Def. \ref{def:edge-types}, appear only boundedly many times between broadcasts.

\paragraph{Outline of left-to-right direction.} Let $\Omega$ be an accepting run in $\execBSW$ on input $\alpha$. By Corollary~\ref{cor:exec psys and puwdsys} it is enough to construct a run $\pi$ in $\puwdsys$ 
whose projection on process $1$ has labeling $\alpha$. Let $\beta$ be the run in $\puwd$ induced by $\Omega$ (recall that every transition of the automaton is induced by a transition of $\puwd$). The construction of $\pi$ is guided by having process $1$ trace $\beta$. We decompose $\beta = \beta' \cdot \beta''$ where $\beta'$ corresponds to the finite prefix of the run $\Omega$ that stays in $\execBSWinit$.

In order to trace $\beta'$, we use the techniques in Section~\ref{sec:solving finite} for finite traces. This leaves us with the task of tracing $\beta''$ which contains either only locally reusable edges or only \green edges. First, observe that tracing a broadcast edge is easy since we can simply append a global broadcast transition to the run $\pi$ being constructed. On the other hand, for each rendezvous edge $e$ we will assign multiple groups of processes to help process $1$ traverse $e$ (the number of groups is discussed later). Each group associated with $e$ has the property that it can trace a \emph{cycle} $C_e$ in which edge $e$ is taken at some point, say by process $p_e$. So, if $e$ is the next edge that process $1$ should take, we progress some group along the cycle $C_e$ to the point where $e$ should be taken, then process $1$ swaps places with process $p_e$ (this is virtual, and merely re-assigns process ids);
and then the group takes the next transition along the cycle $C_e$, and so process $1$ takes $e$. Note that in order for a group to be available to assist process $1$ again in the future, it has to be `reset', i.e., put back to the same position just before the edge $e$ was taken. This is done differently, depending on the type of $e$.
If $e$ is \locr, then so are all subsequent edges $f,g,\dots$ that process $1$ should take; so, since $C_e,C_f,C_g, \dots$ do not contain any broadcast, the group can simply loop around $C_e$ immediately after process $1$ leaves $C_e$; when process $1$ does leave $C_e$, it swaps with $p_f$ in $C_f$, and so on.
If $e$ is \good, then it is on an inner cycle $D_e$, without broadcasts, of $C_e$ (Lemma~\ref{lem: edge types and cycles}), so the group can loop around $D_e$ after process $1$ swaps out --- we call this \emph{recharging} --- thus enabling it to help process $1$ again even though it has not yet completed the outer cycle $C_e$.
Finally, if $e$ is \sgood, this group will only be ready again after the whole cycle $C_e$ is looped once more, which requires waiting for $K$ broadcasts. Thus, until that happens, if process $1$ needs to trace $e$ it will need the help of another group associated with $e$. The key observation is that the number of these groups is bounded. The reason for this is that $\Omega$ is an accepting run, and thus one can deduce that there is a bound on the number of times a \sgood edge is taken until the $K$ broadcasts needed to complete the cycle $C_e$ are taken.

\paragraph{Detailed proof of left-to-right direction.}
Let $\Omega$ be an accepting run in $\execBSW$ on input $\alpha$. Since every transition in $\execBSW$ corresponds to a transition in $\puwd$, let $\beta$ be the corresponding run in $\puwd$. Since $\Omega$ is an accepting run, it either gets trapped in $\execBSWbrd$ or it gets trapped in $\execBSWnobrd$. Decompose $\Omega = \Omega' \cdot \Omega''$ accordingly, i.e., the prefix $\Omega'$ corresponds to the run until it first 
enters $\execBSWnobrd$ or $\execBSWbrd$. Decompose $\alpha = \alpha' \cdot \alpha''$ and $\beta = \beta' \cdot \beta''$ accordingly.

We are required to construct a run of $\psys$ whose projection on process $1$ is labeled $\alpha$. By Corollary~\ref{cor:exec psys and puwdsys} it is enough to construct a run $\pi$ of $\puwdsys$ whose projection onto process $1$ is labeled $\alpha$.
We first construct a finite run $\rho'$ of $\puwdsys$  whose projection on process $1$ is $\beta'$. Since $\Omega'$ stays inside $\execBSWinit$ it is actually also an accepting run of $\execNFW$ on $\alpha'$. Thus, by Corollary~\ref{cor: execNFW recognises execfin}, $\alpha' \in \execfin$. By Corollary~\ref{cor:exec psys and puwdsys} it is also in $\execfin[\puwdsys]$, i.e., there is a run $\rho'$ of $(\puwd)^t$ for some number $t$ of processes whose projection onto process $1$ is $\alpha'$. Let $\puwd_l$ be the component that the run $\rho'$ ends in.

To complete the proof, we will construct an infinite path $\rho''$ of $(\puwd)^n$ for some number $n$ of processes, satisfying the following:
1) its projection on process $1$ is $\beta''$ (note that this implies that $\rho''$ starts in a configuration where process $1$ is in the same state as when it ended $\rho'$), and 2) it starts in a configuration in $\puwd_l$. To see why this is enough to complete the proof, proceed as follows in order to compose $\rho'$ and $\rho''$.
Apply Lemma~\ref{lem: loading} (Loading) to get a finite run $\rho$ in $\puwdsys$ that has the same number of broadcasts as $\rho'$, and ends in a configuration that, when restricted to the first $n$ processes, is the starting configuration of $\rho''$. Note that $\rho'$ may use $m > n$ processes in order to achieve that. By Lemma~\ref{lem: rb-system composition}, we can simultaneously simulate both $\rho$ and $\rho'$ in a run $\pi'$ of  $(\puwd)^{t+m}$. Assume w.l.o.g. that the first $t$ processes are simulating $\rho'$, and that the next $n$ processes are simulating the first $n$ processes of $\rho$. Thus, at the final configuration of $\pi'$, these $n$ processes are exactly in the states needed to start simulating $\rho''$, and processes $1$ and $t+1$ are in the same state. Thus, we can extend the simulation by letting process $1$ exchange roles with process $t+1$ and having processes $1, t+2, \dots t+n$ simulate $\rho''$ (with all other processes doing nothing except responding to broadcasts). The resulting run $\pi$  has the property that its projection onto process $1$ is labeled $\alpha$, as promised.

\paragraph{Constructing $\rho''$}
First, assume w.l.o.g. that $\proctemp$ (and thus also $\puwd$) has no self loops~\footnote{A template can be transformed, in linear time, to one without self loops (and the same set of executions) as follows: for every state $s$ that has a self loop, add a new state $\hat{s}$ with the same labeling, replace every self loop $(s, \sigma, s)$ with $(s, \sigma, \hat{s})$, and for every outgoing transition $(s, \sigma', t)$, including self-loops, add the transition $(\hat{s}, \sigma', t)$.} --- this is not essential, but simplifies some technicalities in the construction.
Second, we differentiate between two cases, depending in which component $\Omega''$ is trapped. We treat the case that $\Omega''$ is trapped in $\execBSWbrd$ (the case that $\Omega''$ is trapped in $\execBSWnobrd$ is simpler, and does not use any new ideas). Note that we will ignore the technicality of keeping track of process numbers, as we find it distracts, rather then helps one understand the proof.

 Let $E_\good$ (resp. $E_\sgood$) be the set of \good\ (resp. \sgood) edges that appear on $\beta''$, and note that these are the only edges that appear on it (by the fact that $\execBSWbrd$ contains only \green edges).
For every edge $e \in E_\sgood$ (resp. $e \in E_\good$), let $C_e$ (resp. $C_e, D_e$) be the witnessing cycle(s) with exactly $K$ broadcasts (for some fixed $K$) from Lemma~\ref{lem: cycle implies one with K broadcasts}, and assume w.l.o.g. that (a) every such cycle $C_e$ starts in a configuration in the component $\puwd_l$ in which $\rho'$ ends, and (b) that if $e'$ is the first edge on $\beta''$, then $e'$ appears in the first transition taken in the cycle $C_{e'}$.
To see how to achieve (a) note that if $e'$ and each subsequent edge is \locr, then because each $C_e$ does not contain any broadcast, each $C_e$ is contained in $\puwd_l$. On the other hand, if the edges $e$ are \green, then $e$ must be on the loop of the lasso, and so $\puwd_l$ is on the loop of the lasso (since $e'$ is), and so since $C_e$ contains at least one broadcast, it must go through $\puwd_l$.
Since $\Omega$ is an accepting run, the counter is bounded on it. Thus, since in $\execBSWbrd$ we increment the counter when reading a \sgood\ edge, and reset it when reading a broadcast edge, we can pick $\mathfrak{m} \in \Nat$ such that every \sgood\ edge appears at most $\mathfrak{m}$ times on any section of $\beta$ that contains $K$ broadcasts.

Let the \emph{designated occurrence} of $e$ on $C_e$ be defined as follows: if $e$ is \sgood\ then it is the first transition in $C_e$ in which $e$ is taken, and if $e$ is \good\ then it is the first transition of $C_e$, that is also on the nested cycle $D_e$, in which $e$ is taken.
For every $h \in [K]$, let $C_e(h)$ denote the portion of $C_e$ just after the $h-1$ broadcast up to (and including) the $h$ broadcast. For $h$ such that $C_e(h)$ contains the designated occurrence of $e$ we divide $C_e(h)$ further into three pieces: $C_e(h,1)$ is the part up to the designated occurrence, $C_e(h,2)$ is the designated occurrence, and $C_e(h,3)$ is the remainder.

Take exactly enough processes to assign them to one copy $G_e^1$ of $C_e$ for every $e \in E_\good$, and $\mathfrak{m}$ copies $G_e^1, \dots G_e^{\mathfrak{m}}$ of $C_e$ for every $e \in E_\sgood$.

Given a group of processes $G_e^i$, for some $i$ and $e$, we define the following operations:
\begin{itemize}
  \item \textbf{flush}: $G_e^i$ simulates (using Lemma~\ref{lem: rb-system composition})
  the portion of $C_e(h)$ which it has not yet simulated, up to but not including the broadcast;
  \item \textbf{load}: in case $C_e(h)$ contains the designated occurrence of $e$, then $G_e^i$ simulates (using Lemma~\ref{lem: rb-system composition}) the path $C_e(h,1)$;
  \item \textbf{swap}: we say that we \emph{swap} process $1$ into $G_e^i$ to mean that process $1$ and process $j$ in $G_e^i$ (where $j$ is a process that takes the edge $e$ in the designated occurrence of $e$ on $C_e$) exchange their group associations. I.e., process $1$ joins $G_e^i$, and process $j$ takes its place in the former group of process $1$;
  \item \textbf{recharge}: if $e$ is \good, we say that we \emph{recharge} $C_e$ to mean that the group $G_e^i$ simulates (using Lemma~\ref{lem: rb-system composition}) tracing $D_e$ until reaching (but not executing) the designated occurrence of $e$.
  \item \textbf{mark}: if $e$ is \sgood, we may mark a group $G_e^i$ as {\em used} or {\em fresh} by setting or resetting a virtual flag.
\end{itemize}

Obviously (except for flush and mark), not every operation above can be taken at any time. In particular, a swap is allowed only at a time process $1$ and $j$ are in the same state, e.g., when process $1$ is at the source of $e$ and $G_e^i$ has just been loaded (and thus process $j$ is also at the source of $e$).

We are now ready to construct $\rho''$. The initial configuration of $\rho''$ is obtained by having all processes at the beginning of the cycles they were assigned to, and process $1$ assigned to the group $G^1_{e'}$ where $e'$ is the first edge on $\beta''$ (this can be done because $e'$ appears in the first transition taken in the cycle $C_{e'}$). Note that this satisfies the requirement that $\rho''$ starts in a configuration in $\puwd_l$. The rest is done in blocks, where in the $i$'th block we extend $\rho''$ with a path $\xi_i$ containing $K$ broadcasts and whose projection on process $1$ is the portion of $\beta''$ after the $(i-1)K$ broadcast up to and including the $iK$ broadcast --- which we call $\beta''_i$. The construction will ensure that $\xi'$, defined as the concatenation of the $\xi_i$s, weakly-simulates all the cycles of all the groups, and thus by Remark~\ref{rem: weakly-simulates}, will maintain the following invariant $(\dagger)$: at the configuration $f$ at the start of each block, the processes in every group of the form $G_e^i$ are in states corresponding to the initial configuration $g$ of $C_e$ (i.e., $\restr{f}{G_e^i} = g$). The invariant obviously holds for the first block by our choice of the initial configuration of $\rho''$.

For $i \in \Nat$, assume that blocks $< i$ have been constructed. We now describe how to build block $i$. First, \textbf{mark} all the groups as \emph{fresh}, then proceed in $K$ rounds by repeating the following algorithm for every $ 1 \leq h \leq K$.
Let $e_1 \dots e_x$ be the prefix of $\beta''_i$ not yet traced by process $1$, up to and including the next broadcast (obviously, the length $x$ of this prefix depends on $h$).
For every $j \in [x]$, if $e_j$ is \good pick the group $G_{e_j}^1$; and if it is \sgood pick the first \emph{fresh} group from among the yet unpicked groups in $G_{e_j}^1 \dots G_{e_j}^{\mathfrak{m}}$, and \textbf{mark} this group as \emph{used} (we can always pick a fresh group since --- by our choice of $\mathfrak{m}$, and since $\beta''_i$ has exactly $K$ broadcasts --- there are at most $\mathfrak{m}$ occurrences of $e_j$ in $\beta''_i$). Denote the group thus picked by $\mathbb{G}_{e_j}$, and let $\mathbb{G}_{e_0}$ denote whatever group process $1$ is in at the beginning of the round.

\begin{figure}	\resizebox{0.75\linewidth}{!}{\begin{tikzpicture}[->,>=stealth',shorten >=2pt,auto,node distance=2cm,
                    semithick]

    \tikzstyle{invisible}=[]
    \tikzstyle{vertex}=[circle,fill=black!25]

    \node[invisible]        at (0,0)    {$\beta$};
    \node[invisible]        at (1.25,0)    {\dots};
    \node[invisible]        at (5.5,0)    {$e_j$};
    \node[invisible]        at (7.5,0)    {$e_{j+1}$};

    \node[invisible]        at (9.5,0)    {\dots};

    \draw [-,dashed]        (-1,-1) -- (12,-1);


    \node[invisible]        at (0,-2)    {$\rho''$};
    \node[invisible]        at (1.25,-1.5)   {\vdots};
    \draw                   (0.75,-4.5) rectangle ++(1,2); 
    \node[invisible]        at (1.25,-3.5)   {$\mathbb{G}_{e_j}$};
    \draw                   (1.75,-3.5) -- node {load} (4,-3.5);
    
    \draw                   (0.75,-7) rectangle ++(1,2); 
    \node[invisible]        at (1.25,-6)   {$\mathbb{G}_{e_{j+1}}$};
    \draw                   (1.75,-6) -- node {load} (6,-6);

    \node[invisible]        at (1.25,-7.75)   {\vdots};

    \node[invisible]        at (4.5,-1.5)   {\vdots};
    \draw                   (4,-4.5) rectangle ++(1,2);
    \node[invisible]        at (4.5,-3) {$\cdot$};                                      
    \draw                   (4.5,-3) -- (6.5,-3);                                       
    \draw[<->]              (4.5,-2) -- node [pos=0.25] {swap process 1} (4.5,-3);      

    \draw                   (5,-3.5) -- node {$e_j$} (6,-3.5);

    \draw                   (6,-4.5) rectangle ++(1,2);
    \node[invisible]        at (6.5,-3) {$\cdot$};                                      
    \draw[<->]              (6.5,-3) -- node [pos=0.725] {swap process 1} (6.5, -5.5);    
    \draw                   (7,-3.5) -- node {flush} (11,-3.5);

    \draw                   (6,-7) rectangle ++(1,2);
    \node[invisible]        at (6.5,-5.5) {$\cdot$};                                    
    \draw                   (7,-6) -- node {$e_{j+1}$} (8,-6);

    \draw                   (8,-7) rectangle ++(1,2);

    \draw                   (6.5,-5.5) -- (8.5,-5.5);                                   
    \node[invisible]        at (8.5,-5.5) {$\cdot$};
    \draw[<->]              (8.5,-5.5) -- node [pos=0.85] {swap process 1} (8.5, -7.5);     

    \draw                   (9,-6) -- node {flush} (11,-6);
    \node[invisible]        at (8.5,-7.75)   {\vdots};

    \draw                   (11,-4.5) rectangle ++(1,2);

    \draw                   (11,-7) rectangle ++(1,2);

\end{tikzpicture}}
	\caption{\label{fig:rho-dprime}An illustration of how to simulate some \sgood edges $e_j$ and $e_{j+1}$ of the run $\beta''$ in the run $\rho''$ using groups $\mathbb{G}_{e_j}$ and $\mathbb{G}_{e_j+1}$ .}
\end{figure}
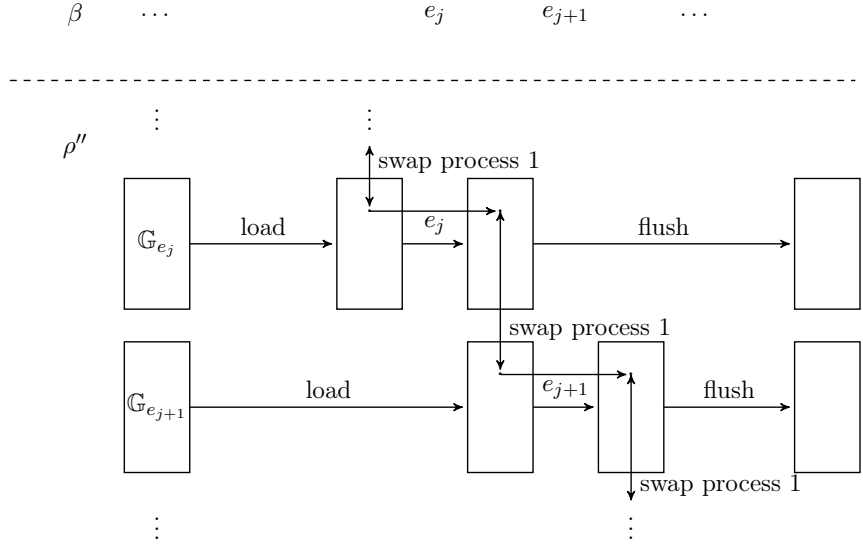

\begin{enumerate}
 \item For $1 \leq j \leq x$ repeat:
     \begin{enumerate}
     \item \textbf{Load} the group $\mathbb{G}_{e_j}$ and \textbf{Swap} process $1$ into it;
     \item if $j < x$ then have group $\mathbb{G}_{e_j}$ simulate the transition $C_{e_j}(h,2)$, with process $1$ taking the edge $e_j$.
     \item If $ j > 1$, and $e_{j-1}$ is a \good edge, then \textbf{recharge} $\mathbb{G}_{e_{j-1}}$.
     \end{enumerate}
 \item \textbf{Flush} all groups except $\mathbb{G}_{e_x}$ (that process $1$ is currently in). Note that since $e_x$ is a broadcast, the loading of $\mathbb{G}_{e_x}$ already put it in a flushed condition.
 \item Perform a broadcast (with process $1$ taking the broadcast edge $e_x$).
\end{enumerate}

We illustrate the execution of the above algorithm for one round in Figure~\ref{fig:rho-dprime}.
Here, we depict two \sgood edges $e_j$ and $e_{j+1}$ of the run $\beta''$ along with the groups $\mathbb{G}_{e_j}$ and $\mathbb{G}_{e_j+1}$ that enable the simulation of these edges in the run $\rho''$. 
The illustration shows that each of the groups is first loaded (in order to prepare for the transition $e_j$ resp. $e_{j+1}$ to be taken), then process 1 is swapped into the group (allowing process 1 to execute the respective transition), after that transition $e_j$ resp. $e_{j+1}$ is taken, and finally each group is flushed to prepare for the next symmetric broadcast (in order to be ready to participate in the next round of the simulation).
Since the invariant $(\dagger)$ holds at the start of each block, it is easy to see that the algorithm can actually be executed. 
Indeed, the invariant ensures that the algorithm can \textbf{load} when needed (and thus, \textbf{swap} and \textbf{recharge} when needed). 
It is not hard to see that, as promised, the resulting path $\xi_i$ weakly-simulates all the cycles of all the groups. To see that the projection of $\xi_i$ on process $1$ is $\beta''_i$, observe that the moves performed by process $1$ in lines $1(b)$ and $3$ of the algorithm trace exactly $\beta''_i$. Furthermore, process $1$ is moved only in these lines since, by our assumption that $\puwd$ does not contain any self loop, every edge in $\beta''$ is different than the edge just before and just after it and thus, process $1$ is never in a group when it is being loaded (except at the very beginning of the first block, in which case the \textbf{load} in line $1(a)$ of the algorithm does nothing since this group is already in a loaded position in the initial configuration of $\rho''$).

This completes the proof of Theorem~\ref{thm: BSW correctness}.

In Section~\ref{sec: deciding types of edges} we will show how to decide the type of the edges in $\puwd$ (Theorem~\ref{thm: edge type decidability}) in polynomial time in the size of $\puwd$. Thus, we can build the B-automaton $\execBSW$ in exponential time in the size of $P$.  Combining this with Theorem~\ref{thm: BSW correctness} we get Theorem~\ref{thm: main dec} that says that the PMCP for NBW/\LTL specifications of RB-systems can be solved in \EXPTIME. 

\begin{proof}[Proof of Theorem~\ref{thm: main dec}]

We reduce the PMCP problem to the emptiness problem for B-automata.

Given a process template $\proctemp$, and the corresponding B-automaton $\execBSW$ (whose B\"uchi set is trivial), suppose the specification is given as an \LTL formula $\varphi$ (the case of NBW is given afterwards). Let $L(\execBSW)$ denote the language of $\execBSW$, and let $L(\lnot \varphi)$ denote the set of models of $\lnot \varphi$. Then, every execution in $\execinf$ satisfies $\varphi$ if and only if $L(\execBSW) \cap L(\lnot \varphi) = \emptyset$.
Using Theorem~~\ref{thm:vardi-wolper}, let $\A_{\lnot \varphi}$ be an NBW accepting all models of $\lnot \varphi$, and denote its states by $Q$ and its B\"uchi set by $\buchiset$. Build the synchronous product of $\A_{\lnot \varphi}$ and $\execBSW$ to get an NBW with one counter, call it $M$, whose language is equal to $L(\execBSW) \cap L(\A_{\lnot \varphi})$. By Lemma~\ref{lem:single-counter-b-automaton-emptiness}, one can test whether $L(M)$ is empty in \PTIME. Thus, this PMCP algorithm is exponential in the size of the template $P$ (since computing $\execBSW$ can be done in time exponential in the size of $P$) and exponential in the size of $\varphi$ (since computing $\A_{\lnot \varphi}$ can be done in time exponential in the size of $\varphi$).

For the case that the specification is an NBW $\A$, we proceed in a similar way by noting that we can build an NBW $\A'$ for the complement of $L(\A)$ in time exponential in the size of $\A$ (see for example~\cite{DBLP:conf/focs/Safra88}).
\end{proof}


\section{Deciding Edge Types}
\label{sec: deciding types of edges}

\newcommand\Dim{k}
\newcommand\states{S}
\newcommand\transRel{R}
\newcommand\transTuple{\mathbf{\mathfrak{t}}}
\newcommand\run{r}
\newcommand\config{c}
\newcommand\B{B}
\newcommand\configAlt{d}
\newcommand\transition{t}
\newcommand\mult{\alpha}
\newcommand\multAlt{\gamma}
\newcommand\broadMultAlt{\delta}
\newcommand\multMsg{\mu}
\newcommand\multMsgAlt{\nu}
\newcommand\firingSetOp{\mathcal{F}}
\newcommand\revFiringSetOp{\mathcal{R}}
\newcommand\forwardOp{\textsf{forw}}
\newcommand\backwardOp{\textsf{back}}
\newcommand\positiveOp[1]{(#1)^{\neq 0}}
\newcommand\positiveSet{H}
\newcommand{\rdz}{\textsf{rdz}}
\newcommand{\incoming}{\textsf{in}}
\newcommand{\outgoing}{\textsf{out}}
\newcommand{\BroadTrans}{Q}
\newcommand{\trace}{\xi}
\newcommand{\traceAlt}{\varrho}

This section is dedicated to proving the following result:

\begin{theorem}
\label{thm: edge type decidability}
  Given a {reachability-unwinding $\puwd$ of a process template $\proctemp$}, the type (\good, \sgood, and \locr) of each edge $e$ in $\puwd$ can be decided in \PTIME (in the size of $\puwd$).
\end{theorem}

We will develop the proof of Theorem~\ref{thm: edge type decidability} in several steps:

\textbf{1.} The starting point for the proof of Theorem~\ref{thm: edge type decidability} is the characterization of edge types in Lemma~\ref{lem: edge types and cycles} through the existence (or lack thereof) of suitable cycles in $\puwdsys$.
In Subsection~\ref{subsec:pseudo-cycles}, we weaken this characterization using the notion of \emph{pseudo-cycles}.
Pseudo cycles are paths that start and end in configurations that are identical up to the renaming of processes, i.e., there are exactly the same number of processes in every state, though the identities of the processes in each of the states may differ.
Pseudo-cycles can always be pumped to a cycle by iterating the pseudo-cycle until the initial configuration is reached again.
Hence, pseudo-cycles can be seen as more compact representations of cycles.
Importantly, we are able to obtain a bound on the number of broadcasts in pseudo-cycles, which we need for deciding edge types.

\textbf{2.} In order to be able to conveniently reason about pseudo-cycles we will work with counter abstractions of R-Systems:
In Subsection~\ref{subsec:cvrs}, we define vector rendezvous systems (VRS) and their continuous relaxation, called continuous vector rendezvous systems (CVRS).
The notions of VRSs and CVRS are inspired by the notion of Vector Addition Systems (VAS)~\cite{journals/tcs/HopcroftP79}, where configurations only store the number of processes for every process state but not the identity of the processes.
VRSs are counter abstractions of R-Systems; in particular, pseudo-cycles of R-Systems correspond to cycles in VRSs and vice versa.
CVRSs are a continuous relaxation of VRSs in which steps can be taken by a rational fraction.
CVRSs have the advantage that we can characterize reachability and the existence of cycles in them by solving linear programming problems over the rationals.
We will be able to work with CVRs instead of VRSs because we will show that we can scale a CVRSs cycle to a VRS cycle (as we are interested in the parameterized verification problem we can always scale the number of processes).
We then reduce the existence of witnessing pseudo-cycles for the type of an edge to corresponding reachability statements for CVRSs, as outlined below.

\textbf{3.} In Subsection~\ref{subsec:characterization-reachability-cvas}, we develop a characterization of reachability for CVRSs.
This characterization will give rise to an equation system and a fixed point algorithm, which is the basis for our edge type computation procedure.
The results in this subsection already allow us to compute the \locr edges and the \good edges (under the assumption that the \green edges are already known; the computation of the \green edges is, however, only done in the next subsection).

\textbf{4.} In Subsection~\ref{subsec:deciding-pseudo-cycle}, we show how to decide whether an edge of $\puwd$ is \green.
This problem represents the main difficulty in deciding the type of an edge.
We give an algorithm that computes all \green edge of $\puwd$ and prove its correctness.
We further remark that, though not actually needed, one can derive a pseudo-cycle that contains all \green edges from our procedure.

\subsection{Pseudo-cycles}
\label{subsec:pseudo-cycles}

\begin{definition} \label{dfn: pseudo-cycles}
Let $\proctemp$ be an RB-template with state set $S$, and let $n \in \Nat$. Two configurations $f,f'$ of $\sysinst{n}$ are called \emph{twins} if every state is covered by the same number of processes in $f$ and $f'$, i.e. $|f^{-1}(q)| = |f^{\prime -1}(q)|$ for every $q \in S$.
\end{definition}
Let $g \circ h$ denote the composition $g(h(\cdot))$. Observe that $f, f'$ are twins if and only if there is a (not necessarily unique) permutation $\theta: [n] \to [n]$ such that $f' = f \circ \theta$. Intuitively, $\theta$ maps each process in $f'$ to a matching process in $f$ (i.e., one in the same state). Thus, given a transition $\trans{f}{g}{\sigma}$, say $t$, in $\sysinst{n}$, we denote by $t[\theta]$ the transition $\trans{f \circ \theta}{g \circ \theta}{\sigma'}$ resulting from replacing process $i$ with process $\theta(i)$ in $t$; and having $\sigma' = \brd$ if $\sigma = \brd$, and $\sigma' = ((\theta^{-1}(i_1), \msg{a}_1), \dots, (\theta^{-1}(i_k), \msg{a}_k))$ if $\sigma = ((i_1, \msg{a}_1), \dots, (i_k, \msg{a}_k))$.  Note that the rendezvous or broadcast action of the transition taken in $t$ and $t[\theta]$ are the same --- it is only the identities of the processes involved that are different. Extend $\theta$ to paths point-wise, i.e., if $\pi = t_1 t_2 \ldots$ is a path then define $\pi[\theta] = t_1[\theta] t_2[\theta] \ldots$.

\begin{definition} \label{dfn:pseudocycle}
A finite path $\pi$ of an RB-system $\sysinst{n}$ is a {\em pseudo-cycle} if $\src(\pi)$ and $\dst(\pi)$ are twins.
\end{definition}

Obviously, every cycle is also a pseudo-cycle, but not vice-versa.
For example, for $\proctemp$ in Figure~\ref{fig: process_example}, the following path in $\sysinst{4}$ is a pseudo-cycle that is not a cycle: $(p,q,q,r) \xrightarrow{((3,\msg{c}_1), (4,\msg{c}_2))} (p,q,r,p) \xrightarrow{((2,\msg{c}_1), (3,\msg{c}_2))} (p,r,p,p) \xrightarrow{((3,\msg{a}_1), (4,\msg{a}_2))} (p,r,q,q)$.

\begin{figure}[!htb]
             {
\begin{tikzpicture}[->,>=latex,node distance=0.4cm,bend angle=25,auto]

\tikzset{every state/.style={circle,minimum size=.4cm,inner sep=0cm}}
\tikzset{every edge/.append style={font=\small}}

\node[initial, state] (t1) {$p$};
\node[state] (t2) [right= of t1] {$q$};
\node[initial, state] (t3) [below= of t1] {$r$};

\path (t1) edge [bend left] node {{$\msg{a}_1$}} (t2);
\path (t1) edge [bend right, below] node {$\msg{a}_2$} (t2);
\path (t3) edge [bend left, left] node {$\msg{c}_2$} (t1);
\path (t2) edge [bend left] node {$\msg{c}_1$} (t3);

\end{tikzpicture}

%
%
%
%
%
%
}{\caption{R-template.}\label{fig: process_example}}
             {
\begin{tikzpicture}[->,>=latex,node distance=0.6cm,bend angle=25,auto]

\tikzset{every state/.style={circle,minimum size=.4cm,inner sep=0cm}}
\tikzset{every edge/.append style={font=\small}}

\node[initial,state] (t1) {$p$};
\node[state] (t2) [right= of t1] {$q$};

\path (t1) edge [loop above] node {{$\msg{a}_1$}} (t1);
\path (t1) edge [loop below] node {{$\brd$}} (t1);
\path (t1) edge [bend left] node {{$\msg{a}_2$}} (t2);
\path (t2) edge [bend left, below] node {$\brd$} (t1);

\end{tikzpicture}

%
%
%
%
%
%
}{\caption{RB-template.}\label{fig: process not regular}}
\end{figure}

\begin{remark}\label{rem: pseudo cycle start}
Similar to cycles, for which one can chose any point on the cycle as its start (and end) point, one can chose any point along a pseudo-cycle as the start point. Indeed, if $C$ is a pseudo-cycle that starts in a configuration $f$ and ends in a twin $f'$, then given any configuration $g$ along $C$ we can obtain a new pseudo-cycle, that uses exactly the same edges (but with possibly different processes taking these edges) as follows: start in $g$ and traverse the suffix of $C$ until $f'$, reassign process id's according to the permutation transforming $f$ to $f'$ and traverse the prefix of $C$ from $f$ to $g$ using these reassigned processes to reach a twin $g'$ of $g$.
\end{remark}

The following immediate lemma states that every pseudo-cycle $\pi$ can be pumped to a cycle.

\begin{lemma}\label{lem: psuedo cycle pumping}
Given a pseudo-cycle (resp. legal pseudo-cycle) $\pi$ in $\sysinst{n}$, and a permutation $\theta$ such that $\dst(\pi) = \src(\pi) \circ \theta$,
then there is a $z \ge 1$ such that
\[\pi[\theta^0] \pi[\theta^1] \pi[\theta^2] \ldots \pi[\theta^{z-1}],\]
is a cycle (resp. legal cycle) in $\sysinst{n}$, where $\theta^j$ denotes the composition of $\theta$ with itself $j$ times.
\end{lemma}
\begin{proof}
  We can choose $z$ such that the permutation $\theta^z$ is the identity (recall that the set of permutations of a finite set forms a finite group, and that the order of every element in a finite group is finite, i.e., there is some $z \in \Nat$ such that $\theta^z$ is the identity permutation).
  We observe that $\dst(\pi[\theta^0] \pi[\theta^1] \pi[\theta^2] \ldots \pi[\theta^{z-1}]) = \src(\pi) \circ \theta^z = \src(\pi)$, and thus it is a cycle. 
\end{proof}

Recall that $n,r$ denote the prefix-length and period of $\puwd$, respectively.
The following lemma states that we can assume that if an edge of $\puwd$ appears on a pseudo-cycle with broadcasts then it also appears on one with exactly $r$ broadcasts.
Knowing this bound will be crucial for decidability of edge types (in contrast, Lemma~\ref{lem: edge types and cycles} only says that a bound exists), as well as for obtaining good complexity for deciding PMCP.

\begin{lemma}[Spiral]\label{lem: pseudo cycle implies one with r broadcasts}
An edge $e$ appears on a legal pseudo-cycle $D$ in $\puwdsys$, which contains broadcasts, iff it appears on a legal pseudo cycle $C$ of $\puwdsys$, containing exactly $r$ broadcasts and starting in a configuration in $\puwd_n$.
\end{lemma}
\begin{proof}
Assume $D$ is a legal pseudo-cycle in $(\puwd)^m$, for some $m \in \Nat$.
By Remark~\ref{rem: pseudo cycle start} we can assume w.l.o.g. that $D$ starts and ends in a configuration in $\puwd_n$.
Observe that (by the lasso structure of $\puwd$) $D$ must have $lr$ broadcast transitions for some $l \in \Nat$, and that after every $r$ broadcasts all processes are in $\puwd_n$.
Let $f_0$ be the initial configuration on $D$.
For every $i \in [l]$, let $f_i$ be the configuration in $D$ just after $ir$ broadcasts, and let $\rho_i$ be the portion of $D$ from $f_{i-1}$ to $f_i$.
Observe that $\rho_i$ contains exactly $r$ broadcasts.
By Lemma~\ref{lem: rb-system composition}, we can compose the paths $\rho_1, \dots \rho_l$ into a single path $C$ in the system $(\puwd)^{ml}$ with $ml$ processes.
It is not hard to see that $C$ is a legal pseudo-cycle.
Indeed, after $r$ broadcasts the processes that were simulating $\rho_i$ are in states that match the configuration $f_{(i+1) \bmod l}$, i.e., the configuration that the processes that simulate $\rho_{(i+1) \bmod l}$ started in.
Obviously, $e$ appears on $C$, which completes the proof.
\end{proof}

We now use the notion of a pseudo-cycle to give a characterization of the edges types that is easier to detect than the one given in Section~\ref{sec:solving infinite}:

\begin{lemma}\label{lem: edge types and pseudo cycles}
An edge $e$ of $\puwd$ is:
\begin{enumerate}[i.]
  \item \locr iff it appears on a legal pseudo-cycle $C_e$ of $\puwdsys$ without broadcasts. \label{lem: edge types and pseudo cycles: locr}
  \item \green iff it appears on a legal pseudo-cycle $C_e$ of $\puwdsys$ with $r$ broadcasts. \label{lem: edge types and pseudo cycles: green}
  \item \good\ iff it appears on a legal pseudo-cycle $C_e$ of $\puwdsys$ without broadcasts that only uses \green edges; \label{lem: edge types and pseudo cycles: good}
  \item \sgood\ iff it is \green and does not appear on a legal pseudo-cycle $C_e$ of $\puwdsys$ without broadcasts that only uses \green edges. \label{lem: edge types and pseudo cycles: sgood}
\end{enumerate}
\end{lemma}
\begin{proof}

Item~\ref{lem: edge types and pseudo cycles: sgood} follows from item~\ref{lem: edge types and pseudo cycles: good} since \green edges are partitioned to \good and \sgood edges.

Items~\ref{lem: edge types and pseudo cycles: locr} and \ref{lem: edge types and pseudo cycles: green} follow immediately from Lemmas~\ref{lem: edge types and cycles} and~\ref{lem: psuedo cycle pumping} and the fact that every cycle is also a pseudo-cycle (for the `only if' direction of item~\ref{lem: edge types and pseudo cycles: green} use Lemma~\ref{lem: pseudo cycle implies one with r broadcasts} to obtain a pseudo-cycle with exactly $r$ broadcasts). The same fact, combined with Lemma~\ref{lem: edge types and cycles}, gives the 'only if' direction of item~\ref{lem: edge types and pseudo cycles: good}.

For the `if' direction of item~\ref{lem: edge types and pseudo cycles: good}, we claim that it is enough to show that: ($\dagger$) there is a configuration $f$ of $\puwdsys$ and two pseudo-cycles $C, D$ starting in $f$, such that $D$ contains broadcasts, and $C$ contains the edge $e$ and does not contain any broadcast edge. Indeed, by Lemma~~\ref{lem: psuedo cycle pumping}, one can pump $C,D$ to obtain cycles $C', D'$ starting and ending in $f$, and the `figure-eight' cycle obtained by concatenating $C'$ to $D'$ is (according to Lemma~\ref{lem: edge types and cycles}) a cycle witnessing that the edge $e$ is \good.

We now show that if $e$ appears on a pseudo-cycle $C_e$ satisfying item~\ref{lem: edge types and pseudo cycles: good} then $(\dagger)$ holds.
Let $g$ be the initial configuration of $C_e$ and let $H := \{ s \mid g(s) > 0 \}$ be the set of states of $\puwd$ for which there is some process in that state in $g$. Assume w.l.o.g. that there is no process that does not move on $C_e$ (such processes can simply be removed), and for every $s \in H$, let $e^s$ be some edge in $\puwd$ appearing in $C_e$ whose source is $s$. Observe that by item~\ref{lem: edge types and pseudo cycles: good} the edge $e^s$ is green, and apply Lemma~\ref{lem: edge types and cycles} to obtain a witnessing cycle $C_{e^s}$ starting with some process in $s$.
By Lemma~\ref{lem: cycle implies one with K broadcasts}, we can assume that the cycles $C_{e^s}$ thus obtained for all states in $H$ have the same number of broadcasts. Hence, by Lemma~\ref{lem: rb-system composition}, we can run together $g(s)$ copies of $C_{e^s}$, for all $s$ in $H$, in one cycle $D$. Let $f$ be the starting configuration of $D$. Note that $f(s) \geq g(s)$ for every state $s$ in $\puwd$, and obtain a pseudo-cycle $C$ starting in $f$ which simulates $C_e$ (since $C_e$ does not have any broadcast, processes that $f$ has in excess of $g$ can simply not move). The proof is complete by noting that $f, C$ and $D$ satisfy $\dagger$.
\end{proof}

\subsection{Vector Rendezvous Systems}
\label{subsec:cvrs}

We now formally introduce VRSs and CVRSs.
We recall that $k$ denotes the number of processes participating in a rendezvous, $\Actions$ denotes a finite set of {\em rendezvous actions}, and
$\Actionprts = \cup_{\msg{a} \in \Actions} \{ \msg{a}_1, \dots, \msg{a}_k\}$ denotes the \emph{rendezvous alphabet}.

Given a finite set $\states$, we can think of $\Rat^\states$ as the set of
rational vectors of dimension $|\states|$, and we use the elements of $\states$ as indices into these vectors.
We also use the standard operations of vector addition and scalar multiplication. Finally, we compare vectors point-wise, i.e., given $\config, \config' \in \Rat^\states$ we say that $\config \leq \config'$ iff $\config(s) \leq \config'(s)$ for all $s \in \states$.

\head{Continuous Vector Rendezvous System (CVRS)}
A \emph{continuous vector rendezvous system}  (CVRS) is a tuple $\vas = \tup{\Actionprts,\states,\transRel}$, where $\states$ is a finite set of states and $\transRel \subseteq \states \times \Actionprts \times \states$ is a finite set of \emph{transitions}.
The \emph{configurations} of $\vas$ are the vectors $\RatGEZ^\states$.
For a transition $\transition = (p,\sigma,q)$, we denote by $\rdz(\transition) = \sigma$ its rendezvous symbol,
by $\src(\transition) = p$ the \emph{source} state of $\transition$, and
by $\dst(\transition) = q$ the \emph{destination} state of $\transition$.
Also, we denote by $\outgoing(\transition) \in \{0,1\}^\states$ the vector that has a $1$ entry at index $p$ and zero entries otherwise, and by $\incoming(\transition) \in \{0,1\}^\states$ the vector that has a $1$ entry at index $q$ and zero entries otherwise.
We now define what it means for a CVRS to take a step.

\begin{definition} [Step of a CVRS]
\label{dfn:step}
Given a $k$-tuple $\transTuple = (\transition_1, \ldots,\transition_k)$ of transitions in $R$, the CVRS $\vas$ can \emph{step} with \emph{multiplicity} $\mult \in \RatGZ$ from a configuration $\config$ to a configuration $\config'$ using transitions $\transTuple = (\transition_1,\ldots,\transition_k)$,
 denoted $\config \xrightarrow{\transition_1,\ldots,\transition_k: \mult} \config'$, or $\config \xrightarrow{\transTuple: \mult} \config'$,
if:
\begin{enumerate}[(i)]
  \item there is an action $\msg{a} \in \Actions$ such that $\rdz(\transition_i) = \msg{a}_i$ for all $i \in [k]$,
  \item $\config \ge \mult \sum_{i=1}^{k} \outgoing(\transition_i)$, and
  \item $\config' = \config + \mult \sum_{i=1}^{k} \left( \incoming(\transition_i) - \outgoing(\transition_i) \right)$.
\end{enumerate}
A step is said to \emph{synchronize} on $\msg{a}$.
We say that a transition $\transition \in \transRel$ \emph{participates} in a step $\config \xrightarrow{\transition_1,\ldots,\transition_k: \mult} \config'$ if $\transition = \transition_i$ for some $i \in [k]$.
\end{definition}
A \emph{trace} of $\vas$ is a sequence of steps $\config_1 \xrightarrow{\transTuple_1: \mult_1} \config_2 \xrightarrow{\transTuple_2: \mult_2} \cdots \config_n$.
We say that a configuration $\config'$ is \emph{reachable} from configuration $\config$, denoted $\config \rightarrow^\star \config'$, if there is a trace $\config_1 \xrightarrow{\transTuple_1: \mult_1} \config_2 \xrightarrow{\transTuple_2: \mult_2} \cdots \config_n$, with $\config = \config_1$ and $\config' = \config_n$.

\head{Vector Rendezvous System (VRS)}
A \emph{vector rendezvous system (VRS)}, is a restriction of a CVRS, where the set of configurations is $\NatZero^\states$ and all steps are restricted to have multiplicity $\alpha = 1$.
Given any VRS $\vas = \tup{\Actionprts, \states, \transRel}$, the \emph{relaxation} of $\vas$ is the same tuple interpreted as a CVRS. Since a natural number is also rational, it is obvious that any VRS $\vas = \tup{\Actionprts,\states,\transRel}$ is also a CVRS. 

Note that every trace of a VRS is also a trace of its relaxation, and that the relaxation has more traces than the VRS since every rational multiple of a step in the VRS is a step in its relaxation.

\begin{remark} \label{rem: VRS for template}
Every R-template $\proctemp = \tup{\AP,\Actionprts,S,I,R,\lambda}$ defines a VRS $\vas = \tup{\Actionprts,\states,R}$
with the same set of rendezvous alphabet, states, and transitions.
These two systems are closely related.
Intuitively, $\vas$ is an abstraction of $\psys$ in the sense that it does not keep track of the state of every individual process, but only keeps track of the number of processes in every state.
More formally, every configuration $f \in \psys$ induces a configuration $\config$ of $\vas$, called its \emph{counter representation}, defined as
$\config(s) := |f^{-1}|(s)$ for $s \in S$, i.e., $\config(s)$ is the number of processes of $f$ that are in state $s$.
Furthermore, $\config \xrightarrow{t_1, \cdots, t_k} \config'$ is a step of $\vas$ if and only if there is a global transition $(f,\sigma,f')$ in $\psys$ such that $\config$ and $\config'$ are the counter representations of $f$ and $f'$ respectively, and $t_1, \cdots, t_k$ are exactly the rendezvous edges taken by the $k$ active processes in the transition $(f,\sigma,f')$.
\end{remark}

We now define operations for manipulating traces.
Given a trace $\trace := \config_1 \xrightarrow{\transTuple_1: \mult_1} \config_2 \xrightarrow{\transTuple_2: \mult_2} \cdots \config_n$ of a CVRS $\vas$, we define the following two operations:
\begin{enumerate}
  \item Multiplication by a scalar $0 < \multAlt$: Let $\multAlt \otimes \trace$ be the trace $\multAlt \config_1 \xrightarrow{\transTuple_1: \multAlt \mult_1} \multAlt \config_2 \xrightarrow{\transTuple_2: \multAlt \mult_2} \cdots \multAlt \config_n$.
  \item Addition of a constant configuration $\config$: let $\config \oplus \trace$ be the trace $\config + \config_1 \xrightarrow{\transTuple_1: \mult_1} \config + \config_2 \xrightarrow{\transTuple_2: \mult_2} \cdots \config + \config_n$.
\end{enumerate}
It is not hard to see, by consulting the definition of a step, that $\multAlt \otimes \trace$ and $\config \oplus \trace$, are indeed traces of $\vas$.
Note, however, that multiplying by a scalar $\multAlt < 0$ would not yield a trace, and that adding a vector $\config$ that is not a configuration (i.e., which has negative coordinates) may sometimes also not yield a trace --- either because intermediate points may not be configurations (due to having some negative coordinates), or since condition (ii) in the definition of a step (Definition~\ref{dfn:step}) is violated.
Finally, traces in CVRSs have a \emph{convexity property} that states that by taking a fraction $0 < \multAlt < 1$ of each step of a trace $\config \rightarrow^\star \config'$ one obtains a trace from $\config$ to the convex combination $(1-\multAlt)\config + \multAlt\config'$:
\begin{proposition}[convexity]\label{prop:convexity}
Let $\trace := \config_1 \xrightarrow{\transTuple_1: \mult_1} \config_2 \xrightarrow{\transTuple_2: \mult_2} \cdots \config_n$
be a trace of a CVRS $\vas$, and let $0 < \multAlt < 1$ be rational.
Define configurations $\config_i' := \multAlt\config_i + (1-\multAlt)\config_1$ for $1 < i \leq n$.
Then $\trace' := \config_1 \xrightarrow{\transTuple_1:\multAlt\mult_1} \config_2' \xrightarrow{\transTuple_2:\multAlt\mult_2} \cdots \config_n'$ is a trace of $\vas$.
\end{proposition}
\begin{proof}

Simply observe that $\trace'$ is the trace $((1- \multAlt) \config_1) \oplus (\multAlt \otimes \trace)$.
\end{proof}

In the following two lemmas we combine Remark~\ref{rem: VRS for template} with Lemma~\ref{lem: edge types and pseudo cycles} in order to rephrase the characterization of edge types in terms of the existence of certain CVRSs traces.
We begin by characterizing the \locr edges and the \good edges (relative to \green edges):

\begin{lemma}
\label{lem:locr-cvrs-characterization}
An edge $e$ of $\puwd$ is \locr iff $e$ participates in a step of a cyclic trace of some component $\puwd_i$ considered as CVRS, i.e., iff $e$ participates in a step of a trace $\trace := \config \rightarrow^\star  \config$ of CVRS $\puwd_i$.
Moreover, $e$ is \good iff $\trace$ uses only \green edges.
\end{lemma}
\begin{proof}
We start with the `only if' direction.
By Lemma~\ref{lem: edge types and pseudo cycles} there is a legal pseudo-cycle $C$, which contains $e$, and which does not contain broadcasts.
Because $C$ does not contain broadcasts, we have that all configurations of $C$ are in some component $\puwd_i$ of $\puwd$.
Let $f$ be the initial configuration of $C$, and let $\config$ be its counter representation. By Remark~\ref{rem: VRS for template}, $C$ induces
a trace $\trace := \config \rightarrow^\star  \config$ in the VRS corresponding to $\puwd_i$, and $e$ participates in a step of $\trace$.
As every VRS is a CVRS the claim follows.

For the other direction, let $e$ participate in a step of some trace $\trace := \config \rightarrow^\star \config$ of the CVRS corresponding to $\puwd_i$.
Let $x$ be the least common multiple of all the denominators that appear in the multiplicities of any step, or any coordinate of any configuration of $\trace$.
Consider the trace $x \otimes \trace :=
x \config \rightarrow^\star x \config$.
Observe that all configurations on this trace are in $\NatZero^{\states_i}$, and that all steps on it are taken with an integer multiplicity.
By replacing every step that uses a multiplicity $y \in \Nat$ with $y$ consecutive steps each with multiplicity $1$, we obtain a trace $\traceAlt := x \config_i \rightarrow^\star x \config'_i$ in $\puwd_i$ considered as VRS.
By Remark~\ref{rem: VRS for template}, there is a corresponding pseudo-cycle $C$ of $\puwd_i$, and $e$ appears on $C$.
The claim then follows from Lemma~\ref{lem: edge types and pseudo cycles}.
\end{proof}

We now turn to characterizing the \green edges.

For the statement of the lemma, recall (see Section~\ref{sec:solving finite}) that the template $\puwd$ is built from component templates arranged in a lasso structure. Let $\loopindices := \{n, n+1, \ldots, n+r-1=m\}$ be the set of indices of the components on the noose of $\puwd$, and for every $i \in \loopindices$ define $\suc{i}$ (resp. $\pre{i}$) to be the component number immediately following $i$ (resp. preceding $i$) along the noose. Note that, in particular, $\suc{n+r-1} = n$ and $\pre{n} = n+r-1$.

\begin{lemma}
\label{lem:green-cvrs-characterization}
An edge $e$ of $\puwd$ is \green iff, for every $i \in \loopindices$: (i) there is a subset $\T_i$ of the transitions of $\puwd_i$, and a subset $\B_i$ of the broadcast transitions from $\puwd_i$ to $\puwd_{\suc{i}}$; (ii) there are coefficients $\mult_\transition \in \RatGZ$ for every $\transition \in \B_i$; and (iii) there is a trace $\trace_i := \config_i \rightarrow^\star \config_i'$ of the CVRS of $\puwd_i$ using exactly the transitions $\T_i$; such that: $e \in \cup_{i \in \loopindices} (\B_i \cup \T_i$), and for all $q \in S^\uwd_i$ the following holds:
\begin{enumerate}[(1)]
  \item $\config_i(q) = \sum_{\transition \in \B_{\pre{i}}, \dst(\transition) = q } \mult_\transition$.
  \item $\config'_i(q) = \sum_{\transition \in \B_i, \src(\transition) = q } \mult_\transition$.
\end{enumerate}
\end{lemma}
\begin{proof}
By Lemma~\ref{lem: edge types and pseudo cycles}, it is enough to show that $e$ appears on a pseudo-cycle $C$, with $r$ broadcasts, of $\puwdsys$ iff the conditions of the lemma hold.
First, assume that such a pseudo-cycle $C$ exists. For every $i \in I$, let $f_i$ (resp. $f'_i$) be the configurations of $C$ in $\puwd_i$ just after (resp. just before) a broadcast, let $C_i$ be the sub-path of $C$ from $f_i$ to $f'_i$, let $\T_i$ be the process transitions appearing on $C_i$.
By Remark~\ref{rem: VRS for template}, $C_i$ induces a trace $\trace_i := \config_i \rightarrow^\star \config_i'$ of the CVRS of $\puwd_i$, using exactly the transitions in $\T_i$, where $\config_i,\config'_i$ are the counter representations of $f_i,f'_i$.
Consider now the global broadcast transition of $C$ from $f'_i$ to $f_{\suc{i}}$, let $\B_i$ be the set of local broadcast transitions that participate in it, and for every $\transition \in \B_i$ let $\mult_\transition$ to be the number of processes in $f'_i$ that take $\transition$. It is easy to see that the conditions of the lemma are satisfied.

For the other direction, assume that sets $\T_i$ and $\B_i$, coefficients $\mult_\transition$, and traces $\trace_i$ satisfying the conditions of the lemma exist.
Using the same argument as in the proof of Lemma~\ref{lem:locr-cvrs-characterization} --- by taking $x$ to be the least common multiple of all the denominators that appear in the multiplicities of any step or any coordinate of any configuration of these traces, as well as any of the coefficients $\mult_\transition$ --- we can obtain from each trace $\trace_i$ a trace $\traceAlt_i := x \config_i \rightarrow^\star x \config'_i$ in the VRS corresponding to $\puwd_i$.
We can now build the required pseudo-cycle $C$ in $r$ rounds. We start by (arbitrarily) picking some $i_i \in \loopindices$, and a configuration $f_{i_1}$ whose counter representations is $x \config_{i_1}$.
At each round $j$, we extend $C$ by concatenating a path (obtained from $\traceAlt_{i_j}$ by Remark~\ref{rem: VRS for template}) from $f_{i_j}$ to a configuration $f'_{i_j}$ whose counter representation is $x \config'_{i_j}$; we then append a global broadcast transition in which each edge $t \in \B_{i_j}$ is taken by exactly $x \mult_\transition$ processes, resulting in a configuration $f_{i_{j+1}}$ whose counter representation is $x \config_{\suc{i_j}}$ (this is possible by equations (1) and (2) in the conditions of the lemma).

It follows that, at the end of the last round, $C$ is a path in $\puwdsys$ from $f_{i_1}$ to $f_{i_{r+1}}$, whose counter representations are $x \config_{i_1}$ and $x \config_{\suc[r]{i_1}} = x \config_{i_1}$, respectively. Thus, $C$ is a pseudo-cycle. 
\end{proof}

In Lemmas~\ref{lem:locr-cvrs-characterization} and~\ref{lem:green-cvrs-characterization} we have
related edge-types to the existence of certain CVRS traces.
Developing a criterion for the existence of CVRS traces is the subject of the next subsection.
Having developed this characterization, we can replace the existence of CVRS traces with this criterion in the two lemmas mentioned.
This will then allow us to present our algorithms for deciding edge-types using linear programming.

\subsection{Characterization of Reachability in CVRSs}
\label{subsec:characterization-reachability-cvas}

Our aim in this section is to develop a characterization of the existence of a trace between two configurations $\config, \config'$ in a CVRS which can be used as a basis for an efficient algorithm.
Looking at Definition~\ref{dfn:step} of a step, one can see that the task is very simple if $c'$ is reachable from $c$ in one step. However, in the general case, a trace from $\config$ to $\config'$ may have a very large number of steps, so we have to find a way to avoid reasoning about each step individually.
The following proposition is our starting point for summarizing a large number of steps by equations, and it follows immediately from the definitions of a step and a trace:

\begin{proposition}
\label{prop:equations}
Let $\vas = \tup{\Actionprts,\states,\transRel}$ be a CVRS, and let $\xi:= \config \rightarrow^\star \config'$ be a trace in it. For every $\transition \in \transRel$, let $\mult_\transition \in \RatGEZ$ be the sum of the multiplicities of all the steps in $\xi$ in which $\transition$ participates. Then:
\begin{enumerate}[(1)]
  \item $\config' = \config + \sum_{ \transition \in \transRel } \mult_\transition \left( \incoming(\transition) - \outgoing(\transition) \right)$; \label{it:1}
  \item For every rendezvous action $\msg{a} \in \Actions$, and every $i,j \in [k]$, we have that
        \[
            \sum_{\transition \in \transRel, \rdz(\transition) = \msg{a}_i } \mult_\transition = \sum_{\transition \in \transRel, \rdz(\transition) = \msg{a}_j } \mult_\transition
        \] \label{it:2}
\end{enumerate}
\end{proposition}

We will see that the inverse of Proposition~\ref{prop:equations} also holds, provided one adds some suitable conditions, as follows.
The equations of Proposition~\ref{prop:equations} ensure that the `accounting' in a trace is done correctly, i.e., that the transitions can be allocated to steps --- item~\ref{item:2}, corresponding to requirement (i) in Definition~\ref{dfn:step} --- and that added to the vector $\config$ they yield the vector $\config'$  --- item~\ref{item:1}, corresponding to requirement (iii) in that definition.
In order to extend Proposition~\ref{prop:equations} to a full characterization of reachability, we need to address the last element in the definition of a step (requirement (ii)) which states, for every state $s$, that $\config(s)$ is at least the multiplicity of the step times the number of transitions that participate in the step whose source is $s$. Observe that if $\config(s) > 0$ this condition can always be satisfied if $\mult$ is small enough; however, if $\config(s) = 0$, and there is some $i \in [k]$ such that $\src(\transition_i) = s$, then the condition is necessarily violated. The characterization we now develop will thus make a distinction between those states for which $\config(s) > 0$ and those for which $\config(s) = 0$.

We will use the following terminology:
For a vector $\config \in \RatGEZ^\states$, let the \emph{support} of $\config$ be the set $\positiveOp{\config} := \{ s \in \states \mid \config(s) > 0 \}$.
We will also make use of the following notation: for a set of transitions $\transRel$, we lift $\src$ and $\dst$ as follows: $\src(\transRel) = \{ \src(\transition) \mid \transition \in \transRel\}$ and $\dst(\transRel) = \{ \dst(\transition) \mid \transition \in \transRel\}$.

The following lemma states that we can summarize a large number of steps between two configurations in case the support of the two configurations is the same and no step increases the support:

\begin{lemma}
\label{lem:reachability-same-support}
Let $\vas = \tup{\Actionprts,\states,\transRel}$ be a CVRS.
Let $\config,\config' \in \RatGEZ^\states$ be configurations and let $\mult_\transition \in \RatGEZ$ be coefficients for every $\transition \in \transRel$ such that:
\begin{enumerate}[(1)]
  \item $\config' = \config + \sum_{ \transition \in \transRel } \mult_\transition \left( \incoming(\transition) - \outgoing(\transition) \right)$; \label{ii:1}
  \item For every rendezvous action $\msg{a} \in \Actions$, and every $i,j \in [k]$, we have that
        \[
\sum_{\transition \in \transRel, \rdz(\transition) = \msg{a}_i } \mult_\transition = \sum_{\transition \in \transRel, \rdz(\transition) = \msg{a}_j } \mult_\transition 
        \] \label{ii:2}
  \item $\positiveOp{\config} = \positiveOp{\config'}$; \label{ii:3}
  \item For $\transRel' := \{ \transition \in \transRel \mid \mult_\transition > 0 \}$ we have $\src(\transRel') \subseteq \positiveOp{\config'}$ and $\dst(\transRel') \subseteq \positiveOp{\config}$. \label{ii:4}
\end{enumerate}
Then, we have $\config \rightarrow^\star \config'$ and all transitions in $\transRel'$ participate in some step of this trace.
\end{lemma}
\begin{proof}
Equations~\ref{ii:1} and~\ref{ii:2} give us a good starting point for showing the existence of a trace from $\config$ to $\config'$.
Indeed, \ref{ii:2} allows us to group the transitions in $\transRel'$ into steps such that each transition $\transition \in \transRel'$ is taken with a combined multiplicity of $\mult_\transition$; and~\ref{ii:1} guarantees that concatenating all these steps will take us from $\config$ to $\config'$.
It remains to find a way to make sure each of these steps satisfies requirement $(ii)$ in the definition of a step (Definition~\ref{dfn:step}).
We begin, however, by ignoring this problem.
That is, we will consider \emph{quasi-steps} and \emph{quasi-traces}.
Formally, quasi-steps are like steps except that they are between arbitrary vectors in $\Rat^\states$ (as opposed to steps which were defined only for vectors in $\RatGEZ^\states$), and that they do not have to satisfy requirement $(ii)$ of Definition~\ref{dfn:step}.
We will denote quasi-steps using $\xLongrightarrow{}$.
A quasi-trace is simply a sequence of quasi-steps where the destination of each quasi-step in the sequence is the source of the next one.
We complete the proof by first constructing a quasi-trace $\traceAlt$ from $\config$ to $\config'$, and then showing how to turn $\traceAlt$ into a trace.

We construct $\traceAlt$ using the following algorithm. If $\config = \config'$ then we are done. Otherwise, at round $0$, set $\traceAlt$ to the empty quasi-trace, and $v_1 := \config$ be the source of the first quasi-step to be constructed. At round $\geq 1$, do the following:
\begin{enumerate}
  \item pick a transition $\transition$ whose $\mult_\transition$ is minimal among all edges in $\transRel'$; let $\msg{a}_i$ be $\rdz(\transition)$;
  \item let $\transition_i := t$, and for every $j \in [k] \setminus \{i\}$ pick some $\transition_j \in \transRel'$ with $\rdz(\transition_j) = \msg{a}_j$;
  \item extend $\traceAlt$ with the quasi-step $v_i \xLongrightarrow{\transition_1,\ldots,\transition_k: \mult_\transition} v_{i+1}$; 
  \item for every $j \in [k]$, subtract $\mult_\transition$ from $\mult_{\transition_j}$, and if the resulting $\mult_{\transition_j}$ is zero then remove $\transition_j$ from $\transRel'$. If $\transRel' = \emptyset$ stop and output $\traceAlt$.
\end{enumerate}

Since at least one transition (namely the transition $t$ picked in step $1$) is removed at the end of each round, the algorithm stops after at most $|\transRel'|$ rounds. It is also easy to see that item~\ref{ii:2} of the lemma is an invariant of the algorithm (i.e., it holds using the updated values at the end of each round). This, together with our choice of $t$, ensures that we are always able to find the required transitions in the second step of every round.
Finally, observe that the resulting final quasi-trace $\traceAlt$ uses every transition $\transition$ in (the initial) $\transRel'$ with a combined multiplicity which is exactly $\mult_\transition$ and thus, by item~\ref{ii:1} of the lemma, we have that $\traceAlt$ ends in $\config'$ as promised.
It remains to show how to derive from $\traceAlt$ a trace from $\config$ to $\config'$.

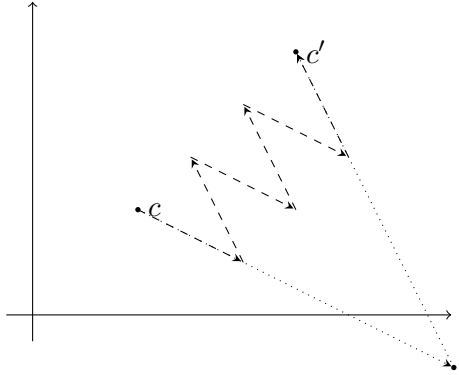
\begin{figure}[!h]
\begin{tikzpicture}[->,shorten >=1pt,auto,node distance=3cm,
    y=.7cm,
    x=.7cm,
    font=\sffamily,
    vector/.style={-stealth,dotted},
	vector guide/.style={dashed,red,thick}
    ]
	\draw[->] (-0.5,0) -- coordinate (x axis mid) (8,0);
    	\draw[->] (0,-0.5) -- coordinate (y axis mid) (0,6);

    \coordinate (P1) at (2,2); 
    \coordinate (P2) at (8,-1); 
    \coordinate (P3) at (5,5); 

    \coordinate (P11) at (4,1); 
    \coordinate (P12) at (3,3); 
    \coordinate (P13) at (5,2); 
    \coordinate (P14) at (4,4); 
    \coordinate (P15) at (6,3); 

    \fill (P1) circle (1pt);
    \fill (P2) circle (1pt);
    \fill (P3) circle (1pt);

    \draw[vector] (P1) node [right] {$\config$} -- (P2) node [left] {};
    \draw[vector] (P2) node [left] {} -- (P3) node [right] {$\config'$};

    \draw[vector, dashed] (P1) -- (P11);
    \draw[vector, dashed] (P11) -- (P12);
    \draw[vector, dashed] (P12) -- (P13);
    \draw[vector, dashed] (P13) -- (P14);
    \draw[vector, dashed] (P14) -- (P15);
    \draw[vector, dashed] (P15) -- (P3);

\end{tikzpicture}
\caption{\label{fig:2d} A graphical representation in two dimensions of the replacement of quasi-steps (in dotted arrows) by steps (in dashed arrows). In this example $h = 3$ and $m = 2$.}

\end{figure}

Let $m$ be the number of quasi-steps in $\traceAlt$. For $j \in [m]$, denote the transitions taken in the $j$-th quasi-step of $\traceAlt$ by $\transition^j_1,\ldots,\transition^j_k$, denote the multiplicity used by $\mult^j$, and the destination reached by $v_j$ (also set $v_0 := \config$). Consider $\config = v_0, v_1, \ldots, v_m = \config'$ as points in the $|\states|$-dimensional Euclidean space.
Let $\positiveSet := \positiveOp{\config} = \positiveOp{\config'}$ (by item~\ref{ii:3} of the lemma).
We note that every point on the line segment $L$ between $\config$ and $\config'$ also has support $\positiveSet$ (being a convex combination of $\config$ and $\config'$). It follows that $L$ does not touch any of the hyperplanes corresponding to the states in $\positiveSet$ (i.e., the hyperplanes defined by the equations $v(s) = 0$, for $s \in \positiveSet$). Let $x > 0$ be the minimum of the Euclidean distances between $L$ and any of these hyperplanes\footnote{Recall that in a Euclidean space the distance between a hyperplane and a line-segment, that does not touch it, is positive.},
let $y$ be the maximum over $j \in {m}$ of the multiplicities $\mult^j$, and let $h \in \Nat$ be large enough to satisfy $\frac{m y k}{h} \leq x$.

Construct a new quasi-trace $\traceAlt'$ from $\config$ to $\config'$ by starting at $\config$, and repeating $h$ times the following: for $j = 1, \dots, m$, extend $\traceAlt'$ by taking a quasi-step using the transitions $\transition^j_1,\ldots,\transition^j_k$ with multiplicity $\mult^j / h$. This construction is illustrated in Figure~\ref{fig:2d}.
We claim that $\traceAlt'$ is actually a trace, i.e., that for every $s \in \states$, and every point $v_{lm + j}$ on $\traceAlt'$ (where $0 \leq l < h$, and $0 \leq j < m$) we have that condition (ii) in Definition~\ref{dfn:step} is satisfied, namely, that $v_{lm + j}(s) \geq o_{lm + j}(s)$, where $o_{lm + j}$ is the vector $\frac{\mult^j}{h} \sum_{i=1}^{k} \outgoing(\transition^j_i)$.
Observe that this not only guarantees that each quasi-step is a step, but also that all points on $\traceAlt'$ are configurations (i.e., have no negative coordinates). Consider first the case of $s \in \states \setminus \positiveSet$. By item~\ref{ii:4} of the lemma, $\src(\transRel') \subseteq \positiveSet$ and $\dst(\transRel') \subseteq \positiveSet$, and thus $v_{lm + j}(s) = \config(s) = 0 = o_{lm + j}(s)$. For $s \in \positiveSet$, note that $o_{lm + j}(s) \leq \frac{yk}{h}$, and thus it is enough to show that $v_{lm + j}(s) - \frac{yk}{h} \geq 0$. Observe that, for every $0 \leq l < h$, the point reached after taking $lm$ quasi-steps of $\traceAlt'$ is $\config + \frac{l}{h} (\config' - \config)$, which is a point on $L$. Hence, by our choice of $x$, we have that $v_{lm}(s) \geq x$. Also note that each quasi-step on $\traceAlt'$ can decrease $v_{lm}(s)$ by at most $\frac{yk}{h}$, thus, for every $0 \leq j < m$, we have that $v_{lm + j}(s) - \frac{yk}{h} \geq  v_{lm}(s) - \frac{m y k}{h} \geq x -x = 0$. 
\end{proof}

The above lemma already allows us to characterize the \locr edges, as well as the \good edges (relative to \green edges):

\begin{lemma}
\label{lem:locr-decidability}
Given a component $\puwd_i$ of $\puwd$, an edge $e$ of $\puwd_i$ is \locr iff there are coefficients $\mult_\transition \in \RatGEZ$ for every $\transition \in R^\uwd_i$, such that:
\begin{enumerate}[(1)]
  \item $0 = \sum_{ \transition \in R^\uwd_i } \mult_\transition \left( \incoming(\transition) - \outgoing(\transition) \right)$;
  \item For every rendezvous action $\msg{a} \in \Actions$, and every $j,h \in [k]$, we have that \\
         $\sum_{\transition \in R^\uwd_i, \rdz(\transition) = \msg{a}_j } \mult_\transition = \sum_{\transition \in R^\uwd_i, \rdz(\transition) = \msg{a}_h } \mult_\transition$;
  \item $\mult_{e} > 0$.
  \end{enumerate}
\end{lemma}
Moreover, $e$ is \good iff we further require that the set $\{ \transition \in \transRel_i \mid \mult_\transition > 0 \}$ only contains \green edges.
\begin{proof}
For the `only if' direction, we observe that by Lemma~\ref{lem:locr-cvrs-characterization} there is a trace $\trace := \config \rightarrow^\star \config$ of the CVRS $\puwd_i$ such that $e$ participates in a step of $\trace$.
Moreover, $\trace$ uses only \green edges in case  $e$  is \good.
By Proposition~\ref{prop:equations}, we can derive coefficients $\mult_\transition \in \RatGEZ$ that satisfy the conditions of the lemma.

For the `if' direction, let $\transRel' = \{ \transition \in \transRel_i \mid \mult_\transition > 0 \}$ and
$\positiveSet = \src(\transRel') \cup \dst(\transRel')$.
Set $\config \in \RatGEZ^{\states_i}$ to the configuration defined by $\config(q) = 1$, if $q \in \positiveSet$, and  $\config(q) = 0$, otherwise.
Then, apply Lemma~\ref{lem:reachability-same-support} and obtain a trace $\trace := \config \rightarrow^\star \config$ of the CVRS $\puwd_i$, such that $e$ participates in a step of $\trace$.
By Lemma~\ref{lem:locr-cvrs-characterization}, the type (\locr, or \good) of $e$ follows.
\end{proof}

Lemma~\ref{lem:locr-decidability} gives rise to a \PTIME algorithm for computing the \locr and \good edges (in case we already know the \green edges):

\begin{corollary}
\label{cor:deciding-locr-good}
  Let $e$ be an edge of $\puwd$.
  We can decide in \PTIME (in the size of $\puwd$), whether $e$ is \locr, and whether $e$ is \good (assuming we already know which edges are \green).
\end{corollary}
\begin{proof}
For every edge $e$ of $\puwd$, we need to find a solution that satisfies the equations of Lemma~\ref{lem:locr-decidability}.
Such a solution can be found by linear programming over the rationals, which is in \PTIME.
\end{proof}

The rest of this section will be dedicated to addressing the problem of deciding whether an edge is \green or not. Observe that --- unlike \locr and \good edges for which a witnessing pseudo-cycle is characterized by Lemma~\ref{lem:locr-cvrs-characterization} using a cyclic trace in a single CVRS --- the characterization of the witnessing pseudo-cycle for a \green edge in terms of CVRS traces (as given by Lemma~\ref{lem:green-cvrs-characterization}) involves non-cyclic traces in multiple different CVRSs. Hence, our first step is to develop a general characterization of reachability in CVRS which extends the characterization given in Lemma~\ref{lem:reachability-same-support} to the case of traces whose source and destination configurations may have a different support. We begin with the following property of the support:

\begin{lemma}
\label{lem:max-pos}
Let $\vas = \tup{\Actionprts,\states,\transRel}$ be a CVRS and let $\config, \config_1, \config_2 \in \RatGEZ^\states$ be configurations of it. Then:
\begin{enumerate}
  \item If $\config \rightarrow^\star \config_1$ and $\config \rightarrow^\star \config_2$ then
there is a trace $\config \rightarrow^\star \config_3$ with $\positiveOp{\config_3} = \positiveOp{\config_1} \cup \positiveOp{\config_2}$;
  \item If $\config_1 \rightarrow^\star \config$ and $\config_2 \rightarrow^\star \config$
then there is a trace $\config_3 \rightarrow^\star \config$ with $\positiveOp{\config_3} = \positiveOp{\config_1} \cup \positiveOp{\config_2}$.
\end{enumerate}
\end{lemma}
\begin{proof}
Let $\config_3 = \frac{1}{2}\config_1 + \frac{1}{2}\config_2$, and observe that $\positiveOp{\config_3} = \positiveOp{\config_1} \cup \positiveOp{\config_2}$.
It remains to show the desired traces between $\config$ and $\config_3$.
For the first item, let $\trace_1 := \config \rightarrow^\star \config_1$, and $\trace_2 := \config \rightarrow^\star \config_2$. A trace from $\config$ to $\config_3$ is obtained by concatenating $\frac{1}{2} \config \oplus (\frac{1}{2} \otimes \trace_1)$ and $\frac{1}{2} \config_1 \oplus (\frac{1}{2} \otimes \trace_2)$.
For the second item, let $\trace_1 := \config_1 \rightarrow^\star \config$, and $\trace_2 := \config_2 \rightarrow^\star \config$. A trace from $\config_3$ to $\config$ is obtained by concatenating $\frac{1}{2} \config_1 \oplus (\frac{1}{2} \otimes \trace_2)$ and $\frac{1}{2} \config \oplus (\frac{1}{2} \otimes \trace_1)$.
\end{proof}

Let $\vas = \tup{\Actionprts,\states,\transRel}$ be a CVRS, let $\transRel' \subseteq \transRel$, and let $\config \in \RatGEZ^\states$ be a configuration.
We say that a set $\positiveSet \subseteq \states$ is \emph{forward (resp. backward) $\transRel'$-accessible from $\config$} if there is a trace $\config \rightarrow^\star \config'$ (resp. $\config' \rightarrow^\star \config$), with all of its steps  using only transitions from $\transRel'$, such that $\positiveOp{\config'} = \positiveSet$.
We define $\forwardOp(\config,\transRel')$ (resp. $\backwardOp(\config,\transRel')$) to be the union of all sets $\positiveSet \subseteq \states$ that are forward (resp. backward) $\transRel'$-accessible from $\config$.

\begin{remark}\label{rem:forward-max}
Observe that Lemma~\ref{lem:max-pos} implies that $\forwardOp(\config,\transRel')$ is forward $\transRel'$-accessible from $\config$, and that $\backwardOp(\config,\transRel')$ is backward $\transRel'$-accessible from $\config$. It follows that $\forwardOp(\config,\transRel')$ (resp. $\backwardOp(\config,\transRel')$) is the maximal subset of $\states$ that is forward (resp. backward) $\transRel'$-accessible from $\config$.
\end{remark}

\begin{proposition}\label{prop:computing-forward}
Given a CVRS $\vas = \tup{\Actionprts,\states,\transRel}$, a subset $\transRel' \subseteq \transRel$, and a configuration $\config$, the sets $\forwardOp(\config,\transRel')$ and $\backwardOp(\config,\transRel')$ can be computed in \PTIME.
\end{proposition}
\begin{proof}
We present a fixed point algorithm for computing $\forwardOp(\config,\transRel')$; computing $\backwardOp(\config,\transRel')$ is done in a symmetric fashion.
Construct an increasing chain $\positiveSet_0 \subsetneq \positiveSet_1 \subsetneq \cdots$ of sets $\positiveSet_i \subseteq \states$, until a larger set cannot be found, as follows.
Let $\positiveSet_0 := \positiveOp{\config}$.
For each $i$, we check if there is an action $\msg{a} \in \Actions$, and there are transitions $\transition_1,\ldots,\transition_k \in \transRel'$
with $\rdz(\transition_j) = \msg{a}_j$ for all $j \in [k]$, such that $\src(\{\transition_1,\ldots,\transition_k\}) \subseteq \positiveSet_i$ and $\dst(\{\transition_1,\ldots,\transition_k\}) \not\subseteq \positiveSet_i$.
In case there are such transitions, we set $\positiveSet_{i+1} := \positiveSet_i \cup \dst(\{\transition_1,\ldots,\transition_k\})$; otherwise, we are done and have $\positiveSet_i = \forwardOp(\config,\transRel')$.

The correctness of the algorithm is demonstrated as follows. First, to see that for every round $i$, the set $\positiveSet_i$ is contained in $\forwardOp(\config,\transRel')$, proceed by induction on $i$. Note that, by the induction hypothesis, there is a trace $\trace := \config \rightarrow^\star \config_i$ with $\positiveOp{\config_i} = \positiveSet_i$, and that this trace can be extended to a trace $\trace' := \config \rightarrow^\star \config_i
\xrightarrow{\transition_1,\ldots,\transition_k:\mult} \config_{i+1}$, with $\positiveOp{\config_{i+1}} = \dst(\{\transition_1,\ldots,\transition_k\}) \cup \positiveSet_i = \positiveSet_{i+1}$, by choosing $0 < \mult < \frac{1}{k} \cdot \min_{q \in \positiveOp{\config_i}} \{\config_i(q)$\}.
To see that the algorithm outputs $\forwardOp(\config,\transRel')$, and not a proper subset of it, let $\trace := \config \rightarrow^\star \config'$ be a trace such that $\positiveOp{\config'} = \forwardOp(\config,\transRel')$.
Assume by way of contradiction that the support of some configuration along $\trace$ is not contained in the output of the algorithm, and let $\config_{i+1}$ be the first such configuration on $\trace$. Consider the step $\config_i \xrightarrow{t_1, \cdots, t_k:\mult_i} \config_{i+1}$ and note that $\src(\{\transition_1,\ldots,\transition_k\}) \subseteq \positiveOp{\config_i}$.  Let $P_0, P_1, P_2, \cdots, P_n$ be the sequence of sets computed by the algorithm.  By minimality of $i$, there is a $j$ such that $\positiveOp{\config_i} \subseteq P_j \subseteq P_n$. Thus, $\src(\{\transition_1,\ldots,\transition_k\}) \subseteq P_n$ and $\dst(\{\transition_1,\ldots,\transition_k\}) \not\subseteq  P_n$. But this contradicts the termination condition of the algorithm. 
\end{proof}

The following simple property of the operators $\forwardOp$ and $\backwardOp$ will be useful:

\begin{proposition}
  \label{prop:support}
  Let $\vas = \tup{\Actionprts,\states,\transRel}$ be a CVRS and let $\trace := \config \rightarrow^\star \config'$ be a trace of $\vas$.
  Let $\transRel'$ be the set of transitions that participate in steps of $\trace$, and let $\mathit{Configs}$ be the set of configurations that appear in $\trace$.
  Then, $\forwardOp(\config,\transRel') = \backwardOp(\config',\transRel') = \bigcup_{\config^\circ \in \mathit{Configs}} \positiveOp{\config^\circ} = \positiveOp{\config} \cup \dst(\transRel') = \positiveOp{\config'} \cup \src(\transRel')$.
\end{proposition}
\begin{proof}
  For every configuration $\config^\circ$ that appears on $\trace$, the prefix $\config \rightarrow^\star \config^\circ$ of $\trace$, and the suffix $\config^\circ \rightarrow^\star \config'$ of $\trace$, obviously only use transitions from $\transRel'$.
  Hence, $\bigcup_{\config^\circ \in \mathit{Configs}} \positiveOp{\config^\circ} \subseteq \forwardOp(\config,\transRel')$ and $\bigcup_{\config^\circ \in \mathit{Configs}} \positiveOp{\config^\circ} \subseteq \backwardOp(\config',\transRel')$; as well as $\bigcup_{\config^\circ \in \mathit{Configs}} \positiveOp{\config^\circ} \subseteq \positiveOp{\config} \cup \dst(\transRel')$ and $\bigcup_{\config^\circ \in \mathit{Configs}} \positiveOp{\config^\circ} \subseteq \positiveOp{\config'} \cup \src(\transRel')$.
  For the other direction, observe that $\forwardOp(\config,\transRel') \subseteq \positiveOp{\config} \cup \dst(\transRel') \subseteq \bigcup_{\config^\circ \in \mathit{Configs}} \positiveOp{\config^\circ}$ and, similarly, $\backwardOp(\config',\transRel') \subseteq  \positiveOp{\config'} \cup \src(\transRel') \subseteq \bigcup_{\config^\circ \in \mathit{Configs}} \positiveOp{\config^\circ}$. 
\end{proof}

\begin{remark}\label{rem:half-prop-support}
It is worth noting that for every configuration $\config$ and set of transitions $\transRel'$, we have that $\forwardOp(\config,\transRel') \subseteq \positiveOp{\config} \cup \dst(\transRel')$ (resp. $\backwardOp(\config,\transRel') \subseteq \positiveOp{\config} \cup \src(\transRel')$); however, only in case there is a path from $\config$ (resp. to $\config$), that uses exactly the transitions in $\transRel'$, do the reverse inclusions also hold.
\end{remark}

We are now ready to state a full characterization of reachability in CVRSs:

\begin{theorem}
\label{thm:reachability-no-broadcasts}
Let $\vas = \tup{\Actionprts,\states,\transRel}$ be a CVRS.
A configuration $\config' \in \RatGEZ^\states$ is reachable from a configuration $\config \in \RatGEZ^\states$ iff there are coefficients $\mult_\transition \in \RatGEZ$ for every $\transition \in \transRel$ such that:
\begin{enumerate}[(1)]
  \item $\config' = \config + \sum_{ \transition \in \transRel } \mult_\transition \left( \incoming(\transition) - \outgoing(\transition) \right)$; \label{item:1}
  \item For every rendezvous action $\msg{a} \in \Actions$, and every $i,j \in [k]$, we have that $\sum_{\transition \in \transRel, \rdz(\transition) = \msg{a}_i } \mult_\transition\!= \sum_{\transition \in \transRel, \rdz(\transition) = \msg{a}_j } \mult_\transition$; \label{item:2}
  \item For $\transRel' := \{ \transition \in \transRel \mid \mult_\transition > 0 \}$ we have that
  $\src(\transRel') \subseteq \backwardOp(\config',\transRel')$,  $\dst(\transRel') \subseteq \forwardOp(\config,\transRel')$, and $\forwardOp(\config,\transRel') = \backwardOp(\config',\transRel')$. \label{item:3}
\end{enumerate}
\end{theorem}
\begin{proof}

For the forward direction, take a trace $\trace := \config \rightarrow^\star \config'$.
For every $\transition \in \transRel$, let $\mult_\transition$ be the sum of the multiplicities of the steps of $\trace$ in which $\transition$ participates.
By Proposition~\ref{prop:equations}, we have that condition~\ref{item:1} and~\ref{item:2} of the lemma are satisfied.
Condition~\ref{item:3} holds by applying Proposition~\ref{prop:support} to $\trace$.

For the reverse direction, assume that there are coefficients $\mult_\transition$ such that conditions~\ref{item:1}-\ref{item:3} are satisfied.
Let $\positiveSet := \forwardOp(\config,\transRel') = \backwardOp(\config',\transRel')$.
By Remark~\ref{rem:forward-max}, we can obtain traces $\trace_1 := \config \rightarrow^\star \config_1$ and $\trace_2 := \config_2 \rightarrow^\star \config'$, using only transitions from $\transRel'$, such that $\positiveOp{\config_1} = \positiveSet = \positiveOp{\config_2}$. It remains to show that $\config_2$ is reachable from $\config_1$. For every transition $\transition \in \transRel$, let $\mult^1_\transition$ (resp. $\mult^2_\transition$) be the sum of the multiplicities of the steps of $\trace_1$ (resp. $\trace_2$) in which $\transition$ participates, and let $\multAlt_\transition = \mult_\transition -  \mult^1_\transition - \mult^2_\transition$. By the convexity property (Proposition~\ref{prop:convexity}) we can assume w.l.o.g. that: ($\S$) $\mult^1_\transition$ and $\mult^2_\transition$ are small enough such that $\multAlt_\transition > 0$ for all $\transition \in \transRel'$ (simply apply the convexity property to $\trace_1$ and $\trace_2$ with a small enough $\gamma$ to obtain, if needed, replacement $\trace_1,\trace_2,\config_1, \config_2$).
By Proposition~\ref{prop:equations}, we get that $\config_1 = \config + \sum_{ \transition \in \transRel' } \mult^1_\transition \left( \incoming(\transition) - \outgoing(\transition) \right)$, and $\config' = \config_2 + \sum_{ \transition \in \transRel' } \mult^2_\transition \left( \incoming(\transition) - \outgoing(\transition) \right)$. Combining the last two equalities with condition~\ref{item:1} of the Lemma, and rearranging terms, we get:
\begin{equation}
 \config_2 = \config_1 + \sum_{ \transition \in \transRel' } \multAlt_\transition \left( \incoming(\transition) - \outgoing(\transition) \right) \tag{$\dagger$}
\end{equation}

Again, by Proposition~\ref{prop:equations}, we have for every $l \in \{1,2\}$:
\begin{equation*}
\forall \msg{a} \in \Actions,  i,j \in [k]: \sum_{\transition \in \transRel', \rdz(\transition) = \msg{a}_i }\mult^l_\transition = \sum_{\transition \in \transRel', \rdz(\transition) = \msg{a}_j } \mult^l_\transition,
\end{equation*}

Together with condition~\ref{item:2} we thus get that:

\begin{equation}
\forall \msg{a} \in \Actions,  i,j \in [k]: \sum_{\transition \in \transRel', \rdz(\transition) = \msg{a}_i } \multAlt_\transition = \sum_{\transition \in \transRel', \rdz(\transition) = \msg{a}_j } \multAlt_\transition \tag{$\ddag$}
\end{equation}

By $(\dagger)$ and $(\ddag)$, and our choice of $\config_1,\config_2$,
the requirements of  Lemma~\ref{lem:reachability-same-support} are satisfied for $\config_1,\config_2$ and coefficients $\multAlt_\transition$.
Hence, $\config_2$ is reachable from $\config_1$.
\end{proof}

\begin{remark}
The characterization in Theorem~\ref{thm:reachability-no-broadcasts} is an adaptation of the characterization of reachability in continuous vector addition systems from~\cite{journals/fuin/FracaH15};
this characterization leads to a PTIME decision procedure for reachability.
We have redeveloped this characterization result here for CVRSs for the following reason:
One can convert VRSs resp. CVRSs into VASs resp. CVASs, resulting, however, in a potentially exponential number of transitions.
That is because for every $k$-wise rendezvous action $\msg{a}$, we have to create a transition 
of the VAS resp. CVAS, for every possible combination of transitions labelled by $\msg{a}_1, \ldots, \msg{a}_k$.
The resulting number of transitions is exponential in $k$.
By redeveloping the results for CVRSs, however, we directly obtain (basically the same) characterization result for CVRSs, which leads to a PTIME decision procedure for reachability (which is used as part of Algorithm~\ref{alg:pseudo-cycle}).
\end{remark}

\subsection{An Algorithm for Finding the \green Edges}
\label{subsec:deciding-pseudo-cycle}

In this section we present a \PTIME algorithm for deciding the existence of pseudo-cycles with $r$ broadcasts in $\puwdsys$.
The algorithm will be developed based on the following characterization which is an immediate consequence of Lemma~\ref{lem:green-cvrs-characterization}, using Theorem~\ref{thm:reachability-no-broadcasts} to characterize the traces $\xi_i$ of this lemma. Recall that $\loopindices$ is the set of indices of the components on the noose of $\puwd$.

\begin{corollary}
\label{cor:reachability-with-broadcast-characterization}
An edge $e$ of $\puwd$ is \green iff, for every $i \in \loopindices$: (i) there is a subset $\T_i$ of the transitions of $\puwd_i$, and a subset $\B_i$ of the broadcast transitions from $\puwd_i$ to $\puwd_{\suc{i}}$, with $e \in \cup_{i \in \loopindices} (\B_i \cup \T_i$); (ii) there are coefficients $\mult_\transition \in \RatGZ$ for every $\transition \in \T_i \cup \B_i$; such that:
\begin{enumerate}[(1)]
  \item $\config'_i = \config_i + \sum_{ \transition \in \T_i } \mult_\transition \left( \incoming(\transition) - \outgoing(\transition) \right)$, where $\config_i, \config'_i$ are defined, for every $q \in S^\uwd_i$, by: \\
         $\config_i(q) := \sum_{\transition \in \B_{\pre{i}}, \dst(\transition) = q } \mult_\transition$, and $\config'_i(q) := \sum_{\transition \in \B_i, \src(\transition) = q } \mult_\transition$;\label{cond:1}
   \item for every rendezvous action $\msg{a} \in \Actions$, and every $j,h \in [k]$, we have that \\ 
        $\sum_{\transition \in \T_i, \rdz(\transition) = \msg{a}_j } \mult_\transition = \sum_{\transition \in \T_i, \rdz(\transition) = \msg{a}_h } \mult_\transition$; \label{cond:2}
  \item $\dst(\T_i) \subseteq \forwardOp(\config_i,\T_i)$, $\src(\T_i) \subseteq \backwardOp(\config_i',\T_i)$, and
      $\forwardOp(\config_i,\T_i) = \backwardOp(\config_i',\T_i)$. \label{cond:3}
\end{enumerate}
\end{corollary}

Corollary~\ref{cor:reachability-with-broadcast-characterization} gives rise to Algorithm~\ref{alg:pseudo-cycle} for computing the set of \green edges of $\puwd$.

We remark that the algorithm presented here for computing \green edges differs from our original algorithm in~\cite{conf/icalp/AminofRZS15}. Here,
we use an algorithm that is inspired by, and extends, Algorithm $2$ in~\cite{journals/fuin/FracaH15} for deciding the reachability of a target configuration from an initial configuration in continuous vector addition systems. The extension is in two ways: first from vector addition systems to CVRSs, and second by adding machinery for handling broadcasts.

\begin{algorithm}[htb]
\SetKw{Initialize}{Initialize}
\SetKw{Repeat}{Repeat}
\SetKw{Until}{Until}
\SetKw{Input}{Input:}
\SetKw{Output}{Output:}

\Input\ $\puwd$.\\
\ \\
\Initialize\ --- for every $i \in \loopindices$ do:\\
$\T_i$ := all rendezvous transitions of $\puwd_i$\;
$\B_i$ := all broadcast transitions from $\puwd_i$ to $\puwd_{\suc{i}}$\;
\ \\
\Repeat:\
For every $i \in \loopindices$: take variables $\config_i(q),\config'_i(q)$ for every $q \in S^\uwd_i$, and $\mult_\transition$ for every $\transition \in \T_i \cup \B_i$.\\
Find a solution to the following constraint system, such that the number of non-zero variables $\mult_\transition$ is maximal:
\begin{itemize}
  \item[$\bullet$] $\mult_\transition \geq 0$ for every $\transition \in \T_i \cup \B_i$;
  \item[$\bullet$] $\config_i(q) = \sum_{\transition \in \B_{\pre{i}}, \dst(\transition) = q } \mult_\transition$, for every $q \in S^\uwd_i$;
  \item[$\bullet$] $\config'_i(q) = \sum_{\transition \in \B_i, \src(\transition) = q } \mult_\transition$, for every $q \in S^\uwd_i$;
  \item[$\bullet$] $\config_i' = \config_i + \sum_{ \transition \in \T_i } \mult_\transition \left( \incoming(\transition) - \outgoing(\transition) \right)$;
  \item[$\bullet$] $\sum_{\transition \in \T_i, \rdz(\transition) = \msg{a}_j } \mult_\transition = \sum_{\transition \in \T_i, \rdz(\transition) = \msg{a}_h } \mult_\transition$ for every $\msg{a} \in \Actions$ and $j,h \in [k]$.
\end{itemize}
Let $\positiveSet_i := \forwardOp(\config_i,\T_i) \cap \backwardOp(\config'_i,\T_i)$\;
Let $\T_i$ := $\{ \transition \in \T_i \mid \mult_\transition > 0 \} \cap \{  \transition \in \T_i \mid \src(\transition) \in \positiveSet_i \wedge \dst(\transition) \in \positiveSet_i \}$\;
Let $\B_i$ := $\{ \transition \in \B_i \mid \mult_\transition > 0 \}$\;
\Until\ neither $\T_i$ nor $\B_i$ change, for any $i \in \loopindices$.\\
\ \\
\Output\ $\bigcup_{i \in \loopindices} \B_i \cup \T_i$

\caption{Algorithm for computing all edges of $\puwd$ that can appear in a pseudo-cycle with broadcasts of $\puwdsys$.}
\label{alg:pseudo-cycle}
\end{algorithm}

\begin{theorem}
\label{thm:computatation-of-green-transitions}
Deciding if an edge of $\puwd$ is \green can be done in \PTIME.
\end{theorem}
\begin{proof}
We begin by making a couple of important observations. In each iteration of the main loop of the algorithm we need to find a solution to a constraint system forming a linear programming problem over the rationals, with no objective function, whose canonical form is: $Ax = 0, x \ge 0$ (for some matrix $A$ and vector $x$ of variables), such that the solution is maximal with respect to the number of non-zero variables $\mult_\transition$. Given that linear programs of this canonical form have the property that the sum of any two solutions is itself a solution, we have that: ($\dagger$) one can find a maximal solution $x$ by looking, for every $\transition \in \T_i \cup \B_i$, for a solution to the system $Ax = 0, x \ge 0, \mult_\transition > 0$, and adding together all the solutions that were found; ($\ddagger$) for every variable $\mult_\transition$ of the system $Ax = 0, x \ge 0$, if there is a solution in which $\mult_\transition > 0$, then $\mult_\transition > 0$ in every maximal solution.

We now show that Algorithm~\ref{alg:pseudo-cycle} outputs exactly the set of \green edges of $\puwd$.

To see that every edge output by the algorithm is \green, we will apply Corollary~\ref{cor:reachability-with-broadcast-characterization} (direction `if') to the sets $\T_i$ and $\B_i$ and the coefficients $\mult_\transition$ from the solution obtained in the last iteration of the algorithm. Observe that at this last iteration all the sets have reached a fixed point. Hence, in particular, for every $\transition \in \T_i \cup \B_i$ we have that $\mult_\transition > 0$, and (\S) $\dst(\T_i) \subseteq \positiveSet_i$, $\src(\T_i) \subseteq \positiveSet_i$. Since the constraint system in the algorithm exactly matches requirements (1) and (2) in Corollary~\ref{cor:reachability-with-broadcast-characterization}, the only thing we have to show before we can apply this corollary is that also condition (3) holds. Observe that, by $\S$, it is enough to show that at the last iteration $\forwardOp(\config_i,\T_i) = \positiveSet_i$ and $\backwardOp(\config_i',\T_i) = \positiveSet_i$. We will show the first equality, the second is shown in a symmetric way. Consider the following chain of inequalities: $\positiveSet_i \subseteq \forwardOp(\config_i,\T_i) \subseteq \positiveOp{\config_i} \cup \dst(\T_i) \subseteq \positiveOp{\config_i} \cup \positiveSet_i \subseteq \positiveSet_i$. The first containment is by the definition of $\positiveSet_i$, the second by Remark~\ref{rem:half-prop-support}, the third by $\S$, and the last by the following argument: for $q \in \positiveOp{\config_i}$, if $q \in \src(\T_i) \cup \dst(\T_i)$ then $q \in \positiveSet_i$ by $\S$; otherwise, $q \in \positiveOp{\config'_i}$ (by the constraint $\config_i' = \config_i + \sum_{ \transition \in \T_i } \mult_\transition \left( \incoming(\transition) - \outgoing(\transition) \right)$), and since by definition $\positiveOp{\config_i} \subseteq \forwardOp(\config_i,\T_i)$ and $\positiveOp{\config'_i} \subseteq \backwardOp(\config'_i,\T_i)$, we have that $q \in \forwardOp(\config_i,\T_i) \cap \backwardOp(\config'_i,\T_i) = \positiveSet_i$.

To see that every \green edge $e$ of $\puwd$ is output by the algorithm, apply Corollary~\ref{cor:reachability-with-broadcast-characterization} (direction `only if') to $e$ to obtain sets $\tilde{\T}_i$ and $\tilde{\B}_i$, and coefficients $\tilde{\mult}_\transition > 0$, satisfying the conditions of the corollary.
We claim that if $\tilde{\T}_i \cup \tilde{\B}_i \subseteq \T_i \cup \B_i$ holds, for every $i \in \loopindices$, at the beginning of an iteration (which is the case at initialization), then it also does at its end. Note that this would conclude the proof since $e \in \cup_{i \in \loopindices} (\tilde{\B}_i \cup \tilde{\T}_i$) by condition (i) of the corollary. To prove the claim, note that if $\tilde{\T}_i \cup \tilde{\B}_i \subseteq \T_i \cup \B_i$, for every $i \in \loopindices$, then the coefficients $\tilde{\mult}_\transition$ induce a (not necessarily maximal) solution to the constraint system. Hence, by observation ($\ddagger$) at the beginning of the proof, every maximal solution $\mult_\transition, \config, \config'$ of this system will have $\mult_\transition > 0$ for every $\transition \in \tilde{\T}_i \cup \tilde{\B}_i$. The claim follows by showing that $\src(\tilde{\T}_i) \cup \dst(\tilde{\T}_i) \subseteq \positiveSet_i$. To see this, observe that $\mult_\transition > 0$ for every $\transition \in \tilde{\B}_i$ implies (by the definition of $\config_i$ and $\config_i'$ in condition (1) of the corollary, and the corresponding constraint in the algorithm) that $\positiveOp{\tilde{\config}_i} \subseteq \positiveOp{\config_i}$ and $\positiveOp{\tilde{\config}_i'} \subseteq \positiveOp{\config_i'}$. Hence, $\forwardOp(\tilde{\config}_i, \tilde{\T}_i) \subseteq \forwardOp(\config_i, \T_i)$ and $\backwardOp(\tilde{\config}_i', \tilde{\T}_i) \subseteq \backwardOp(\config_i', \T_i)$. Since, by condition (3) of the corollary, $\src(\tilde{\T}_i) \cup \dst(\tilde{\T}_i) \subseteq \forwardOp(\tilde{\config_i},\tilde{\T}_i) = \backwardOp(\tilde{\config}_i',\tilde{\T}_i)$ we are done.

It remains to show that the algorithm runs in polynomial time.
Since in every round at least one transition is removed from either $\T_i$ or $\B_i$, for some $i$, the main loop repeats linearly many times in the size of $\puwd$.
By observation ($\dagger$) at the beginning of the proof, finding a maximal solution to the constraint system of each iteration can be done by solving a linear number of linear programming problems (of polynomial size) over the rationals, which is in \PTIME. Since, by Proposition~\ref{prop:computing-forward}, calculating $\forwardOp(\config_i,\T_i)$ and $\backwardOp(\config'_i,\T_i)$ is also in \PTIME, we conclude that the whole algorithm runs in polynomial time. 
\end{proof}
Theorem~\ref{thm:computatation-of-green-transitions}, together with Corollary~\ref{cor:deciding-locr-good}, yield the promised proof to the main theorem (Theorem~\ref{thm: edge type decidability}) of this section.

\section{Conclusion}

We have established the decidability and complexity of the PMCP for safety and liveness properties of RB-systems, which are polynomially inter-reducible with discrete-timed networks.
The lower and upper complexity bounds for safety properties are tight.
We leave open the problem of whether our \EXPTIME upper-bound for liveness properties is tight.
We note that the \PSPACE lower bound for safety properties also implies a \PSPACE lower bound for liveness properties.
The \EXPTIME upper bound is established by (repeatedly) solving an exponentially sized linear program.
As linear programming is known to be \PTIME-complete, it seems unlikely that our techniques can be improved to show a \PSPACE upper bound.

A further direction for future research concerns whether our results for the discrete-time model can be lifted to the continuous-time model without a distinguished controller (note that PMCP for continuous time networks with a distinguished controller is undecidable~\cite{aj03}). For example, time-bounded invariance and time-bounded response properties (expressed as MTL formulae) hold on the discrete-time model iff they hold with the continuous-time model~\cite{hmp92}, whereas \cite{ow03} establish several results on to the use of digitization techniques for timed automata.

\section{Acknowledgments}
This work is partially supported by the Austrian Science Fund (FWF): P 32021, the Austrian National Research Network S11403-N23 (RiSE) of the Austrian Science Fund (FWF), the PRIN project RIPER (No.20203FFYLK), the project MOST (CUP D93C22000400001) under the MUR National Recovery and Resilience Plan funded by the European Union - NextGenerationEU, the FWF project AUTOSARD: ``Automated Sublinear Amortised Resource
\parbox{.8\textwidth}{ Analysis of Data Structures'' No.~P36623 and the
project VASSAL: ``Verification and Analysis for Safety
and Security of Applications in Life'' funded by the European Union under
Horizon Europe WIDERA Coordination and Support Action/Grant Agreement
No. 101160022.
}%
\quad\raisebox{-1.25\baselineskip}{{\euflag{.16\textwidth}{!}}}

\bibliographystyle{alphaurl}
\bibliography{SBC}
\end{document}